\documentclass{article}
\newcommand{\rS}{\mathrm{S}}

\newcommand{\Om}{\Omega}

\usepackage{graphicx}         
\usepackage{cancel}
\usepackage{amsmath,amssymb,amscd,amsthm,bm} 
\usepackage{mathdots} 
\usepackage{tikz-cd}
\usepackage{color}
\usepackage{tikz}
\usetikzlibrary{shapes}

\usepackage{hyperref}
\usepackage{listings}
\usepackage{mathtools}
\usepackage{maple}
\usepackage[utf8]{inputenc}
\usepackage{breqn}
\usepackage{textcomp}

\newcommand{\beq}{\begin{equation}}  
\newcommand{\eeq}{\end{equation}}  
\newcommand{\bea}{\begin{eqnarray}} 
\newcommand{\eea}{\end{eqnarray}}   
\newcommand{\bear}{\begin{array}}  
\newcommand{\eear}{\end{array}} 

\newcommand\dd{\mathrm{d}}
\newtheorem{thm}{Theorem}[section] 

\newtheorem{propn}[thm]{Proposition}

\newtheorem{lem}[thm]{Lemma}

\newtheorem{cor}[thm]{Corollary}

\newenvironment{prfr}{\trivlist \item [\hskip 
\labelsep {\bf Proof of Proposition \ref{laurentpiv} (reprise):}]\ignorespaces}{\qed \endtrivlist}

\theoremstyle{definition}
\newtheorem{definition}[thm]{Definition}
\newtheorem{exa}[thm]{Example}
\newtheorem{remark}[thm]{Remark}

\newcommand{\N}{{\mathbb N}}
\newcommand{\Q}{{\mathbb Q}}
\newcommand{\Z}{{\mathbb Z}}
\newcommand{\C}{{\mathbb C}}

\newcommand{\Pro}{{\mathbb P}}
\newcommand{\F}{{\mathbb F}}

\newcommand{\rF}{\mathrm{F}}
\newcommand{\rX}{\mathrm{X}}

\newcommand{\rd}{\mathrm{d}}
\newcommand{\ri}{\mathrm{i}}

\newcommand\al{{\alpha}}
\newcommand\be{{\beta}}

\newcommand\gam{{\gamma}}
\newcommand\eps{{\epsilon}}
\newcommand\de{{\delta}}
\newcommand\om{{\omega}}

\DeclareMathOperator{\Jac}{Jac}

\newcommand{\Res}{\mathop{\textrm{Res}}\limits}

\newcommand\si{{\sigma}}
\newcommand\lax{{\bf L}}
\newcommand\mma{{\bf M}}

\newcommand\tX{\tilde{X}}

\newcommand{\cC}{{\cal C}}
\newcommand{\cQ}{{\cal Q}}
\newcommand{\cP}{{\cal P}}

\newcommand{\cR}{{\cal R}}
\newcommand{\cV}{{\cal V}}

\newcommand{\cS}{{\cal S}}
\newcommand{\cT}{{\cal T}}

\def\lc{\left\lfloor}   
\def\rc{\right\rfloor}

\newcommand{\p}[1]{p_{#1}}
\newcommand{\q}[1]{q_{#1}}
\renewcommand{\r}[1]{r_{#1}}
\newcommand{\pb}[1]{\left\{#1\right\}}
\newcommand{\Pb}[1]{\left\{\cdot\,,#1\right\}}
\newcommand{\PB}{\left\{\cdot\,,\cdot\right\}}
\newcommand{\cD}{\mathcal D}

\title{A family of integrable maps associated with the Volterra lattice} 

\author{A.N.W. Hone
  \thanks{Work begun while on leave from School of Mathematics, Statistics \& Actuarial Science,
    University of Kent, Canterbury CT2 7NF, UK.}
  $\,$ and J.A.G. Roberts~\\ School of Mathematics \&  Statistics~\\ University of New South Wales~\\
  Sydney, NSW 2052, Australia.~\\ 
A.N.W.Hone@kent.ac.uk$\qquad$j.a.g.roberts@unsw.edu.au \\ ~\\
  P. Vanhaecke~\\ Laboratoire de Math\'ematiques et Applications\\ UMR 7348 CNRS-Universit\'e de Poitiers~\\
  86360 Chasseneuil-du-Poitou, France.~\\
Pol.Vanhaecke@univ-poitiers.fr}
 
\begin{document} 

\maketitle
 
\begin{abstract} 
Recently Gubbiotti, Joshi, Tran and Viallet classified birational maps in four dimensions admitting two invariants
(first integrals) with a particular degree structure, by considering recurrences of fourth order with a certain
symmetry.  The last three of the maps so obtained were shown to be Liouville integrable, in the sense of admitting
a non-degenerate Poisson bracket with two first integrals in involution. Here we show how the first of these three
Liouville integrable maps corresponds to genus 2 solutions of the infinite Volterra lattice, being the $g=2$ case
of a family of maps associated with the Stieltjes continued fraction expansion of a certain function on a
hyperelliptic curve of genus $g\geqslant 1$. The continued fraction method provides explicit Hankel determinant
formulae for tau functions of the solutions, together with an algebro-geometric description via a Lax
representation for each member of the family, associating it with an algebraic completely integrable system.  In
particular, in the elliptic case ($g=1$), as a byproduct we obtain Hankel determinant expressions for the solutions
of the Somos-5 recurrence, but different to those previously derived by Chang, Hu and Xin.  By applying contraction
to the Stieltjes fraction, we recover integrable maps associated with Jacobi continued fractions on hyperelliptic
curves, that one of us considered previously, as well as the Miura-type transformation between the Volterra and
Toda lattices.
\end{abstract} 

\section{Introduction} 

\setcounter{equation}{0}

In classical mechanics, the study of integrable Hamiltonian systems, given by Hamiltonian vector fields with a
sufficient number of functionally independent first integrals in involution with respect to a Poisson bracket, has
a long history that goes back to the origins of calculus. It was further enriched in the latter half of the last
century by the discovery of the method of inverse scattering for solving certain Hamiltonian partial differential
equations, which gave new perspectives and new techniques for deriving finite-dimensional integrable systems
obtained as reductions of the latter.  The case of discrete integrable systems, in the form of difference equations
or maps preserving a symplectic (or Poisson) structure and satisfying the conditions for a discrete analogue of
Liouville's theorem, soon began to attract attention \cite{bruschi, maeda, veselov}, but it is fair to say that,
despite the fact that many examples are now known, the theory of discrete integrability is much less well
developed.  For integrable birational maps in the plane, the archetypal example is provided by the QRT family of
maps \cite{qrt}, whose level sets are biquadratic curves (generically, of genus one), 
are associated with elliptic fibrations \cite{duistermaat}.  If one imposes a requirement of subexponential degree
growth (zero algebraic entropy, in the terminology of \cite{bv, viallet}), then in two dimensions the only possibilities are
maps that preserve a pencil of genus one curves (like QRT), maps that preserve a pencil of rational curves, or
completely periodic maps \cite{df}. An example of a quadratic map in the projective plane that preserves a pencil of cubic curves was studied in \cite{PS}. 
This fits in with an observation of Veselov \cite{veselov}, that for an
infinite order birational map of the plane with an algebraic invariant, the level curves can have genus at most one
(as a consequence of the Hurwitz theorem on the automorphism group of a Riemann surface).

Poisson maps in three dimensions with two first integrals, of which one is a Casimir, can be reduced to the
two-dimensional case by restricting to symplectic leaves, and the common level sets are curves, so in an
algebro-geometric setting this will typically lead to elliptic fibrations. Thus, in order to see new geometrical
features, with invariant tori of dimension greater than one, it is necessary to look to integrable maps in four
dimensions.  Building on the work \cite{jv} and \cite{gjtv1}, which was based on considering autonomous versions of
the fourth-order members of hierarchies of discrete Painlev\'e I/II equations from \cite{cj}, in \cite{gjtv2}
Gubbiotti et al.\ presented a classification of four-dimensional birational maps of recurrence type, that is
\beq\label{themap} \varphi: \qquad (w_0,w_1,w_2,w_3)\mapsto \Big(w_1,w_2,w_3 ,\rF(w_0,w_1,w_2,w_3)\Big), \eeq for a
suitable rational function $\rF$ of affine coordinates $(w_0,w_1,w_2,w_3)\in \C^4$, with $\varphi$ being invariant
under the involution $\iota:\, (w_0,w_1,w_2,w_3)\mapsto (w_3,w_2,w_1,w_0)$ and having two functionally independent
polynomial invariants, $K_1$, $K_2$ say, with specific degree patterns $(\deg_{w_0}K_j, \deg_{w_1}K_j,
\deg_{w_2}K_j, \deg_{w_3}K_j)=(1,3,3,1)$ and $(2,4,4,2)$ for $j=1,2$, respectively. The result of this
classification was six maps with parameters, labelled (P.i-vi), together with six associated maps, labelled
(Q.i-vi), dual to them in the sense of \cite{quispel}, meaning that they arise as discrete integrating factors for
linear combinations of the first integrals.  (Note that, aside from the original connection with \cite{cj}, the
letter P in this nomenclature has nothing to do with the usual labelling of continuous Painlev\'e equations.)

\subsection{The map (P.iv): an integrable map in 4D} 

From our point of view, the most interesting examples among those presented in \cite{gjtv2}
  are the maps labelled (P.iv), (P.v) and (P.vi). 
According to Table 1 in \cite{gjtv2}, these are the only ones arising from a discrete variational principle (Lagrangian), which leads 
to a non-degenerate Poisson bracket in four dimensions, such that the two first integrals $K_1$, $K_2$ are in involution, and this 
means that in the real case the Liouville tori are two-dimensional (cf.\ Fig.1).  In this paper, our main concern will be  
the case of (P.iv), which is 
the birational map given in affine form by the recurrence 
\beq\label{pivmap} 
\begin{array}{l} 
w_{n+4} w_{n+3} w_{n+2} 
+
w_{n+2} w_{n+1} w_{n} 
+
2w_{n+2}^2( w_{n+3} +w_{n+1}) \\ 
+
w_{n+2} (w_{n+3}^2+ w_{n+3} w_{n+1} +w_{n+1}^2) 
 +w_{n+2}^3 +
\nu w_{n+2} (w_{n+3} +w_{n+2}+ w_{n+1})+ 
b w_{n+2} +a =0.
\end{array} 
\eeq 
The above map depends on three essential parameters $a,b,\nu$ (compared with \cite{gjtv2}, 
by rescaling we have 
set the parameter $d=1$),  
and it can be written in the form   
(\ref{themap}), 
with  
$$\rF
= -\frac{w_0w_1w_2+w_1w_2w_3+w_1^2w_2+w_2w_3^2+2w_1w_2^2+2w_2^2w_3+w_2^3
+\nu(w_1w_2+w_2w_3+w_2^2)+bw_2+a}
{w_2w_3},  
$$  which 
is the rational function of $w_0,w_1,w_2,w_3$ obtained by 
solving for $w_4$ in the recurrence (\ref{pivmap}) with $n=0$. 
More recently,  
Gubbiotti showed how the equation (\ref{pivmap}) 
also arises from a classification of additive fourth-order difference equations, 
based on the requirement of a discrete Lagrangian structure alone \cite{gub}.


The first integral denoted $I_{\mathrm{low}}^{\mathrm{P}.\mathrm{iv}}$ in 
\cite{gjtv2} is given in affine coordinates by 
\beq\label{pivh1} 
K_1=w_1w_2\Big(w_2w_3+w_0w_1-w_0w_3
+(w_1+w_2)^2 +\nu(w_1+w_2) + b 
\Big)
+a(w_1+w_2).
\eeq
The latter has the degree pattern $(1,3,3,1)$. In particular, it is linear in $w_3$, which implies that, on each three-dimensional level 
set $K_1=k_1=\,\,$const, the map (\ref{pivmap}) reduces to a birational map in three dimensions, given 
by the recurrence 
$$
\begin{array}{rcl} 
w_{n+3}w_{n+2}w_{n+1}(w_{n+2}-w_n) 
+w_{n+2}w_{n+1}^2w_n 
 +w_{n+2}w_{n+1}(w_{n+1}+w_{n+2})^2 &&  \\ 
+\nu \,w_{n+2}w_{n+1}(w_{n+1}+w_{n+2})+b\,w_{n+2}w_{n+1}
+a\,(w_{n+1}+w_{n+2})&=&k_1. 
\end{array} 
$$  
A second functionally independent invariant for (\ref{pivmap}), with degree pattern $(2,4,4,2)$, is given by 
\beq\label{pivh2} 
\begin{array}{rcr}  
K_2 &=& w_1w_2\left(\begin{array}{c}  
w_0^2w_1 + w_3^2w_2 + w_0w_3(w_1+w_2) 
+w_0(w_2^2+2w_1^2)  +w_3(w_1^2+2w_2^2) \\ +\,3(w_0+w_3)w_1w_2 +(w_1+w_2)^3 \\ 
+\nu\,\,\Big(w_0w_3+(w_0+w_3)(w_1+w_2) +(w_1+w_2)^2\Big) 
+b \,(w_0+w_1+w_2+w_3)
\end{array}  
\right) \\ 
&& +a\,\Big(w_0w_1+w_3w_2+(w_1+w_2)^2\Big) \,\,.
\end{array} 
\eeq 
This differs slightly from the second invariant presented in \cite{gjtv2}, 
which is 
$I_{\mathrm{high}}^{\mathrm{P}.\mathrm{iv}}=K_2-\nu K_1$. 

The non-degenerate Poisson bracket between the coordinates, which was obtained  in \cite{gjtv2} by making use of %
a discrete  Lagrangian for (\ref{pivmap}),  is given by 
\beq\label{pbiv} 
\{ \, w_n,w_{n+1}\,\} =0, \, 
\{ \, w_n,w_{n+2}\,\} = \frac{1}{w_{n+1}}, \, 
\{ \, w_n,w_{n+3}\,\} = - \frac{w_n +2w_{n+1}+
2w_{n+2}+w_{n+3}+\nu}{w_{n+1}w_{n+2}},
\eeq
for all $n$. So the 4D map of the form (\ref{themap}) defined by (\ref{pivmap}) is a Poisson map, 
in the sense 
that $\{\, \varphi^*G,\varphi^*H\,\} =\varphi^*\{\,G,H\,\}$ for all functions $G,H$ 
on $\C^4$. 
Equivalently, the fact that $\varphi$ is  Poisson with respect to a non-degenerate bracket 
means that it preserves a symplectic form $\om$, such that $\varphi^*\om=\om$. 
The two functionally independent invariants given in \cite{gjtv2} are in involution with respect to the 
Poisson bracket, which 
is equivalent to the involutivity of functions (\ref{pivh1}) and (\ref{pivh2}),  that is 
$$\{\,K_1,K_2\,\}=0.$$ 
Hence the four-dimensional map defined by 
(\ref{pivmap}) is integrable in the Liouville sense.

Computing the Hamiltonian vector field for the first flow, generated by $K_1$, we find 
that  on the phase space $\C^4$ with coordinates $(w_0,w_1,w_2,w_3)$ this takes the form 
\beq\label{volterra} 
\frac{\rd w_n}{\rd t} = w_n (w_{n+1}-w_{n-1}) \qquad \mathrm{for}\quad n=1,2, 
\eeq 
while the components of the vector field for $n=0,3$ appear to be more complicated rational 
functions of these 4 coordinates and the parameters $a,b,\nu$. 
However, since (\ref{pivmap}) is a Poisson map it commutes with this flow, 
so it follows that the relation (\ref{volterra}) extends to all 
$n\in\Z$. To see this, note that  the vecor field 
$\tfrac{\rd}{\rd t}=\{ \cdot, K_1\}$ commutes with the action of 
$\varphi$, and $\varphi^* (w_n)=w_{n+1}$; hence, if (\ref{volterra}) holds for some particular $n$, 
then 
$$ 
\frac{\rd w_{n+1}}{\rd t} =\varphi^*\left(\frac{\rd w_n}{\rd t}\right) = \varphi^*\big(w_n (w_{n+1}-w_{n-1})\big) = w_{n+1} (w_{n+2}-w_{n}), 
$$ 
which is just the same equation with $n\to n+1$. 
 Thus the combined solutions of the iterated map and the flow, which are compatible with one another, generate a sequence of functions 
$\big(w_n(t)\big)_{n\in\Z}$ satisfying (\ref{volterra}), which is the Volterra lattice equation, first considered by Kac and van Moerbeke \cite{kacvm}. We will see that, in a certain sense to be made precise, these 
are genus 2 solutions of this lattice equation. 

\begin{figure}[h]\label{orbit}
  \centering
  \epsfig{file=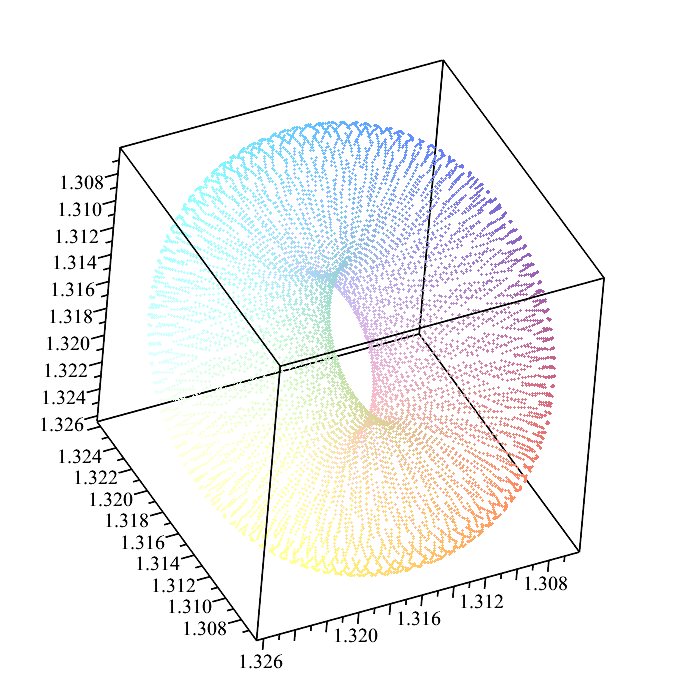, height=3.7in, width=3.7in}
  \caption{Plot of the 3D projection of 10000 points on the orbit  of  (\ref{pivmap}) with initial values
    $(\tfrac{21}{16},\tfrac{21}{16},\tfrac{452}{343},\tfrac{3124}{2373})$ and parameters $a=-9$, $b=29$, $\nu=-10$.} 
\end{figure} 

A wide variety of difference equations admitting Lax pairs and explicit formulae 
for first integrals have been presented by Svinin \cite{svinin1, svinin2}, including a family 
that arises as reductions of the hierarchy of symmetries of the Volterra lattice. By eliminating the parameter $b$ from (\ref{pivmap}) we get an equation of fifth order, that is
\beq \label{fifth}
w_{n+4}\left(\sum_{j=2}^5w_{n+j}+\nu \right) +\frac{a}{w_{n+3}}=
w_{n+1}\left(\sum_{j=0}^3w_{n+j}+\nu \right) +\frac{a}{w_{n+2}},  
\eeq 
and upon setting $a=0$  this reduces to equation (1) in \cite{svinin3} when $s=4$ (cf.\ also the case $N=4$ in \cite{hkq}, where an equivalent equation is obtained 
via a periodic reduction of the lattice KdV equation); 
thus the map (\ref{pivmap}) can be viewed as a 1-parameter generalization of one of Svinin's symmetry reductions of the Volterra lattice hierarchy, 
which in turn is a generalization of one of the maps considered in \cite{dtkq}. 

\subsection{Outline of the paper} 

The purpose of this article is to give a complete description of the complex geometry of the solutions of the map defined by (\ref{pivmap}). In reaching this goal, we found that all of the structures we obtained could naturally be extended to analogous constructions 
associated with a family of curves of arbitrary genus $g$ (elliptic for $g=1$, hyperelliptic for 
$g\geqslant 2$).

Section 2 presents a set of empirical observations, numerical examples and standalone results about the (P.iv) map. 
Originally, these were the specific clues that  led us to uncover the geometrical structure 
of the solutions of (\ref{pivmap}).  To begin with, 
we use a $p$-adic method to identify the singularity pattern of the solutions, 
leading us to introduce a tau function $\tau_n$, which lifts (\ref{pivmap}) to a recurrence of order 7 with the Laurent property;  
this is a Laurentification of the original map, in the sense of \cite{hhkq}. By considering a 
pattern of  initial values that approaches a singularity, and substituting this set of initial data into the expressions 
(\ref{pivh1}) and (\ref{pivh2}) on the level set $K_j=k_j$ for $j=1,2$, 
in  the limit that the singularity is reached  we find a hyperelliptic curve of genus 2, isomorphic to the Weierstrass quintic
\beq\label{wg2} 
y^2=(1+\nu x +bx^2)^2+4a(1+\nu x +bx^2)x^3+4k_1x^4+4(k_2+\nu k_1)x^5.
\eeq 
We also show that the tau function $\tau_n$ satisfies a Somos-9 recurrence with coefficients that depend on $a,b,\nu$ and the values of $K_1,K_2$ along each orbit of (\ref{pivmap}). 
It turns out 
that both the singularity pattern, and the corresponding tau function substitution 
$ w_n ={\tau_n\tau_{n+3}}/({\tau_{n+1}\tau_{n+2}}) $ 
found for (\ref{pivmap}), are the same as for the QRT map associated with the Somos-5 recurrence \cite{hones5}, 
which is associated with a family of elliptic curves; so this 
was a strong initial hint that analogues of the (P.iv) map should exist for any genus $g$. 
For enthusiasts of detective stories, the results in this section provide motivation and insight into how we made the first steps on the trail that led to the rest of the paper. 
However, a reader who is not particularly fond of experimental mathematics  can safely omit  Section 2  on first reading, since the subsequent sections  are not logically dependent on it, and are written in a more linear, deductive style.

In Section 3, we start from a hyperelliptic curve $\Gamma_f$ of arbitrary genus $g\geqslant 1$, 
given by a Weierstrass equation $y^2=f(x)$ where $f\in\C[x]$ is of
odd degree $2g+1$, analogous to (\ref{wg2}), 
together with a particular choice of rational function $F_0$ on the curve, and show how a
Stieltjes continued fraction (S-fraction) expansion of this function, 
of the form 
\beq\label{sfracorig} 
F_0=1-\frac{w_1x}{F_1} =1-\cfrac{w_1x}{1-\cfrac{w_2x}{F_2}} 
= 1-\cfrac{w_1x}{ 1 -\cfrac {w_2x}{ 1-\cfrac{w_3x}{1- \cdots} } } \;,
\eeq 
 leads to a birational map on the coefficients $w_j$ of the
fraction, in dimension $3g+1$, which we refer to as the \textit{Volterra map} ${\cal V}_g$.  
As we shall see,  iterating ${\cal V}_g$ for generic initial data produces the infinite sequence of 
coefficients $w_n$ for $n\geqslant 1$ that appear in the fraction \eqref{sfracorig}, while applying the inverse map 
${\cal V}_g^{-1}$ extends this sequence to $n\leqslant 0$. 
Furthermore, the recursion for the S-fraction
can be rewritten in the form of a discrete Lax equation. 
In this setting, the
hyperelliptic curve $\Gamma_f$ is the spectral curve, and the polynomial $f$ has $2g+1$ non-trivial coefficients which 
provide  conserved quantities
(first integrals) for the map. 
In particular, for $g=1$ it reduces to a QRT map whose tau functions satisfy the Somos-5
recurrence, while when $g=2$ we find that, by fixing the values of three of the first integrals  
to reduce it to four
dimensions, the map is precisely (\ref{pivmap}).  We also show from the S-fraction that, for any $g$, the solutions of the map can be
written explicitly in terms of tau functions that (up to gauge transformations) are expressed as Hankel determinants.

Next, in Section 4, we introduce a family of compatible Poisson brackets for the map ${\cal V}_g$: it is a Poisson
map with respect to any of these brackets, and the conserved quantities provide a sufficient number of invariants
in involution, so we have a Liouville integrable map for any positive integer $g$. The maps 
${\cal V}_g$ are examples of discrete
a.c.i.\ systems, which we define as follows:

\begin{definition}\label{def:adi}
  Suppose that $\C^n$ is equipped with a rational Poisson structure of rank $2r$. A birational map
  $\varphi:\C^n\to\C^n$ is said to be a \emph{discrete a.c.i.\ system} if it is a Poisson map having $s=n-r$
  functionally independent invariants $F_1,\dots,F_s$ that are pairwise in involution (so that the map is Liouville integrable)
  and such that
  \begin{enumerate}
    \item[(1)] The generic fiber of the momentum map $\mu:=(F_1,\dots,F_s):\C^n\to\C^s$ are affine parts of
      Abelian varieties ($r$-dimensional complex algebraic tori);
    \item[(2)] The restriction of $\varphi$ to the generic fiber is a translation.
  \end{enumerate}
\end{definition}

This is the natural discrete analogue of the concept of an algebraic completely integrable (a.c.i.)\ system, on which
there is a considerable literature \cite{amv,vanhaecke}, and that we will also discuss in Section
\ref{par:integrability} since a.c.i.\ systems and discrete a.c.i.\ systems are intimately connected. Although
Definition \ref{def:adi} is new, several examples of discrete a.c.i.\ systems in the particular case of $r=1$ have
already been discovered and thoroughly analyzed in the literature, often in connection with QRT maps (see
\cite{duistermaat}). The case $r=1$ is a very particular case, as the Abelian varieties are in this case
one-dimensional, that is they are elliptic curves. In the present paper, the Abelian varieties that appear as
compactified level sets of the invariants are affine parts of Abelian varieties of dimension~$g$, namely the fiber is
the Jacobian $\Jac(\bar{\Gamma}_f)$ of the (completion of the) corresponding spectral curve $\bar{\Gamma}_f$, and
the restriction of the map to any of these complex tori is indeed given by translation over a fixed vector which we
describe explicitly.  In particular, for the map (P.iv) given by (\ref{pivmap}), each generic level set defined by
fixing $K_j=k_j$ for $j=1,2$ is an affine part of an Abelian surface -- that is the Jacobian of the curve (\ref{wg2}) --
and the map restricts to a translation on each of these Abelian surfaces. It is worth clarifying that the requirement 
for a discrete system  to be a.c.i.\ 
is much more restrictive than just being integrable (in the Liouville sense): for example, the invariants of a generic linear map 
are transcendental (see section III in \cite{honesenthil}), so such maps cannot be a.c.i.\ except in certain special cases (and see also 
\cite{honep} for some examples of  
integrable Poisson maps in 3D with transcendental invariants).

The whole basis of our construction is the S-fraction expansion 
(\ref{sfracorig}), which may appear to be a deus ex machina in Section 3, 
but in fact 
has many antecedents in the literature on integrable systems, and especially 
in the development of van der Poorten's results on 
Jacobi fraction (J-fraction) expansions in elliptic~\cite{vdp1} and hyperelliptic function fields \cite{vdphyp, vdp2,vdp3}, 
as presented 
in recent work by 
one of us \cite{contfrac}.  
The latter revealed the integrable structure of maps generated by J-fractions of the form 
\beq\label{jfracvdp} 
Y_0 = \alpha_0(X) + \frac{1}{Y_1} = \al_0(X) +\cfrac{1}{\al_1(X)+\cfrac{1}{Y_2}}  
= \al_0(X)+\cfrac{1}{ \al_1(X) +\cfrac {1}{ \al_2(X)+\cfrac{1}{ \ddots} } } \;,
\eeq   
where $Y_0=\big(Y+P_0(X)\big)/Q_0(X)$ is a rational function on a hyperelliptic curve ${\cal C}$ defined 
by a polynomial of even degree $2g+2$, that is ${\cal C}:\,Y^2=P_0^2+Q_0Q_{-1}$, for polynomials $P_j,Q_j$   
of degrees $g+1,g$ in $X$, respectively, with the coefficents $\al_j=\al_j(X)$ in (\ref{jfracvdp}) 
being linear in $X$. It was shown in \cite{contfrac} that, for appropriate 
such $Y_0$, the shift from one line of the J-fraction to the next defines a Liouville integrable map 
on a phase space of dimension $3g+1$, which (on a generic level set of the first integrals) 
corresponds to a fixed translation on the Jacobian of (the completion of) ${\cal C}$.

There are classical results going back to Abel on the continued fraction expansion of the square root of  
an even degree polynomial (i.e.\ 
the function $Y$ on an even hyperelliptic curve $\cC$), 
although  the fact that the sequence of degrees of the coefficients $\al_j(X)$ in such an expansion 
is eventually periodic was proved only very recently \cite{zannier} (they need not all be linear in $X$, as per the above assumption about 
the function $Y_0$ in (\ref{jfracvdp})).   This is intimately related to elliptic \cite{akhiezer} and hyperelliptic analogues of orthogonal 
polynomials \cite{apt, chenits},  as well as more general  types of Pad\'e approximation problems connected with integrable systems \cite{bertola, ds}. 
In fact,  Stieltjes  continued fractions (of finite type)  were already used  in the solution of the finite Volterra lattice by Moser \cite{moser}, 
and such fractions were applied to obtain Hankel determinant solutions for non-isospectral extensions more 
recently \cite{chen}. 

Section \ref{sec:continuous} of the paper starts by considering 
the continuous Hamiltonian system that shares the same phase space 
with the Volterra map ${\cal V}_g$. After proving that this continuous system is a.c.i., we show that iteration of ${\cal V}_g$ 
(and its inverse) leads to an infinite sequence 
of meromorphic functions $\big(w_n(t)\big)_{n\in\Z}$ of $t$,  the time associated with one of the commuting Hamiltonian flows, 
providing a meromorphic solution of the Volterra lattice equation \eqref{volterra} (hence reproducing the above observation about (P.iv)  in the  
particular case $g=2$). Next we show that this also produces a meromorphic solution 
of the Toda lattice, taken in the form 
\begin{equation}\label{eq:TodaL}
\frac{\rd d_n}{\rd t}  =  d_n(v_{n-1} - v_{n})\;,  \qquad  
  \frac{\rd v_n}{\rd t}  =  d_n - d_{n+1}\;, \qquad n\in\Z, 
\end{equation}
by applying the well-known Miura transformation between the Volterra and Toda lattices. 
We further show that the latter transformation  arises naturally via the contraction procedure for continued fractions, 
due to Stieltjes \cite{stieltjes}, which combines successive pairs of lines in an S-fraction into a single line in a 
J-fraction, and thereby maps a generic solution of the Volterra map   ${\cal V}_g$ to an associated 
solution of the map generated by van der Poorten's construction in genus $g$. 
%
The paper ends in Section \ref{sec:conc}, with a few conclusions and observations concerning transformations relating solutions of (P.v) and (P.vi) to 
solutions of the map (P.iv), which we plan to  discuss in detail elsewhere. Also, in Appendix A (Section \ref{par:even_mumford}), 
a birational Poisson isomorphism is established between the genus $g$ even Mumford 
system (see \cite{vanhaecke}) and the Hamiltonian system associated with the  Volterra map  ${\cal V}_g$, and in Appendix B we provide the details of MAPLE code used to 
carry out the proof of Proposition \ref{bils9} using computer algebra.

\section{Laurentification and tau functions for the map (P.iv)}\label{sec:singularity} 
\setcounter{equation}{0}

In this section we exhibit certain phenomena displayed by the iterates of the map (P.iv), which are related to its discrete
integrability. Firstly, we describe the singularity pattern of the iterates, which is found from an empirical $p$-adic approach, and leads to the introduction of a
sequence of tau functions $\tau_n$ for these iterates. On the one hand, these tau functions satisfy a homogeneous 
recursion relation of order 7 with the Laurent property; so this is a Laurentification of (P.iv), as we state here and prove in Section \ref{par:hankel}. 
On the other hand, these tau functions are also shown to satisfy a Somos-9 relation, with invariants of (P.iv) as
coefficients. Secondly, by considering the limit where a solution of (P.iv) approaches 
a singularity, we are  led to a family of genus two curves which
will turn out to be at the core of the Stieltjes continued fractions (Section \ref{sec:volterra}) and the algebraic
integrability of (P.iv) (Section~\ref{sec:adi}).

The Laurent property is a very special feature of certain birational transformations, appearing in cluster algebras  and their generalizations \cite{fz1, lp}, 
which a priori is unrelated to integrability \cite{conf}. However, it turns out that the solutions of discrete integrable systems are often  encoded 
by tau functions satisfying relations that have the Laurent property, such as  bilinear equations of discrete Hirota type \cite{mase}.  %
Despite the fact that integrable maps occurring ``in the wild'' typically do not exhibit the Laurent phenomenon, it nevertheless seems to be a common feature of such maps 
that they  admit 
Laurentification, that is, a lift to a higher-dimensional 
relation that does have the Laurent property \cite{hhkq}. 
For some time,  
singularity analysis has been used as a tool to detect integrability of maps (see \cite{sing} and references),  and when the pattern of 
places where the solutions have a zero or pole is sufficiently simple, this can further suggest an appropriate way to introduce tau functions and perform Laurentification \cite{hkq}.   


To start with,  we apply the $p$-adic approach described in \cite{hkq} (see also \cite{kanki}) to the map (P.iv) defined  by (\ref{pivmap}),  and
derive a singularity pattern from it. This empirical approach is based on examining the prime factorization of the
terms of specific orbits $(w_n)_{n\in\N}$ defined in $\Q$, chosen arbitrarily, and considering the behaviour of the
$p$-adic norms $|w_n|_p$ for particular primes $p$.  As a concrete example, 
upon picking the specific parameter values $\nu=3$,
$a=5$, $b=7$ and setting all four initial values to be 1, we then find a sequence of rational numbers given by 
\begin{equation}\label{rational_orbit} 
  -30,\tfrac{743}{30},\tfrac{10541}{22290}, \tfrac{3819540}{7831963},
 -\tfrac{4315187227}{1342059038}, \tfrac{6624290612327}{739436079902}, -\tfrac{23965197528782842}{3649794341246183},
 -\tfrac{304709076970269230792}{118290200741883010693}, \ldots,
\end{equation}
where the latter terms factorize as
\begin{align*}%
 -2\cdot 3\cdot 5,\  \tfrac{743}{2\cdot 3\cdot 5},\  \tfrac{83\cdot127}{2\cdot
  3\cdot 5\cdot 743},\  \tfrac{2^2\cdot 3 \cdot 5 \cdot 63659}{83\cdot 127\cdot 743},\  -\tfrac{13\cdot 743\cdot
  446753}{2\cdot 83\cdot 127\cdot63659},\  \tfrac{19\cdot 83\cdot 127\cdot 1579\cdot 20947}{2\cdot 13\cdot
   63659\cdot 446753},\  \\
 -\tfrac{2\cdot 59\cdot 51593\cdot 61837\cdot 63659} {13\cdot 19\cdot
  1579\cdot20947\cdot446753},\  -\tfrac{2^3\cdot 13\cdot 967\cdot 446753\cdot 6782004923} {19\cdot 59\cdot
   1579\cdot 20947\cdot 51593\cdot 61837},\  \ldots .
\end{align*}
For several different primes, e.g.\ $p=3$, $5$, $83$, $127$, $743, \ldots$, this reveals a common pattern whereby,
for some~$n$,
\begin{equation}\label{eq:p-adic_norms}
  |w_{n}|_p=p^{-1}\;,\quad |w_{n+1}|_p=p\;,\quad |w_{n+2}|_p=p\;,\quad |w_{n+3}|_p=p^{-1}\;,
\end{equation}
with the prime $p$ being absent from the factorization on the previous and on the next terms: 
$|w_{n-1}|_p=|w_{n+4}|_p=1$. The $p$-adic norms (\ref{eq:p-adic_norms}) identify places where the orbit of the map over the finite field $\F_p$ has a zero or pole, 
as well as the order of these \cite{kanki}. 
Since the recurrence (\ref{pivmap}) defines a birational map, any orbit defined for 
$n\geqslant 0$ can be extended to $n<0$ (at least, provided that it does not reach a singularity, where $w_n=0$ for some $n$; but see Corollary \ref{generics} below). 
Hence, the pattern~\eqref{eq:p-adic_norms} suggests that  for $n\in\Z$ one should make the tau function 
substitution
\begin{equation}\label{pivtau}
  w_n=\frac{\tau_{n}\tau_{n+3}}{\tau_{n+1}\tau_{n+2}}\;,
\end{equation}
so that the places where a prime factor $p$ appears in the numerators or denominators of the sequence
$(w_n)_{n\in\Z}$ can be encoded 
by the appearance of the factor $p$ in the terms of the tau function sequence
$(\tau_n)_{n\in\Z}$.  These tau functions can be defined recursively, in two quite
different ways:
\begin{propn}\label{bils9}  
Suppose that $(w_n)_{n\in\Z}$ is a solution of (\ref{pivmap}). Then the corresponding sequence $(\tau_n)_{n\in\Z}$ 
satisfies
\begin{enumerate}
\item[(1)] A homogeneous recurrence of order 7 and degree 8:
  \beq\label{taurec}
\begin{array}{ccl}
\tau_{n+7}\tau_{n+4}^2\tau_{n+3}^3\tau_{n+2}^2+\tau_{n+6}^2\tau_{n+3}^4\tau_{n+2}^2 
+2\tau_{n+6}\tau_{n+5}^2\tau_{n+3}^2\tau_{n+2}^3 
+\tau_{n+6}\tau_{n+5}\tau_{n+4}^2\tau_{n+3}^2\tau_{n+2}\tau_{n+1} 
&& \\ 
+\, \tau_{n+5}^4\tau_{n+2}^4 
+2\tau_{n+5}^3\tau_{n+4}^2\tau_{n+2}^2\tau_{n+1} 
+\tau_{n+5}^2\tau_{n+4}^4\tau_{n+1}^2 
+\tau_{n+5}^2\tau_{n+4}^3\tau_{n+3}^2\tau_{n} 
&& \\
+\, \nu (\tau_{n+6}\tau_{n+5}\tau_{n+4}\tau_{n+3}^3\tau_{n+2}^2
+\tau_{n+5}^3\tau_{n+4}\tau_{n+3}\tau_{n+2}^3    
+\tau_{n+5}^2\tau_{n+4}^3\tau_{n+3}\tau_{n+2}\tau_{n+1} ) 
&& \\ 
+\,b\,  \tau_{n+5}^2\tau_{n+4}^2\tau_{n+3}^2\tau_{n+2}^2
+a\, \tau_{n+5}\tau_{n+4}^3\tau_{n+3}^3\tau_{n+2}= 0\;.
\end{array} 
\eeq  
\item[(2)] A (generalized) Somos-9 recurrence:
\begin{equation}\label{s9iv}
  \al_1\, \tau_{n+9}\tau_{n} +\al_2\, \tau_{n+8}\tau_{n+1} +\al_3\, \tau_{n+7}\tau_{n+2}
  +\al_4\, \tau_{n+6}\tau_{n+3} +\al_5\, \tau_{n+5}\tau_{n+4}=0\;. 
\end{equation} 
The coefficients $\al_i$ are polynomial functions of the parameters of the map and of the values $k_1,k_2$ of the
invariants $K_1, K_2$, hence are constant along each orbit $(w_n)_{n\in\Z}$. They are given by
$$
\begin{array}{c} 
\al_1=k_1, \qquad
\al_2=ak_2-k_1^2, \qquad 
\al_3 = a\Big(ak_2-2k_1^2\Big), \\  
\al_4 = a\Big(k_2^2+\nu k_1k_2+bk_1^2+a^2k_1\Big), \qquad
\al_5 = -k_1\Big(k_2^2+\nu k_1k_2+bk_1^2+a^2k_1\Big)\;.
\end{array} 
$$ 
\end{enumerate}
\end{propn} 
\begin{proof}
The proof of \emph{(1)} is by a direct substitution of \eqref{pivtau} in \eqref{pivmap}, while \emph{(2)} is
derived by implementing the method from \cite{hones6}, which involves computing determinants of matrices
with entries that are of degree two in tau functions. 
More precisely, for \emph{(2)} one should write five copies of the Somos-9 relation as a 
matrix equation $\mathrm{M}_n\,{\boldsymbol{\alpha}}=\mathbf{0}$, that is 
$$ 
\left( \begin{array}{ccccc} 
\tau_{n+9}\tau_n & \tau_{n+8}\tau_{n+1} & \tau_{n+7}\tau_{n+2} & \tau_{n+6}\tau_{n+3} & \tau_{n+5}\tau_{n+4} \\ 
\tau_{n+10}\tau_{n+1} & \tau_{n+9}\tau_{n+2} & \tau_{n+8}\tau_{n+3} & \tau_{n+7}\tau_{n+4} & \tau_{n+6}\tau_{n+5} \\ 
\tau_{n+11}\tau_{n+2} & \tau_{n+10}\tau_{n+3} & \tau_{n+9}\tau_{n+4} & \tau_{n+8}\tau_{n+5} & \tau_{n+7}\tau_{n+6} \\ 
\tau_{n+12}\tau_{n+3} & \tau_{n+11}\tau_{n+4} & \tau_{n+10}\tau_{n+5} & \tau_{n+9}\tau_{n+6} & \tau_{n+8}\tau_{n+7} \\ 
\tau_{n+13}\tau_{n+4} & \tau_{n+12}\tau_{n+5} & \tau_{n+11}\tau_{n+6} & \tau_{n+10}\tau_{n+7} & \tau_{n+9}\tau_{n+8} 
\end{array} 
\right) \, 
\left( \begin{array}{c} \al_1 \\ \al_2 \\ \al_3 \\ \al_4 \\ \al_5
\end{array} 
\right) = \left( \begin{array}{c} 0 \\ 0 \\ 0 \\ 0 \\ 0
\end{array} 
\right)
$$ 
and then verify that there is a non-zero vector 
of coefficients ${\boldsymbol{\alpha}}$, lying in the kernel of $\mathrm{M}_n$, that is independent of $n$. By iterating \emph{(1)}, all of the entries of the matrix $\mathrm{M}_n$ 
can be written in terms of a fixed set of 7 initial tau functions; so say for $n=0$ one can take $\tau_j$ for $0\leq j\leq 6$ as initial data, verify directly that $\det ( \mathrm{M}_0)=0$, then find a vector ${\boldsymbol{\alpha}}$ in the kernel of this matrix, and check that it is invariant under a shift of index $j\to j+1$ applied to all of these initial $\tau_j$, hence each of the 
components $\al_1,\ldots,\al_5$ is constant along an orbit.   
The latter calculations all require extensive use of computer
algebra (see Appendix B). Note that, in contrast to \emph{(1)}, there is not a strict equivalence between solutions of (\ref{pivmap}) and solutions of the Somos-9 relation \emph{(2)}, because one cannot choose 9 initial data arbitrarily; rather, the 
latter relation is only satisfied for particular sequences $(\tau_n)$, specified by 7 initial tau functions, with the coefficients $\al_j$ being fixed by these and the parameters $a,b,\nu$. 
\end{proof}
The recursion defined by (\ref{pivmap}) requires 4 initial values, while (\ref{taurec}) requires 7, and the
discrepancy between the two is described by the three-parameter group $(\C^*)^3$ of gauge transformations, with
action given by
\begin{equation}\label{gauge}
  \tau_n\mapsto A_\pm \, B^n\, \tau_n\;, \qquad A_+,A_-,B\in\C^*\;,
\end{equation}
corresponding to the freedom to rescale even/odd terms by a different factor $A_\pm$, and apply a rescaling $B^n$
that is exponential in $n$ to all terms. Any 4 non-zero initial values $w_0,\dots,w_3$ of (\ref{pivmap}) allow a
corresponding set of non-zero initial data to be determined for (\ref{taurec}), up to this gauge freedom; for example, we may
 take as corresponding initial data $\tau_n=1$ for $n=0,1,2$ and $\tau_3=w_0$, $\tau_4=w_0w_1$,
$\tau_5=w_0^2w_1w_2$ and $\tau_6=w_0^2w_1^2w_2w_3$, which are polynomials (in fact, monomials) 
in $w_0,w_1,w_2,w_3$. Notice that \eqref{taurec} can also be solved rationally for
$\tau_{n}$ in terms of $\tau_{n+1},\dots,\tau_{n+7}$ so that the sequence 
$(\tau_n)$ is actually defined for all
$n\in\Z$.

\begin{exa}\label{experitau} In the case $\nu=3$, $a=5$, $b=7$, taking the initial values $\tau_j=1$ for $0\leqslant j\leqslant 6$ in (\ref{taurec}) yields the sequence of tau functions beginning 
\beq\label{rattauseq} 
1,1,1,1,1,1,1,-30,-743,10541,127318, 5807789,628430947,-188231024119,52465590084328, \ldots, 
\eeq 
which consists of integers, and the terms after the initial 1s have prime factorizations 
given by 
$-2\cdot3\cdot 5,
-743,
83\cdot 127, 
2\cdot 63659,
13\cdot 446753,
19\cdot 1579\cdot 20947,
-59\cdot61837\cdot51593,
2^3\cdot967\cdot6782004923$, etc. These correspond to the prime factors appearing in the numerators and denominators in the particular sequence of rational values of $w_n$ illustrated in (\ref{rational_orbit}) above. Due to a reversing symmetry of the recurrence (\ref{taurec}), this sequence extends backwards to $n<0$ in such a way that 
the property 
$\tau_{6-n}=\tau_n$ holds for all $n\in\Z$, since the 7 initial data have this symmetry. 
Furthermore, these tau functions also satisfy the Somos-9 recurrence 
\beq\label{s9ex} 
 28\, \tau_{n+9}\tau_{n} -239\, \tau_{n+8}\tau_{n+1} -5115\, \tau_{n+7}\tau_{n+2}
  +136125\, \tau_{n+6}\tau_{n+3} -762300\, \tau_{n+5}\tau_{n+4}=0,
\eeq  
which corresponds to (\ref{s9iv}) with the  coefficients $\al_j$ being fixed (up to overall rescaling) by the specified choices of $\nu,a,b$, together with the fact that the first integrals take the  values 
$K_1=28$, $K_2=109$. For any solution of  (\ref{taurec}), the subsequences consisting of even/odd index terms, that is $\hat{\tau}_n=\tau_{2n}$ or $\tau_{2n+1}$, respectively, also satisfy a Somos-8 relation, of the form 
\begin{equation}\label{s8iv}
 \hat{\al}_1\,  \hat{\tau}_{n+8} \hat{\tau}_{n} +\hat{\al}_2\,  \hat{\tau}_{n+7} \hat{\tau}_{n+1} + \hat{\al}_3\,  \hat{\tau}_{n+6} \hat{\tau}_{n+2}
  +  \hat{\al}_4\,   \hat{\tau}_{n+5}  \hat{\tau}_{n+3} +  \hat{\al}_5\,   \hat{\tau}_{n+4}^2=0\;. 
\end{equation} 
For the particular integer sequence above,  up to overall scale 
the 
coefficients are given by
$$
\begin{array}{c} 
 \hat{\al}_1=195848, \qquad
 \hat{\al}_2=-61660241775, \qquad 
  \hat{\al}_3 = 13236763233189375, \\  
  \hat{\al}_4 = -8064076031989579800, \qquad
  \hat{\al}_5 = -3603810041796109733\;.
\end{array} 
$$ 
The relation (\ref{s8iv}) can be regarded as an ordinary difference reduction of a constraint 
for a tau function defined on a multidimensional lattice, which arises from a  
Hermite-Pad\'e approximation problem (cf.\ equation (2.10) in \cite{ds}). 
An explanation for why this Somos-8 relation must hold will be provided in Section \ref{sec:continuous}, via the connection with the Toda lattice and the results in \cite{contfrac}.
\end{exa} 

\smallskip

The following proposition shows that the recurrence \eqref{taurec} is a Laurentification of (P.iv). In particular, this explains why the tau functions in 
the preceding example are all integers.

\begin{propn}\label{laurentpiv} 
  The recurrence (\ref{taurec}) has the Laurent property. More precisely, for all $n\in\Z$,
  \begin{equation*}
    \tau_n \in \Z[a,b,\nu,\tau_0,\tau_1,\tau_2^{\pm 1},\tau_3^{\pm 1},\tau_4^{\pm 1},\tau_5,\tau_6]\;.
  \end{equation*} 
\end{propn} 

In principle, for $n\geqslant0$ the proof of the proposition is  a direct application of Theorem 2 in
\cite{hhkq}, and the same method of proof applies for $n<0$, because the  map defined by (\ref{taurec}) is birational. 
However, the formal verification to be done by this method quickly  gets out of hand, since the
rational functions obtained as the formulae for the first few iterations of the map defined by \eqref{taurec} soon become
too complicated for the checks to be carried out by a simple computer algebra program. 
A general proof of Proposition~\ref{laurentpiv} that does not require any computer algebra will
be given in Section \ref{par:hankel}. 

\begin{remark} As mentioned in the introduction, through the relation (\ref{fifth}), the special case $a=0$ of the (P.iv) map 
is closely linked to a difference equation appearing in the work of Svinin, and to a reduction 
of the lattice KdV equation.
 Thus, when the parameter $a=0$, the result of Proposition \ref{laurentpiv} is subtly related to the fact that  the order 7, degree 6 recurrence 
\beq\label{alberec} 
\begin{array}{ccl}
\tau_{n+7}\tau_{n+4}^2 \tau_{n+3}\tau_{n+2}\tau_{n+1} 
+\tau_{n+6}\tau_{n+5} \tau_{n+4}\tau_{n+3}^2\tau_{n}
+\tau_{n+6}^2\tau_{n+3}^2 \tau_{n+2}\tau_{n+1} 
+\tau_{n+6}\tau_{n+5} \tau_{n+4}^2\tau_{n+1}^2 
&& 
\\
+\tau_{n+6}\tau_{n+5}^2 \tau_{n+2}^2\tau_{n+1}
-\al \tau_{n+5}\tau_{n+4}^2 \tau_{n+3}^2\tau_{n+2} 
+\be \tau_{n+6}\tau_{n+5}\tau_{n+4}\tau_{n+3} \tau_{n+2}\tau_{n+1}
=0 
\end{array} 
\eeq
has the Laurent property, in the sense that it 
generates Laurent polynomials in $\tau_0,\ldots,\tau_6$ with coefficients in 
$\Z[\al,\be]$ (the case  $N=4$ of Proposition 2.3 in
\cite{hkq}). 
The point is that (\ref{pivtau}) coincides with the tau function substitution found for maps 
related with lattice KdV reductions in \cite{hkq}, and the relation (\ref{alberec}) holds for 
solutions of   (\ref{taurec}) obtained by setting 
$a=0$, $\nu=\be$, and taking an orbit for which the value of the first integral $K_1$ is 
fixed to be $k_1=-\al$.  
\end{remark}  
\begin{exa} Upon taking $a=0$, $b=-17$, $\nu=-11$ and making a specific choice of 7 integer initial values for (\ref{taurec}), with the 3 central values fixed to be 1, we generate an integer sequence that begins as follows:
$$ 
3,2,1,1,1,4,5,699,-25626,453024,-112570254, 23354432973,61327997061471, 
-35520663450983076, \ldots .
$$ 
Then we see that this sequence also satisfies the relation (\ref{alberec}) with 
$\be=-11$ and $\al=-k_1=327$. Also, since three of the coefficients in (\ref{s9iv}) vanish when $a=0$, 
we see that the Somos-9 relation for this sequence takes the shorter (three-term) Gale-Robinson form 
$$ 
\tau_{n+9}\tau_n +327\,\tau_{n+8}\tau_{n+1}+3850083\,\tau_{n+5}\tau_{n+4}=0, 
$$ 
which is a reduction of the discrete Hirota equation, and is in agreement with the $N=4$ case of 
Theorem 1.1 in \cite{hkq}.
\end{exa} 
\smallskip

The Laurent property for (\ref{taurec}), together with the formula (\ref{pivtau}), immediately implies that the
generic orbit $(w_n)$ of (P.iv) is well-defined, as stated in the following corollary.

\begin{cor} \label{generics}
  For generic non-zero initial values $(w_0,w_1,w_2,w_3)\in(\C^*)^4$, the orbit $(w_n)_{n\in\Z}$ exists, with
  $w_n\in \Pro^1=\C\cup\{\infty\}$.
\end{cor} 
\begin{proof}
As we shall see in Section \ref{par:hankel}, the initial tau functions can be chosen so that all $\tau_n$ are
polynomials in the initial data for (\ref{pivmap}). For a fixed index $n$ it is then clear from (\ref{pivtau}) that $w_n$ is an
indeterminate element of $\Pro^1$ only when at least two out of three successive tau functions in (\ref{pivtau})
vanish, i.e.\ belong to a certain proper Zariski closed subset of the space of (non-zero) initial data for
(\ref{pivmap}). Considering this condition for all $n$ yields a subset of this space of initial data, which is the
intersection of a countable family of Zariski open subsets. Such an intersection is a residual, hence dense, subset
so that for generic initial data, the orbit is well-defined.
\end{proof} 

Note that (P.iv) was originally defined as a birational affine map in $\C^4$, but the above corollary allows the existence of certain orbits defined in 
$(\Pro^1)^4$.
Notice also that this corollary does not state that the subset of initial data for which the orbit exists is open. This
stronger statement will follow from algebraic integrability, without use of the tau functions (see Section~\ref{sec:adi}).

\smallskip

As an initial foray into the geometry of the map, defined by (\ref{pivmap}), we now consider the singularity
pattern in more detail, by taking three non-zero initial values $w_0,w_1,w_2\in\C^*$ followed by a fourth
value proportional to a small parameter $\eps\in\C$, and consider the behaviour of the solution in the limit $\eps\to
0$. To reformulate this in terms of tau functions, we set
$$
  \tau_0=Z\;, \;\tau_1=\tau_2=\tau_3=1,\,\tau_4=X\;,\;\tau_5=Y\;,\;\tau_6=\eps\;,\; \quad XYZ\neq 0\;, 
$$
where three adjacent values have been set equal to 1 by a choice of gauge. This gives, using \eqref{pivtau}, four
non-zero initial values
\begin{equation}\label{wvals}
  w_0=Z\;,\quad w_1=X\;,\quad w_2=\frac{Y}{X}\;,\quad w_3=\frac{\eps}{XY}\;,
\end{equation}
for the map (\ref{pivmap}), such that the fourth value $w_3\to 0$ as $\eps\to 0$. Upon substituting these values in \eqref{pivmap}
we find as subsequent values 
$$ 
  w_4=C_4\, \eps^{-1}+O(1)\;, \quad w_5=C_5\,\eps^{-1}+O(1)\;,\quad  w_6 =C_6\,\eps+O(\eps^2)\;, \quad w_7=C_7+O(\eps)\;, 
$$
for certain coefficients $C_j$ which are rational functions of $X,Y,Z$. Notice that the leading order behaviours of $w_3,w_4,w_5,w_6$ are $\eps,\eps^{-1},\eps^{-1},\eps$, respectively, with terms of $O(1)$ on either side, which  corresponds to the
singularity pattern (\ref{eq:p-adic_norms}) obtained above by the $p$-adic method. Now if we substitute the initial values (\ref{wvals})
into $K_1=k_1$,\; $K_2=k_2$, where $K_1$ and $K_2$ are the invariants (\ref{pivh1}) and (\ref{pivh2}), and take the
limit $\eps\to 0$, then (after clearing denominators) we get two polynomial relations between $X,Y,Z$, which define
an affine algebraic curve. Upon taking resultants with respect to $Z$, this yields a single relation between $X$
and $Y$, namely 
\begin{align}\label{eq:curve_non_wei}
  & \big(aY^2 - (\nu k_1 +k_2)Y - a k_1\big) X^4 + \big((a\nu - k_1) Y^2 + (a^2 - b k_1) Y + k_1^2\big) X^3 \\ 
  &+\big(2aY^3 + (ab - \nu k_1) Y^2 - 2ak_1 Y\big) X^2 + Y^2\big((a\nu - k_1) Y + a^2\big) X + aY^4=0\;. \nonumber
\end{align} 
This plane curve is birationally equivalent to the Weierstrass quintic (\ref{wg2}), which can be seen from the
transformation
\begin{equation}\label{eq:curve_to_wei}
  x=\frac{X(aY-k_1X)}{Y(aX^2 - k_1X + aY)}\;, 
\end{equation}
together with a corresponding formula for $y=y(X,Y)$, which is rather unwieldy and so is omitted. 
For generic values of $k_1,k_2$, the curve \eqref{eq:curve_non_wei} is a smooth hyperelliptic curve of genus 2. As
already said, in its Weierstrass form \eqref{wg2}, this family of curves will play an important role in all that
follows.

\section{S-fractions on hyperelliptic curves and Volterra maps}\label{sec:volterra} 
\setcounter{equation}{0}

In this section, for a fixed integer $g>0$, we introduce an affine space of triplets of polynomials, reminiscent of
the phase space of the Mumford system \cite{tata2,vanhaecke} and use Stieltjes continued fractions (S-fractions) to
construct a series of birational automorphisms of the affine space, indexed by $g$, which we describe in several
ways. When $g=1$ we recover several known integrable maps, and for $g=2$ we recover the map (P.iv), which was 
been the primary motivation for this study, while the maps for $g>2$ appear to be new. We also give solutions in
terms of Hankel determinants of the iterates of these maps, i.e., of the corresponding recursion relations.


\smallskip

For a fixed $g>0$ we consider the affine space
\begin{equation} \label{Mg}
  M_g:=\left\{(\cP(x),\cQ(x),\cR(x))\in\C[x]^3  ~\Bigl|~
  \begin{array}{ll}
    \deg\cP(x) \leqslant g\;,&\cP(0)=1\\
    \deg\cQ(x) \leqslant g\;,&\cQ(0)=2\\
    \deg\cR(x) \leqslant g+1\;,&\cR(0)=0
  \end{array}
 \right\}\;.
\end{equation}
It is clear that $M_g$ is an affine space of dimension $3g+1$: writing
\begin{equation}\label{PQform}
  \cP(x)=1+\sum_{i=1}^g\p ix^i\;,  \qquad
  \cQ(x)=2+\sum_{i=1}^g\q ix^i\;,  \qquad
  \cR(x)=\sum_{i=1}^{g+1}\r ix^i\;,
\end{equation}
the coefficients $\p1,\dots,\p g,\q1,\dots,\q g,\r1,\dots,\r{g+1}$ provide a natural system of linear coordinates
on $M_g$.  We will often write an element $(\cP,\cQ,\cR)$ of $M_g$ as a traceless $2\times 2$ matrix
\begin{equation}\label{eq:def_lax}
  \lax(x) := \left(\begin{array}{lr} 
    \cP(x) & \cR(x)\\ 
    \cQ(x) & -\cP(x) 
  \end{array}\right)\;,
\end{equation}
which will later serve  as a Lax operator, and think of $M_g$ as an affine space of matrices (Lax operators). It is
then natural to consider the map $\mu$  
defined by
\begin{equation}\label{eq:momentum_map}
  \begin{array}{lcrcl}
    \mu&:&M_g&\to&\C[x]\\
    & & \lax(x)=\left(\begin{array}{lr}
  \cP(x)&\cR(x)\\
  \cQ(x)&-\cP(x)
  \end{array}\right)&\mapsto&-\det \lax(x)=\cP(x)^2+\cQ(x)\cR(x)\;.
  \end{array}
\end{equation}
%
%
In view of the degree constraints on the entries of $\lax(x)$, the polynomial $-\det \lax(x)$ has degree at most
$2g+1$ and its constant term is 1; moreover, every such polynomial is contained in the image of $\mu$. In the $g=2$
case these curves are precisely the ones 
encountered in Section \ref{sec:singularity} in the
singularity analysis of (P.iv), see \eqref{wg2}, \eqref{eq:curve_non_wei} and \eqref{eq:curve_to_wei}, which in
part motivates the choice of constraints on the polynomials $\cP,\cQ$ and $\cR$. 
(When $M_g$ is endowed with a Poisson structure, as in Section \ref{sec:adi}, 
then $\mu$ can be viewed a momentum map.)

\smallskip

Throughout this section, $g>0$ is fixed. In each of the following 
subsections, the results and phenomena being discussed will  be specialized and illustrated for the cases of $g=1$ and
$g=2$, when $M_g$ has dimension 4 and 7, respectively. 

\subsection{Stieltjes continued fractions}\label{par:stieltjes} 
We start from a hyperelliptic curve $\Gamma_f$, defined by an odd Weierstrass equation 
\begin{equation}\label{weier}
  \Gamma_f: y^2 = f(x)\;, \quad\hbox{with}\quad f(x):=1+\sum_{j=1}^{2g+1} c_jx^j\in\C[x]\;.
\end{equation}
When $f$ has degree $2g+1$ and has no multiple roots, $\Gamma_f$ is non-singular and its genus is $g$, which
explains the notation used. Let $(\cP,\cQ,\cR)$ be any point in $\mu^{-1}(f)$, the fiber of $\mu$ over~$f$, so that
$f(x)=\cP^2(x)+\cQ(x)\cR(x)$, and the spectral curve  $\Gamma_f$ is 
the characteristic equation $\det ( \lax(x)-y\mathbf{1})=0$. 
We consider on $\Gamma_f$ the rational function, given by
\begin{equation}\label{func}
  F := \frac{y+\cP(x)}{\cQ(x)}=\frac{\cR(x)}{y-\cP(x)}\;.
\end{equation}

In preparation for constructing the Stieltjes continued fraction of $F$, we show in the following lemma how the
triplet $(\cP,\cQ,\cR)$ leads to another triplet $(\tilde\cP,\tilde\cQ,\tilde\cR)$ in the same fiber of $\mu$,
under the assumption that $(\cP,\cQ,\cR)$ is \emph{regular}, meaning that 
\begin{equation} \label{generic}
\frac{2\cP(x)-\cQ(x)+\cR(x)}{x}\big\vert_{x=0}= 2 p_1 - q_1 + r_1 \neq 0\;.
\end{equation}
\begin{lem}\label{lma:stieltjes}
  Given a regular triplet $(\cP,\cQ,\cR)$ in $\mu^{-1}(f)$, there exists a unique
  $w\in\C^*$ and a unique triplet $(\tilde\cP,\tilde\cQ,\tilde\cR)$ in $\mu^{-1}(f)$ such that
  \begin{equation}\label{eq:trans_in_lemma}
    \frac{y+\cP(x)}{\cQ(x)}=1-\frac{w\,x}{\frac{y+\tilde\cP(x)}{\tilde\cQ(x)}}\;.
  \end{equation}
  The two triplets are related by
  \begin{equation}\label{PQrec2}
  \tilde\cP(x)=\cQ(x)-\cP(x)\;,\quad\tilde\cQ(x)=\frac{2\cP(x)-\cQ(x)+\cR(x)}{-w x}\;,\quad
  \tilde\cR(x)=-w x\cQ(x)\;,
  \end{equation}
  and 
  \begin{equation}\label{weq}
      w = - \frac{2\,p_1-q_1+r_1}{2} =-\frac{\tilde{r}_1}2\;.
  \end{equation}
\end{lem}
\begin{proof}
We will constructively show that we can achieve \eqref{eq:trans_in_lemma} with $w\in\C^*$ and
$(\tilde\cP,\tilde\cQ,\tilde\cR)\in\mu^{-1}(f)$ uniquely determined.  Clearing the denominators in
\eqref{eq:trans_in_lemma} and using $y^2=f(x)$ we get
\begin{equation}\label{eq:clearing}
  y(\cP(x)+\tilde\cP(x)-\cQ(x))+f(x)+\cP(x)\tilde\cP(x)-\tilde\cP(x)\cQ(x)+w x\cQ(x)\tilde\cQ(x)=0\;,
\end{equation}
an equality which holds in the field of fractions of $\C[x,y]/(y^2-f(x))$, so that the coefficients of $y$ and of
$y^0$ in~\eqref{eq:clearing} must be zero. The vanishing of the former coefficient gives the first equation in
\eqref{PQrec2} and guarantees $\tilde\cP(0)=1$ and $\deg\tilde\cP\leqslant g$. The vanishing of the $y^0$
coefficient gives
\begin{equation}\label{nextQ}
  -w\,\tilde\cQ(x)=\frac{f(x)+(\cP(x)-\cQ(x))\tilde\cP(x)}{x\cQ(x)}
  =\frac{f(x)-(\cP(x)-\cQ(x))^2}{x\cQ(x)}=\frac{2\cP(x)-\cQ(x)+\cR(x)}{x}\;,
\end{equation}
where we have used in the last step that $f(x)=\cP^2(x)+\cQ(x)\cR(x)$. Notice that in view of the values of the
constant terms in \eqref{PQform}, the last numerator in \eqref{nextQ} vanishes for $x=0$ and is of degree at most
$g+1$. Also, as the triplet $(\cP,\cQ,\cR)$ is assumed to be regular, the polynomial $(2\cP-\cQ+\cR)/x$ does not
vanish at $x=0$, hence we can (uniquely) choose $w\in\C^*$ such that $\tilde\cQ(0)=2$. This gives the first
equality in \eqref{weq} and the second equation in \eqref{PQrec2}. The first equality in \eqref{nextQ} also shows
that $f(x)-\tilde\cP^2(x)$ is divisible by $\tilde\cQ(x)$, with quotient $-wx\cQ(x)$; thus, if we take the third
equality in \eqref{PQrec2} to define $\tilde\cR$, then $\deg\tilde\cR = g+1$, $\tilde\cR(0)=0$ and
$f(x)=\tilde\cP^2(x)+\tilde\cQ(x)\tilde\cR(x)$, completing the proof that $(\tilde\cP,\tilde\cQ,\tilde\cR)$ belongs
to $\mu^{-1}(f)$. Notice that the third equality in \eqref{PQrec2} implies the alternative formula for $w$ in
\eqref{weq}, since $\cQ(0)=2$.
\end{proof}
Applying the lemma to all regular points of the fiber $\mu^{-1}(f)$ yields a rational map of the fiber to itself,
given by $(\cP,\cQ,\cR)\mapsto (\tilde\cP,\tilde\cQ,\tilde\cR)$ with $(\tilde\cP,\tilde\cQ,\tilde\cR)$ given by
\eqref{PQrec2}. Since the latter can also be solved rationally for $(\cP,\cQ,\cR)$ in terms of
$(\tilde\cP,\tilde\cQ,\tilde\cR)$ by using the second expression for $w$ in \eqref{weq}, this rational map is
actually a birational automorphism of the fiber. Iterating this map starting from a triplet
$(\cP_0,\cQ_0,\cR_0)=(\cP,\cQ,\cR)\in\mu^{-1}(f)$, we get an infinite sequence $(\cP_n,\cQ_n,\cR_n)_{n\in\Z}$ of triplets
as well as an infinite sequence $(w_n)_{n\in\Z}$ in $\C^*$, such that $(\cP_{n+1},\cQ_{n+1},\cR_{n+1})$
and~$w_{n+1}$ are related to $(\cP_n,\cQ_n,\cR_n)$ as dictated by the lemma:
\begin{equation}\label{PQrecn}
  \cP_{n+1}(x)=\cQ_n(x)-\cP_n(x)\;,\quad\cQ_{n+1}(x)=\frac{2\cP_n(x)-\cQ_n(x)+\cR_n(x)}{-w_{n+1} x}\;,\quad
  \cR_{n+1}(x)=-w_{n+1}x\cQ_n(x)\;.
\end{equation}
Writing, as in \eqref{PQform}, $\cP_n(x)=1+\sum_{i=1}^g\p{n,i}x^i$, and similarly for $\cQ_n(x)$ and $\cR_n(x)$,
the value of $w_{n+1}$ is given, according to \eqref{weq}, by
\begin{equation}\label{eq:wn}
  w_{n+1}=-\frac{2\p{n,1}-\q{n,1}+\r{n,1}}2=-\frac{\r{n+1,1}}2\;,
\end{equation}
for all $n\in\Z$.

\goodbreak

It is clear that $(\cP_n,\cQ_n,\cR_n)$ is obtained from $(\cP_0,\cQ_0,\cR_0)$ by repeating  the map $n$ times (or, 
when $n<0$, repeating the inverse of the map $-n$ times). As pointed out in the lemma, the starting triplet
$(\cP_0,\cQ_0,\cR_0)$ must be regular in order for $(\cP_1,\cQ_1,\cR_1)$ and $w_1\in\C^*$ to exist. But nothing
guarantees that $(\cP_1,\cQ_1,\cR_1)$ will also be regular, and in general it need not be so; assuming
$(\cP_1,\cQ_1,\cR_1)$ to be a regular triplet puts an open linear condition on the coefficients of
$(\cP_1,\cQ_1,\cR_1)$, namely that $2\p{1,1} - \q{1,1} + \r{1,1} \neq 0$, which amounts to an open polynomial
condition on $(\cP_0,\cQ_0,\cR_0)$, and for the existence of every extra term of the sequence such an extra
condition is to be added to the triplet $(\cP_0,\cQ_0,\cR_0)$. However, since this amounts to a countable number of
open conditions on the latter, this means that when $(\cP_0,\cQ_0,\cR_0)$ is \emph{generic}, in the sense that it
belongs to a residual subset of the fiber $\mu^{-1}(f)$, the sequence of triplets of polynomials
$(\cP_n,\cQ_n,\cR_n)\in\mu^{-1}(f)$ and the sequence of constants $w_n\in\C^*$, both indexed by $n\in\Z$, exist.
For  generic $(\cP_0,\cQ_0,\cR_0)$, iterating \eqref{eq:trans_in_lemma} gives
\begin{equation}\label{eq:F1}
  F_0:= \frac{y+\cP_0(x)}{\cQ_0(x)}=1-\frac{w_1x}{\frac{y+\cP_1(x)}{\cQ_1(x)}}=\cdots=
  1-\cfrac{w_1x}{ 1 -\cfrac {w_2x}{ 1-\cfrac{w_3x}{1- \cdots} } } \;,
\end{equation}
yielding the \emph{Stieltjes continued fraction}, also called \emph{$S$-fraction}, of $F_0$. 
Similarly, each triple $(\cP_n,\cQ_n,\cR_n)$, $n\in\Z$, is associated with a rational function $F_n$, 
with a corresponding S-fraction obtained by shifting each of the indices in (\ref{eq:F1}), which for $n>0$ appears 
on the $n$th line below the top.  
\begin{exa}
Suppose that $g=1$. Then the entries of the triplets $(\cP,\cQ,\cR)$ and sequences
$(\cP_n,\cQ_n,\cR_n)_{n\in\Z}$ in $M_g=M_1$ take the form
\begin{equation}
  \begin{array}{rcl}
    \cP(x)&=&1+\p1x\;,\\
    \cQ(x)&=&2+\q1x\;,\\
    \cR(x)&=&\r1x+\r2x^2\;,
  \end{array}
  \qquad\hbox{and}\qquad
    \begin{array}{rcl}
    \cP_n(x)&=&1+\p{n,1}x\;,\\
    \cQ_n(x)&=&2+\q{n,1}x\;,\\
    \cR_n(x)&=&\r{n,1}x+\r{n,2}x^2\;.
  \end{array}
\end{equation}
The birational automorphism \eqref{PQrec2}, constructed in Lemma \ref{lma:stieltjes}, and its iterates are given by
\begin{equation}\label{eq:exa_rec_g=1}
  \begin{array}{rcl}
    \tilde p_1&=&\q1-\p1\;,\\
    \tilde q_1&=&-r_2/w\;,\\
    \tilde r_1&=&-2w\;,\\
    \tilde r_2&=&-wq_1\;,
  \end{array}
  \qquad\hbox{and}\qquad
  \begin{array}{rcl}
    \p{n+1,1}&=&\q{n,1}-\p{n,1}\;,\\
    \q{n+1,1}&=&-\r{n,2}/w_{n+1}\;,\\
    \r{n+1,1}&=&-2w_{n+1}\;,\\
    \r{n+1,2}&=&-w_{n+1}\q{n,1}\;,
  \end{array}
\end{equation}
where $w = - \frac{2\,p_1-q_1+r_1}{2}$ and $w_{n+1}= - \frac{2\,\p{n,1}-\q{n,1}+\r{n,1}}{2} $ for $n\in\Z$.
\end{exa}

\begin{exa}
Suppose now that $g=2$. The entries of the triplets $(\cP,\cQ,\cR)$ and sequences
$(\cP_n,\cQ_n,\cR_n)_{n\in\Z}$ in $M_2$ now take the form
\begin{equation}
  \begin{array}{rcl}
    \cP(x)&=&1+\p1x+\p2x^2\;,\\
    \cQ(x)&=&2+\q1x+\q2x^2\;,\\
    \cR(x)&=&\r1x+\r2x^2+\r3x^3\;,
  \end{array}
  \qquad\hbox{and}\qquad
    \begin{array}{rcl}
    \cP_n(x)&=&1+\p{n,1}x+\p{n,2}x^2\;,\\
    \cQ_n(x)&=&2+\q{n,1}x+\q{n,2}x^2\;,\\
    \cR_n(x)&=&\r{n,1}x+\r{n,2}x^2+\r{n,3}x^3\;.
  \end{array}
\end{equation}
From the construction in Lemma \ref{lma:stieltjes},
the birational automorphism \eqref{PQrec2}  and its iterates are given by
\eqref{eq:exa_rec_g=1}, with the expression for $\tilde q_1$ modified to 
$$ 
\tilde q_1 = -(2p_2-q_2+r_2)/w, \quad \mathrm{and} \quad 
q_{n+1,1} = -(2p_{n,2}-q_{n,2}+r_{n,2})/w_{n+1},
$$ 
further supplemented with the following formulae:
\begin{equation}\label{eq:exa_rec_g=2}
  \begin{array}{rcl}
    \tilde p_2&=&\q2-\p2\;,\\
    \tilde q_2&=&-r_3/w\;,\\
    \tilde r_3&=&-w\q2\;,
  \end{array}
  \qquad\hbox{and}\qquad
  \begin{array}{rcl}
    \p{n+1,2}&=&\q{n,2}-\p{n,2}\;,\\
    \q{n+1,2}&=&-\r{n,3}/w_{n+1}\;,\\
    \r{n+1,3}&=&-w_{n+1}\q{n,2}\;,
  \end{array}
\end{equation}
where, as in the genus 1 case, $w = - \frac{2\,p_1-q_1+r_1}{2}$ and $w_{n+1}= - \frac{2\,\p{n,1}-\q{n,1}+\r{n,1}}{2} $
for $n\in\Z$.
\end{exa}

\goodbreak

\subsection{Lax equation and invariants}
In Section \ref{par:stieltjes} 
we defined  a birational automorphism of the fiber $\mu^{-1}(f)\subset M_g$,
where $f=f(x)$ is any polynomial of degree at most~$2g+1$, satisfying $f(0)=1$. This map, given by \eqref{PQrec2},
is not just defined on $\mu^{-1}(f)$, but is also as it stands a well-defined birational automorphism of~$M_g$. In view
of its relation to the Volterra lattice (see Section \ref{sec:continuous}), we call it the \emph{Volterra map},
denoted $\cV_g$; explicitly, 
$$\cV_g: \qquad (\cP,\cQ,\cR) \mapsto (\tilde\cP,\tilde\cQ,\tilde\cR),$$ 
where the entries of the latter
are given by \eqref{PQrec2}; also, we can write $\cV_g(\cP_n,\cQ_n,\cR_n)=(\cP_{n+1},\cQ_{n+1},\cR_{n+1})$ for all $n\in\Z$. 
For a fixed initial triplet $(\cP_0,\cQ_0,\cR_0)$, the entire sequence of triplets $(\cP_n,\cQ_n,\cR_n)_{n\in\Z}$ in $M_g$ is called 
the \textit{orbit} of $\cV_g$ through $(\cP_0,\cQ_0,\cR_0)$. We also refer to a sequence of triplets that satisfies the recursion relations 
\eqref{PQrecn} for all $n$ as a \textit{solution}. 
The equations~\eqref{PQrec2} for the Volterra map, as well as the recursion relations \eqref{PQrecn} for its
iterates, are easily rewritten as discrete Lax equations; this fact has many important consequences which will be
worked out in what follows.
\begin{propn} \label{laxGamma}
The Volterra map $\cV_g$ can be written in the compact form
\begin{equation}\label{dlax} 
\cV_g: \qquad 
  \lax(x)\mma(x) = \mma(x)\widetilde\lax(x)\;, 
\end{equation}
where $\lax(x)$ is given by \eqref{eq:def_lax}, $\widetilde\lax(x)$ is $\lax(x)$ with $\cP,\,\cQ,\,\cR$ replaced by
$\tilde\cP,\,\tilde\cQ,\,\tilde\cR$, and $ \mma(x) := \left(\begin{array}{cc} 1 & -w x \\ 1 & 0\end{array}\right),$
with $w$ given by \eqref{weq}. As a consequence, the $2g+1$ polynomial functions $H_1,\dots,H_{2g+1}$ on $M_g$,
defined by
  \begin{equation}\label{eq:Ham_def}
    \cP(x)^2+\cQ(x)\cR(x)=1+\sum_{i=1}^{2g+1}H_ix^i
  \end{equation}
are invariants of the Volterra map, i.e. $\tilde H_i=H_i$ for $i=1,\dots,2g+1.$
\end{propn}
\begin{proof}
  It is easily checked by direct computation that \eqref{PQrec2} and \eqref{dlax} are the same set of
  equations. Since the latter says that $\widetilde\lax(x)$ is obtained from $\lax(x)$ by conjugation with
  $\mma(x)$, the spectrum of $\lax(x)$ is preserved, hence also all coefficients of the determinant of $\lax(x)$,
  i.e., the coefficients $H_i$ of $\mu(\lax(x))=\cP(x)^2+\cQ(x)\cR(x)$.
\end{proof}
%
%
Upon iterating the Volterra map, as discussed in Section \ref{par:stieltjes}, starting from a generic triplet $\lax_0(x)$
of $M_g$  we get a sequence of triplets $$\lax_n(x):= \left(\begin{array}{cr} \cP_n(x) & \cR_n(x)\\ \cQ_n(x) & -\cP_n(x)
\end{array}\right)$$ of $M_g$.  According to \eqref{dlax}, a discrete Lax equation for this sequence is
given by
\begin{equation}\label{dlaxn} 
  \lax_n(x)\mma_n(x) = \mma_n(x)\lax_{n+1}(x)\;, 
\end{equation}
where 
\begin{equation}\label{LM} 
  \mma_n(x) := \left(\begin{array}{cc}
    1 & -w_{n+1} x \\ 
    1 & 0 
  \end{array}\right)\;,\qquad\hbox{with}\qquad w_{n+1}=-\frac{2\p{n,1}-\q{n,1}+\r{n,1}}2=-\frac{\r{n+1,1}}2\;.
\end{equation}

\begin{exa}
When $g=1$, respectively when $g=2$, the invariants $H_i$ can be computed from \eqref{eq:Ham_def}:
\begin{equation}\label{eq:hams_g=1}
\begin{array}{rcl}
  H_1&=&2(p_1+r_1)\;,\\
  H_2&=&p_1^2+q_1r_1+2r_2\;,\\
  H_3&=&q_1r_2\;,
\end{array}
\qquad\qquad
\begin{array}{rcl}
  H_1&=&2(p_1+r_1)\;,\\
  H_2&=&2p_2+p_1^2+q_1r_1+2r_2\;,\\
  H_3&=&2p_1p_2+2r_3+q_1r_2+q_2r_1\;,\\
  H_4&=&p_2^2+q_1r_3+q_2r_2\;,\\
  H_5&=&q_2r_3\;.
\end{array}
\end{equation}
The formulae on the left, which correspond to $g=1$, can be obtained from the first three formulae on the right by
setting $\p2=\q2=\r3=0$ in them. 
\end{exa}


\subsection{The Volterra map and its reductions}\label{par:rec_from_volterra}
The invariants $H_i$ 
can be used to reduce the 
Volterra map to the submanifolds obtained by
fixing the values of some of these invariants. Here we will use this to express the Volterra map in terms of the
variables $w_i$ which we introduced when constructing the S-fraction \eqref{eq:F1}.

\smallskip

We start from the linear coordinates $\p1,\dots,\p g,\q1,\dots,\q g,\r1,\dots,\r{g+1}$ of $M_g$, which are identified 
with $\p{0,1},\dots,\p{0,g},\q{0,1},\dots,\q{0,g},\r{0,1},\dots,\r{0,g+1}$.  The latter functions are used to
define recursively $\p{n,1},\dots,\p{n,g},$ $\q{n,1},\dots,\q{n,g},\r{n,1},\dots,\r{n,g+1}$, as well as $w_n$, for
all $n\in\Z$.  Recall that this is done using \eqref{PQrecn} and \eqref{eq:wn}.

\smallskip

In a first step, we will use a birational map to replace our linear coordinates for $M_g$ by $\p{0,1},\dots,\p{0,g}$
and some of its iterates $\p{n,1},\dots\p{n,g}$. To do this, we fix the value of the invariant $H_1=2(p_1+r_1)$ to
an arbitrary constant~$c_1$. It means that we consider the hyperplane $H_1=c_1$ of $M_g$, which denote by
$M_g^{c_1}$. On it, we can take $\p1,\dots,\p g,\q1,\dots,\q g,\r2,\dots,\r{g+1}$ as linear coordinates (we left
out $r_1$).  The invariance of $H_1$ implies that $2(\p{n,1}+\r{n,1})=c_1$ which, combined with $r_{n,1}=-2w_n$
(see \eqref{eq:wn}), leads to
\begin{equation}\label{eq:wn_to_p}
  w_n=\frac12\left(\p{n,1}-\frac{c_1}2\right)\;
\end{equation}
for all $n\in\Z$.  Using this and the first and last equations in \eqref{PQrecn}, we can express the above variables
in terms of $\p{0,1},\dots,\p{0,g}$ and their iterates:
\begin{equation}\label{eq:toPn}
  \q{n,k}=\p{n+1,k}+\p{n,k}\;,\quad\hbox{and}\quad
  \r{n,k+1}=-w_{n}\q{n-1,k}=\frac12\left(\frac{c_1}2-\p{n,1}\right)(\p{n,k}+\p{n-1,k})\;,
\end{equation}
where $k=1,\dots,g$. Taking $n=0$ we get 
\begin{equation}\label{eq:toP}
  \q{0,k}=\p{1,k}+\p{0,k}\;,\quad\hbox{and}\quad
  \r{0,k+1}=\frac12\left(\frac{c_1}2-\p{0,1}\right)(\p{0,k}+\p{-1,k})\;,
\end{equation}
and we have expressed all coordinates of $M_g^{c_1}$ in terms of the $3g$ coefficients of the polynomials
$\cP_{-1},\cP_0$ and~$\cP_1$. It is clear that the \eqref{eq:toP} can be solved rationally for $\p{1,k}$ and
$\p{-1,k}$ so that \eqref{eq:toP} defines a birational morphism from $M_g^{c_1}$ to the space of triplets
$(\cP_{-1},\cP_0,\cP_1)$. For later use, we also express the iterates of the Volterra map as a recursion relation
on the polynomials $\cP_n$. To do this, we apply~\eqref{PQrecn} several times to get
\begin{align}\label{eq:P_rec_P}
  \cP_{n+2}&=\cQ_{n+1}(x)-\cP_{n+1}(x)=\frac{2\cP_n(x)-\cQ_n(x)+\cR_n(x)}{-w_{n+1}x}-\cP_{n+1}(x)\nonumber\\
    &=-\cP_{n+1}(x)+\frac1{w_{n+1}x}(\cP_{n+1}(x)-\cP_n(x))+\frac{w_{n}}{w_{n+1}}(\cP_n(x)+\cP_{n-1}(x))\;.
\end{align}

We will now go one step further and show how the above coordinates $\p{0,k},\p{1,k}$ and $\p{-1,k}$ can be
expressed birationally in terms of $\p{0,1}$ and some of its iterates $\p{n,1}$. To do this, we do a further
reduction, namely we also fix the value of each one of the invariants $H_2,\dots,H_g$ to an arbitrary constant
$c_2,\dots,c_g$ and consider the subvariety $\cap_{i=1}^g(H_i=c_i)$ of $M_g$, which we denote by $M_g^{c}$, so $c$
stands now for $(c_1,\dots,c_g)$. Notice that this subvariety may be singular, but that does not affect the
reduction or the recursion relations. Using \eqref{eq:toPn}, we get the following formula for the invariants $H_i$
in terms of $\cP_{n},\cP_{n+1}$ and $\cP_{n-1}$, valid for any $n\in\Z$:
\begin{equation*}
  1+\sum_{k=1}^{2g+1}H_ix^i=\cP_n^2(x)+\cQ_n(x)\cR_n(x)=\cP_n^2(x)-w_{n}x(\cP_n(x)+\cP_{n+1}(x))(\cP_n(x)+\cP_{n-1}(x))\;,
\end{equation*}
with $w_n$ given by \eqref{eq:wn_to_p}. Upon comparing the coefficient of $x^k$ on both sides, for $k=1$ we recover 
\eqref{eq:wn_to_p}, while for $k=2,\dots,g$ we  recursively obtain  $\p{n,k}$ 
 in terms of $\p{0,1}$ and its iterates, via the following formulae:
\begin{align}\label{eq:ck}
  c_k&=2\p{n,k}+\sum_{i=1}^{k-1}\p{n,i}\p{n,k-i}-2w_{n}(2\p{n,k-1}+\p{n-1,k-1}+\p{n+1,k-1})\nonumber\\
    &\;\quad\quad-w_{n}\sum_{i=1}^{k-2}(\p{n,i}+\p{n+1,i})(\p{n,k-1-i}+\p{n-1,k-1-i})\;.
\end{align}
Indeed,  aside from the linear term in $\p{n,k}$, the above equation contains only the variables $\p{n,i}$ and $\p{n\pm1,i}$ with
$1\leqslant i<k$. For $k=2$ one gets
\begin{equation*}
  c_2=2\p{n,2}+\p{n,1}^2-2w_{n}(2\p{n,1}+p_{n-1,1}+\p{n+1,1})\;,
\end{equation*}
from which it is clear that $\p{n,2}$ depends (polynomially) only on $\p{n-1,1},\p{n,1}$ and $\p{n+1,1}$ (see
\eqref{eq:wn_to_p} for the formula for $w_{n}$). An easy recursion on $k$ using \eqref{eq:ck} shows that $\p{n,k}$
depends on $\p{n-k+1,1},\dots,p_{n+k-1,1}$ only. Taking $n=-1$, $n=0$ and $n=-1$ we get that the coefficients of
$\cP_{-1},\cP_0$ and $\cP_1$ depend only on $\p{-g,1},\dots,p_{g,1}$. Conversely, it is obvious from
\eqref{eq:P_rec_P} that the coefficients of $\cP_{2}$, and hence of all $\cP_n$ with $n\geqslant2$, are rational
functions of the coefficients of $\cP_{-1},\cP_0$ and $\cP_1$. This applies in particular to $\p{n,1}$ with
$n\geqslant2$. Using the inverse 
recursion, which yields a formula similar to \eqref{eq:P_rec_P} expressing
$\cP_{n-2}$ in terms of $\cP_{n+1},\cP_n$ and $\cP_{n-1}$ one obtains similarly that all $\p{n,1}$ with $n\in\Z$
are rational functions of the coefficients of $\cP_{-1},\cP_0$ and $\cP_1$. The upshot is that we have a birational
map between $M_g^c$ and $\C^{2g+1}$, equipped with the coordinates $\p{-g,1},\dots,\p{g,1}$. In view of
\eqref{eq:wn_to_p}, which we now write as
\begin{equation}\label{eq:p_to_w}
  \p{n,1}=2w_n+\frac{c_1}2\;,  
\end{equation}
it amounts to a birational map between $M_g^c$ and $\C^{2g+1}_w$ where the latter denotes $\C^{2g+1}$, equipped
with the coordinates $w_{-g},\dots,w_g$.

\smallskip

Hence we can use the birational map between $M_g^c$ and $\C^{2g+1}_w$ to write the Volterra map on $M_g$, restricted
to $M_g^{c}$, as a birational automorphism of $\C^{2g+1}_w$. Since we already gave in \eqref{eq:P_rec_P} the
Volterra map and its iterates in terms of the variables $\cP_i$, we set  $n=0$ therein, which gives the Volterra map itself,
and take the leading terms of both sides of \eqref{eq:P_rec_P}:
\begin{equation}\label{eq:rec_gen}
  w_{1}(\p{2,g}+\p{1,g})=w_{0}(\p{0,g}+\p{-1,g})\;.
\end{equation}
In view of the dependence of $\p{-1,g},\dots,\p{2,g}$ on the variables $w_i$, \eqref{eq:rec_gen} gives an equation for 
$w_{g+1}$, which appears linearly in it, and the birational automorphism
$(w_{-g},w_{1-g},\dots,w_g)\mapsto(w_{1-g},\dots,w_g,w_{g+1})$ is the Volterra map on~$\C_w^{2g+1}$. Explicit
expressions for it will be given in the examples below.

\smallskip

However, we can do a further reduction, restricting the Volterra map to the subvariety $H_k=c_k$ for some~$k$ with
$g<k\leqslant 2g+1$. The relation
\begin{equation*}
  H_k(w_{-g},\dots,w_0,\dots,w_g)=c_k
\end{equation*}
defines $w_g$ as a rational function of $w_{-g},\dots,w_{g-1}$ because by inspection $w_g$ appears linearly in it
(the same applies to $w_{-g}$). As we will see in the examples below, we can therefore take any of the
invariants~$H_k$, with $g<k\leqslant 2g+1$ which will give a birational automorphism which is an incarnation of the
Volterra map $\cV_g$ on $M_g^c\cap(H_k=c_k)$; precisely it is conjugate, via the above birational map,
to the Volterra map, restricted to $M_g^c\cap(H_k=c_k)$, where values of $c=(c_1,\dots,c_g)$ and $c_k$ are
arbitrary.

\begin{exa}\label{exa:g=1_rec}
We first consider the case $g=1$. In this case, we only need to consider $k=1$ in \eqref{eq:toP}, which combined
with \eqref{eq:p_to_w} yields the following birational map between $M_g^c=M_g^{c_1}$ and $\C^3_w$:
\begin{align}\label{eq:pg1}
  \p{1}&=\p{0,1}=2w_0+c_1/2\;,\nonumber\\
  \q{1}&=\q{0,1}=\p{1,1}+\p{0,1}=2(w_0+w_1)+c_1\;,\nonumber\\
  \r{1}&=\r{0,1}=c_1/2-\p{0,1}\;,\\
  \r{2}&=\r{0,2}=-w_0(\p{0,1}+p_{-1,1})=-2w_0(w_0+w_{-1}+c_1/2)\;.\nonumber
\end{align}
In terms of the polynomials $\cP,\cQ$ and $\cR$ this can also be written as
\begin{align}\label{eq:pqr_n_g=1}
  \cP(x)&=1+(2w_{0}+{c_1}/2)x\;,\nonumber\\
  \cQ(x)&=2+2(w_{1}+w_{0}+{c_1}/2)x\;,\\
  \cR(x)&=-2w_{0}x(1+(w_{0}+w_{-1}+{c_1}/2)x)\;.\nonumber
\end{align}
For $g=1$ the formula \eqref{eq:rec_gen} takes the form $w_{1}(\p{2,1}+\p{1,1})=w_{0}(\p{0,1}+\p{-1,1})$, and can
be expressed immediately in terms of the quantities $w_{-1},\dots,w_2$ since $\p{n,1}=2w_{n}+c_1/2$, for all $n$:
\begin{equation}\label{eq:rec_g1}
  w_1(2w_{2}+2w_1+c_1)=w_{0}(2w_{0}+2w_{-1}+c_1)\;.
\end{equation}
It defines the Volterra map $(w_{-1},w_0,w_1)\mapsto(w_0,w_1,w_2)$ on $\C^3_w$, being the same as equation (2) in
\cite{svinin3}, where Svinin used continued fraction expansions to construct particular solutions;  
equation \eqref{eq:rec_g1} also appears in \cite{hones5}.  Substituting~\eqref{eq:pqr_n_g=1} in
$\cP^2(x)+\cQ(x)\cR(x)=1+H_1x+H_2x^2+H_3x^3$ we get the following formulas for the invariants $H_2$ and $H_3$ of
the Volterra map in terms of the variables $w_{-1},w_0$ and $w_1$:
\begin{align}
  H_2&=-4w_{0}(w_{1}+w_{0}+w_{-1}+{c_1}/2)+{c_1^2}/4\;,\label{eq:H2_g1}\\
  H_3&=-4w_{0}(w_{1}+w_{0}+{c_1}/2)(w_{0}+w_{-1}+{c_1}/2)\;.\label{eq:H3_g1}
\end{align}
We now fix $c_3$ and consider the Volterra map on the subvariety $H_3=c_3$ of $M_g^c$. According to
\eqref{eq:H3_g1} we get
\begin{equation}\label{eq:g=1_red_H3}
  4w_{0}(w_{1}+w_{0}+{c_1}/2)(w_{0}+w_{-1}+{c_1}/2)+c_3=0\;,
\end{equation}
which defines a $2D$ map $(w_{-1},w_0)\mapsto (w_0,w_1)$, where $w_1$ is computed from \eqref{eq:g=1_red_H3}. It
has $H_2$ (from which~$w_{1}$ is eliminated using \eqref{eq:g=1_red_H3}) as invariant, given by 
\beq\label{biquadH2}
  H_2=\frac{c_3}{w_0+w_{-1}+c_1/2}-4w_0w_{-1}+\frac{c_1^2}4\;.
\eeq 
We observe that $H_2$ is a ratio of symmetric biquadratics in $w_0$ and $w_{-1}$ and it can be checked that \eqref{eq:g=1_red_H3} is a symmetric QRT map \cite{qrt} 
(and it is of type (III) in the classification of \cite{rcgo}).
We next fix $c_2$ and consider the Volterra map on the subvariety $H_2=c_2$ of $M_1^c$. From \eqref{eq:H2_g1} we now
get
\begin{equation}\label{eq:g=1_red_H2}
  4w_{0}(w_{1}+w_{0}+w_{-1}+{c_1}/2)-{c_1^2}/4+c_2=0\;. 
\end{equation}
It defines a $2D$ map $(w_{-1},w_0)\mapsto (w_0,w_1)$ which is an additive QRT map (type (I) in \cite{rcgo}) 
with $H_3$  as invariant,
which (after using
\eqref{eq:g=1_red_H2}) to eliminate $w_{1}$) takes the form 
\beq\label{biquadH3}
  H_3=\left(w_0+w_{-1}+\frac{c_1}2\right)\left(4w_0w_{-1}+c_2-\frac{c_1^2}4\right)\;.
\eeq 

To finish the $g=1$ example we will present a slightly more involved reduction, leading to a map which is closely
related to Somos-5. To do this, we first compare two different ways of writing of the genus 1 curve $y^2=f(x)$, which lead
to an alternative generating set of invariants of the recursion. If we write
\beq\label{weicub} 
  y^2=1+c_1x+c_2x^2+c_3x^3=(1-c_1'x)(1-c_1'x+4c_2'x^2)-4c_3'x^3\;,
\eeq 
then the constants  $c_i$ and $c_i'$ are related by
\begin{equation}\label{eq:c_to_c'}
  \renewcommand*{\arraystretch}{1.5}
\begin{array}{rcl}
  c_1&=&-2c'_1\;,\\ 
  c_2&=&4c_2'+{c_1'}^2\;,\\
  c_3&=&-4(c'_1c'_2+c'_3)\;,
\end{array}
\qquad\qquad
\begin{array}{rcl}
  c'_1&=&-c_1/2\;,\\ 
  c'_2&=&\frac14\left(c_2-\frac{c_1^2}4\right)\;,\\
  c'_3&=&\frac{c_1}8\left(c_2-\frac{c_1^2}4\right)-\frac{c_3}4\;.
\end{array}
\end{equation}
Next, if we write the recursion relations \eqref{eq:g=1_red_H3} and \eqref{eq:g=1_red_H2} in terms of the constants $c_i'$
using \eqref{eq:c_to_c'}, we get respectively
\begin{gather}
  w_{0}\left(w_{1}+w_{0}-{c'_1}\right)\left(w_{0}+w_{-1}-{c'_1}\right)=c'_1c'_2+c'_3\;,\label{eq:rec_1}\\
  w_{0}\left(w_{1}+w_{0}+w_{-1}-{c'_1}\right)+c'_2=0\label{eq:rec_2}\;.
\end{gather}
Notice that $c_1'$ now appears linearly in \eqref{eq:rec_2}, so we can easily eliminate $c'_1$ between
\eqref{eq:rec_1} and \eqref{eq:rec_2}, which yields the following simple relation
\begin{equation}\label{eq:rec_for_somos}
  w_{1}w_{-1}=c_2'+\frac{c_3'}{w_0}\;
\end{equation}
on the generic level surface $(H_2'=c_2')\cap(H_3'=c_3')$, which is also birational with $\C^2$. It defines a $2D$
map $(w_{-1},w_0)\mapsto (w_0,w_1)$ which is a multiplicative QRT map with $c_2'$ and $c_3'$ as parameters 
(being of type (II) in \cite{rcgo}). To
get an invariant for this map, we eliminate $w_{1}$ between \eqref{eq:rec_1} and \eqref{eq:rec_2}, to get
\beq\label{biquadc} 
c_1'w_{0}w_{-1}=(w_{0}+w_{-1})w_{0}w_{-1}+c_2'(w_{0}+w_{-1})+c_3'.
\eeq 
It leads upon division by $w_0w_{-1}$ to the
following explicit formula for the invariant:
\beq\label{s5curve} 
  H_1'=  w_{0}+w_{-1}+c_2'\left(\frac1{w_{0}}+\frac1{w_{-1}}\right)+\frac{c_3'}{w_{0}w_{-1}}\;.
\eeq
Also, the tau function substitution
\begin{equation*}
  w_{n}=\frac{\tau_{n}\tau_{n+3}}{\tau_{n+1}\tau_{n+2}}
\end{equation*}
in \eqref{eq:rec_for_somos}, which we now write in the form of the recursion relation
$w_{n+1}w_{n-1}=c_2'+{c_3'}/{w_n}$, yields the general form of the   Somos-5 recursion relation, namely 
\beq\label{s5}
  \tau_{n+5}\tau_n =c_2'\tau_{n+4}\tau_{n+1}+c_3'\tau_{n+3}\tau_{n+2}\;,\qquad n\in\Z\;.
\eeq
In our previous work \cite{hones5} we showed how to solve the initial value problem for \eqref{s5} explicitly in
terms of the 
 sigma function, but in Section \ref{par:hankel} below we will show how it can also be
solved in terms of Hankel determinants, using the S-fraction \eqref{eq:F1}.

Note that, apart from (\ref{eq:rec_for_somos}), the maps (\ref{eq:g=1_red_H3}) and   (\ref{eq:g=1_red_H2}) are also examples of symmetric QRT maps  \cite{qrt}, 
and  the orbits 
of all three maps can be identified by restricting to particular level sets of their invariants, which is a common feature of families of these maps \cite{ir2}. To see how the orbits of these different 2D maps coincide, it is necessary  to identify the parameters and values of the invariants 
(\ref{biquadH2}), (\ref{biquadH3}) and (\ref{s5curve}) in an appropriate way, from which it can be seen that (on fixed level sets) $H_2$, $H_3$ and $H_1'$ define an identical biquadratic curve in the 
$(w_{-1},w_0)$ plane, namely  (\ref{biquadc}), whose coefficients can be rewritten in terms of $c_1,c_2,c_3$ using (\ref{eq:c_to_c'}). Moreover, this curve is birationally equivalent 
to the Weierstrass cubic (\ref{weicub}). 
\end{exa}

\begin{exa}\label{exa:g=2_rec} %
Using \eqref{eq:rec_gen} when $g=2$,  on $M_2^c$ (which is birational
to $\C^5$) the formula for the Volterra map takes 
the form
\begin{equation}\label{eq:rec_genus2}
  w_{1}(\p{2,2}+\p{1,2})=w_{0}(\p{0,2}+\p{-1,2})\;,
\end{equation}
so we need to express $\p{n,2}$ for $n=-1,\dots,2$ in terms of the variables $w_i$. To do this, we use
\eqref{eq:ck}, keeping in mind that $\p{n,1}=2w_{n}+c_1/2$ for all $n$:
\begin{align}\label{eq:p2_g=2}
  2\p{n,2}&=-\p{n,1}^2+2w_{n}(\p{n+1,1}+2\p{n,1}+\p{n-1,1})+c_2\nonumber\\
                 &=-\left(2w_{n}+\frac{c_1}2\right)^2+4w_{n}(w_{n+1}+2w_{n}+w_{n-1}+c_1)+c_2\nonumber\\
                 &=4w_{n}\left(w_{n+1}+w_{n}+w_{n-1}+\frac{c_1}2\right)+c_2-\frac{c_1^2}4\;.
\end{align}
Substituted in \eqref{eq:rec_genus2} and slightly reordering the terms, we get the following symmetric relation:
\begin{align}\label{eq:rec_genus2_1}
  &\quad w_1\left(2w_{2}(w_{3}+w_{2})+2w_1(w_1+w_{0})+4w_{2}w_1+(w_{2}+w_1)c_1+c_2-\frac{c_1^2}4\right)
  \nonumber\\
 =\ &w_{0}\left(2w_{0}(w_{1}+w_{0})+2w_{-1}(w_{-1}+w_{-2})+4w_{0}w_{-1}+(w_{0}+w_{-1})c_1+c_2-\frac{c_1^2}4\right)\;.
\end{align}
This defines a $5D$ map $(w_{-2},w_{-1},w_0,w_1,w_2)\mapsto (w_{-1},w_0,w_1,w_2,w_3)$, where $w_3$ is computed from
\eqref{eq:rec_genus2_1}.  Using the first equation in \eqref{PQrec2}, the above formulae for $\p{n,1}$ and
$\p{n,2}$ lead at once to the following expressions for the coefficients of $\cQ_n$:
\begin{align}\label{eq:q_g=2}
  \q{n,1}&=2w_{n+1}+2w_{n}+c_1\;,\nonumber\\
  \q{n,2}&=2(w_{n}w_{n-1}+w_{n+1}w_{n+2})+2(w_{n+1}+w_{n})^2+(w_{n+1}+w_{n})c_1+c_2-\frac{c_1^2}4\;.
\end{align}
The formulae for $\r{n,k}$ then follow from $\r{n,1}=-2w_{n}$ and $\r{n,k}=-w_{n}\q{n-1,k-1}$ for $k>1$, by applying 
the third equation in \eqref{PQrec2}. With these formulae we can express the invariants $H_3,\dots,H_5$ in
terms of the variables $w_i$. We write this out for $H_3$, in order to find a recursion relation of order 4. According to
\eqref{eq:hams_g=1}, and using the above expressions for $\r{n,k}$,
\begin{align*}
  H_3&=2p_1p_2+2r_3+q_1r_2+q_2r_1=\p{0,1}\p{0,2}+2\r{0,3}+\q{0,1}\r{0,2}+\q{0,2}\r{0,1}\\
          &=2\p{0,1}\p{0,2}-w_{0}(\q{0,1}\q{-1,1}+2\q{-1,2}+2\q{0,2})\;,
\end{align*}
which can be written completely in terms of the variables $w_i$ using \eqref{eq:p2_g=2}, \eqref{eq:q_g=2} and the
fundamental formula $\p{n,1}=2w_{n}+c_1/2$. After some simplification, we get on the hypersurface $H_3=c_3$,
which is birational to $\C^4$,
\begin{align*}
  &w_{2} w_{1} w_{0} +w_{0} w_{-1} w_{-2}+2w_{0}^2( w_{1} +w_{-1})+
    w_{0}^3+w_{0} (w_{1}^2+ w_{1} w_{-1} +w_{-1}^2)  \\ 
    &\quad+\frac{c_1}2 w_{0} (w_{1} +w_{0}+ w_{-1})+\frac12\left(c_2-\frac{c_1^2}4\right) w_{0}+
      \frac{c_3}4-\frac{c_1c_2}8+\frac{c_1^3}{32}=0\;.
\end{align*} 
This is exactly the equation \eqref{pivmap} defining the $4D$ map (P.iv), after setting 
$n=-2$  
and
\beq\label{g2nuab}
  \nu=\frac{c_1}2\;,\quad a=\frac{c_3}4-\frac{c_1c_2}8+\frac{c_1^3}{32}\;,\quad b=\frac12\left(c_2-\frac{c_1^2}4\right)\;.
\eeq
The invariants $H_4$ and $H_5$ yield the invariants for (P.iv), given in \eqref{pivh1} and \eqref{pivh2}.
\end{exa}

\subsection{Hankel determinant solutions}\label{par:hankel} \label{sec:hankel} 

The function $F_0$ in \eqref{eq:F1} that defines the S-fraction admits a series expansion in $x$ around $(0,1)\in\Gamma_f$
that we shall use to give explicit solutions to the recurrence relation defined by the Volterra map. For a generic
point $(\cP,\cQ,\cR)$ in $\mu^{-1}(f)$, with $f$ as in \eqref{weier}, we introduce new variables
$s_1,s_2,\dots$ by writing
\begin{equation}\label{eq:F1_bis}
  1-\cfrac{w_1x}{ 1 -\cfrac {w_2x}{ 1-\cfrac{w_3x}{1- \cdots} } }=1-\sum_{j=1}^\infty s_j x^j =1-S(x)\;,
\end{equation}
where the latter equality defines the power series $S(x)$, which can be regarded as a generating function for the moments $s_j$. 
The moments can be defined from the integral 
$$ 
s_j = \frac{1}{2\pi \ri}\oint \frac{(1-F_0)}{x^{j+1}}\,\rd x , 
$$
for a sufficiently small contour around the point $(0,1)$ on $\Gamma_f$, and this leads to a 
linear functional (defined on polynomials in $x^{-1}$), and the connection with the classical theory of orthogonal polynomials \cite{shohat}, 
but we shall not pursue this further here.   
In view of \eqref{eq:F1}, $F_0=1-S(x)$ in the sense that
the latter is a Taylor series expansion of the rational function $F_0$ at $(0,1)$.  It was shown by Stieltjes
\cite{stieltjes} that the variables $w_i$ can be expressed as Hankel determinants of the
variables~$s_i$. Precisely, he showed that
\begin{equation}\label{hankelform} 
w_n = \frac{\Delta_{n-3}\Delta_n}{\Delta_{n-2}\Delta_{n-1}} \qquad \mathrm{for}\,\,n\geqslant1\;, 
\end{equation}
where $\Delta_{2k-1}=\det (s_{i+j-1})_{i,j=1,\ldots, k}$ and  $\Delta_{2k}=\det (s_{i+j})_{i,j=1,\ldots, k}$, for
$k\geqslant1$, that is,
\begin{equation}
\Delta_{2k-1}=
\left| \begin{array}{cccc} 
s_1     & s_2 & \cdots & s_{k}   \\ 
s_2     &         & \iddots &        \vdots       \\
\vdots & \iddots &          &           \vdots    \\        
s_{k}     &  \cdots & \cdots & s_{2k-1}  
\end{array} 
\right|
\quad\hbox{and}\quad
\Delta_{2k}=
\left| \begin{array}{cccc} 
s_2    & s_3 & \cdots & s_{k+1}   \\ 
s_3     &         & \iddots &        \vdots       \\
\vdots & \iddots &          &           \vdots    \\        
s_{k+1}     &  \cdots & \cdots & s_{2k}  
\end{array} 
\right| \;.
\end{equation}  
Also, by definition, $\Delta_{-2}=\Delta_{-1}=\Delta_0=1$. For example, $w_1=s_1$, $w_2=s_2/s_1$,
$w_3=(s_1s_3-s_2^2)/(s_1s_2)$, and so on. It is clear from \eqref{eq:F1_bis} that, conversely, $s_i$ can be expressed
as a polynomial in $w_1,\dots,w_i$, for example $s_1=w_1$, $s_2=w_1w_2$, $s_3=w_1w_2(w_2+w_3)$, and so on. 
(While these expressions can be found by expanding geometric series, a systematic method using continuants 
is presented in Section \ref{par:birat_trans}.) 
\begin{thm}\label{hankel} 
The terms $w_n$, $n\geqslant1$, of the recurrence sequence defined by the Volterra map on $\mu^{-1}(f)$ can be
written as
\begin{equation*}
  w_n = \frac{\Delta_{n-3}\Delta_n}{\Delta_{n-2}\Delta_{n-1}}\;,
\end{equation*} 
where the entries of the Hankel matrices satisfy the recursion relation 
\begin{equation}\label{sjrec}
  s_j = \sum_{k=1}^g\big(\p{k}-\q{k}\big)s_{j-k}+\sum_{i=1}^{j-1}s_is_{j-i}+
  \frac12\sum_{k=1}^g\q{k}\sum_{i=1}^{j-k-1}s_is_{j-k-i}, \quad j\geqslant g+2\;.
\end{equation}
The initial values $s_1,s_2,\ldots,s_{g+1}$ and the coefficients for the recursion \eqref{sjrec} are provided by a generic triple $(\cP,\cQ,\cR)\in\mu^{-1}(f)$.
\end{thm} 
\begin{proof}
In order to prove the recursion formula \eqref{sjrec}, we will derive a quadratic formula for $S(x)=1-F_0$,
introduced in \eqref{eq:F1_bis}. We use (\ref{eq:F1}) to write 
\begin{equation}\label{eq:G_1} 
  y=-\cP(x)+F_0\cQ(x)=-\cP(x)+(1-S(x))\cQ(x)\;,
\end{equation}
which we substitute in
\begin{equation}\label{eq:curve_1}
  y^2=f(x)=\cP(x)^2+\cQ(x)\cR(x)
\end{equation}
to get the following quadratic equation for $S(x)$:
\begin{equation}\label{eq:quad_for_S}
  \big(\cQ(x)-\cP(x)\big)\, S(x)-\frac{1}{2}\cQ(x)\,S^2(x)+\cP(x) -\frac{1}{2}\cQ(x) +\frac{1}{2} \cR(x)=0\;.
\end{equation}
Substituting the power series for $S(x)$ into the quadratic, as well as the polynomials  $\cP,\cQ,\cR$, the coefficients 
of $x^j$ for $1\leqslant j\leqslant g+1$ allow the $g+1$ initial values $s_1,\ldots,s_{g+1}$ to be determined from 
these polynomials. Then, upon 
taking the coefficient of $x^j$ for $j\geqslant g+2$, the recursion relation (\ref{sjrec}) is obtained directly. Observe 
that the number of initial values plus independent coefficients appearing linearly in the recursion is $g+1+2g=3g+1=\,$dim$\,M_g$.
\end{proof}
As we will see in the examples, it is often more practical not to fix the curve $y^2=f(x)$, i.e.\ not to fix the
values of all invariants, but only fix some of them and take $w_0,\dots,w_g$ as extra initial conditions.
\begin{exa}\label{genus1hankel}
We specialize the above results to $g=1$, fixing arbitrary constants $c_1$ and $c_2$ and taking the Volterra map
on the surface $H_1=c_1,\;H_2=c_2$, as in one of the reductions considered in Example \ref{exa:g=1_rec}. For the
recursion \eqref{sjrec} we can take $p_1,q_1$ as initial conditions, since given $p_1$ and $q_1$, specifying the
values of $H_1$ and $H_2$ is equivalent to specifying the values of $r_1$ and $r_2$ (see the explicit formulas for
$H_1$ and $H_2$ in the left column of \eqref{eq:hams_g=1}). It follows from \eqref{eq:pg1} that
\begin{equation*}
  p_1-q_1=-2w_1-\frac{c_1}2\;,\quad\hbox{and}\quad
  q_1=2\left(w_1+w_0+\frac{c_1}2\right)=-2w_2+\frac1{2w_1}\left(\frac{c_1^2}4-c_2\right)\;,
\end{equation*}
where we obtained the last equality by using the recursion relation \eqref{eq:g=1_red_H2}, shifted by 1, to replace
$w_0$ by $w_2$. Substituted in \eqref{sjrec}, we get the following recursive formula for $s_j$ ($j\geqslant3$) in
terms of $w_1,w_2$, which we can take as initial conditions, instead of $p_1$ and $q_1$:
\begin{equation}\label{eq:rec_s_g=1}
  s_j=\left(-2w_1-\frac{c_1}2\right)s_{j-1}+
  \sum_{i=1}^{j-1}s_is_{j-i}+\left(\frac1{4w_1}\left(\frac{c_1^2}4-c_2\right)-w_2\right)\sum_{i=1}^{j-2}s_is_{j-1-i}\;.
\end{equation}
Then specifying the two initial values $w_1,w_2$ fixes the  initial conditions $s_1,s_2$ for the above, as $s_1=w_1$ and $s_2=w_1w_2$.
Notice that $s_j$ is a polynomial of degree $j$ in $w_1,w_2$, with $w_1|s_j$ for all $j$. 

\smallskip

As a concrete example, consider the curve $y^2=1-10x+29x^2-24x^3$, with initial conditions $s_1=1$, $s_2=2$ (or,
equivalently, $w_1=1$, $w_2=2$). Since $c_1=-10$ and $c_2=29$, the recursion \eqref{eq:rec_s_g=1} becomes
$$ 
  s_j = 3s_{j-1}+ \sum_{i=1}^{j-1}s_is_{j-i}-3\sum_{i=1}^{j-2}s_is_{j-i-1}, \quad j\geqslant 3, 
$$
which generates the sequence 
$$ 
(s_j)_{j\geqslant 1}: \quad 1,2,7,27, 109, 456, 1969, 8746, 39825, 185266, \ldots, 
$$ 
producing  
$
\Delta_1=1$, 
$\Delta_2=2$, 
$$ 
\Delta_3=\left| \begin{array}{cc} 1 & 2 \\ 2 & 7
\end{array}\right|=3,\,  
\Delta_4=\left| \begin{array}{cc} 2 & 7 \\ 7 & 27
\end{array}\right|=5,\,
\Delta_5=\left| \begin{array}{ccc} 1 & 2  & 7 \\ 2 & 7 & 27 \\ 7 & 27 & 109
\end{array}\right|=11,\,  
\Delta_6=\left| \begin{array}{ccc} 2 & 7 & 27 \\ 7 & 27 & 109 \\ 27 & 109 & 456
\end{array}\right|=37,\ldots, 
$$
which extends symmetrically to all $n\in\Z$ to produce the original Somos-5 sequence \cite{oeis}, 
$$ 
\ldots, 3,2, 
1,1,1,1,1,2,3,5,11,37,83,274,1217,6161,22833,\ldots, 
$$ 
generated by the bilinear recurrence 
\beq\label{s5orig}
\Delta_{n+5}\Delta_n=\Delta_{n+4}\Delta_{n+1}+\Delta_{n+3}\Delta_{n+2}. 
\eeq
It is a particular case of \eqref{s5} with $c_2'=c_3'=1$, where the latter are computed from $c_1=-10,$ $c_2=29$,
$c_3=-24$, using \eqref{eq:c_to_c'}. 
\end{exa}

\begin{remark}
Note that Hankel determinant formulae for Somos-5 were previously obtained in the work 
of Chang, Hu and Xin, using a bilinear B\"acklund transformation for Somos-4. 
We will return to the connection with Somos-4 in Section \ref{sec:continuous}, 
but for now we just point out that the Hankel determinant expressions found in  \cite{chang} are more complicated than the above, 
because two different moment sequences are required for the terms with even/odd indices.  
In particular, for the original Somos-5 sequence, there 
are two sequences of moments, namely $(\bar{s}_j)_{j\geqslant 0}: \, 1,-1,4,-8,25,-65,197,-571,1753,-5351,16746,-52626, \ldots$, and 
$(\hat{s}_j)_{j\geqslant 0}:\, 
2,-1,3,-1,12,2,61,39,352,374,2210,3162, \ldots$, which are defined by 
$$ 
\bar{s}_0=1, \,\bar{s}_1=-1,\, \bar{s}_{j+1}=-\bar{s}_j+2\bar{s}_{j-1}+\sum_{i=0}^{j-1} \bar{s}_i\bar{s}_{j-i-1},  \quad\mathrm{and}\quad
\hat{s}_0=2, \,\hat{s}_1=-1, \,\hat{s}_{j+1}=\hat{s}_j+\sum_{i=0}^{j-1} \hat{s}_i\hat{s}_{j-i-1},  
$$
respectively,  where the recursions hold for $j\geq 1$,  
and (with the indexing convention of Theorem 1.2 in \cite{chang}) the terms of 
the Somos-5 sequence \cite{oeis} are then given by 
$\rS_0=1,\rS_1=1,\rS_2=\bar{s}_0=1,\rS_3=\hat{s}_0=2$, and 
$$ 
\rS_4=\left| \begin{array}{cc} 1 & -1 \\ -1 & 4
\end{array}\right|=3,\,  
\rS_5=\left| \begin{array}{cc} 2 & -1 \\ -1 & 3
\end{array}\right|=5,\,
\rS_6=\left| \begin{array}{ccc} 1 & -1  & 4 \\ -1 & 4 & -8 \\ 4 & -8 & 25
\end{array}\right|=11,\,  
\rS_7=\left| \begin{array}{ccc} 2 & -1 & 3 \\ -1 & 3 & -1 \\ 3 & -1 & 12
\end{array}\right|=37,\ldots, 
$$
which are not related to the determinants in Example \ref{genus1hankel} in a 
straightforward manner.
\end{remark}

\begin{exa}\label{genus2hankel}
We now specialize the above results to $g=2$, thereby continuing Example \ref{exa:g=2_rec}. From (\ref{g2nuab}) 
it is clear that fixing the values $c_1,c_2,c_3$ is equivalent to specifying the parameters 
$a,b$ and $\nu$, 
which we fix, since these are the coefficients 
in the recurrence relation (P.iv), and we can take $w_0,w_1,w_2,w_3$ (or $w_1,w_2,w_3,w_4$) as initial data for the latter. 
We now have 
$$ 
s_1=w_1, \quad s_2=w_1w_2, \quad s_3 = w_1w_2(w_2+w_3),
$$
providing the 3 
initial values for the
recursion \eqref{sjrec}, which takes the form
\begin{equation}\label{g2rec}
   s_j = \hat\alpha\,s_{j-1}+\hat\beta\, s_{j-2}+\sum_{i=1}^{j-1}s_is_{j-i} +\hat\gamma \, \sum_{i=1}^{j-2}s_is_{j-i-1} +
   \hat\delta\,\sum_{i=1}^{j-3}s_is_{j-i-2}\;, \quad j\geqslant 4\;.
\end{equation}
While $s_1,s_2,s_3$ are determined by $w_1,w_2,w_3$ only, $w_0$ and $a,b,\nu$ are required to find 
the coefficients $\hat\al,\dots,\hat\delta$, which are computed using $\p{1}=2w_0+\nu$ (recall that $\nu=c_1/2$) and
\eqref{eq:p2_g=2}, \eqref{eq:q_g=2} for $n=1$, to give 
\begin{align}\label{abgam}
  \hat{\alpha} &=\p{1}-\q{1}= -2w_1-\nu\;,\nonumber\\
  \hat{\beta} &=\p{2}-\q{2}= -2w_1(w_2+w_1+w_0+\nu)-b\;,\nonumber\\
  \hat{\gamma}&=\frac12\q{1}= w_1+w_0+\nu\;, \\ 
  \hat{\delta}&= \frac12\q{2}=w_{-1}w_0+w_1w_2+(w_0+w_1)^2+\nu(w_0+w_1)+b\nonumber\\
  &=-\Big(w_2w_3+w_1w_2+w_2^2+w_0w_2+\nu w_2+\frac{a}{w_1}\Big)\;,\nonumber 
\end{align}
where, in the last equality, we have used the recurrence relation \eqref{pivmap} to replace $w_{-1}$ by $w_3$. If desired, one can apply \eqref{pivmap} once again to 
replace $w_0$ by $w_4$ in the above expressions, but we have not done this.

As a particular numerical example, 
we take the rational orbit (\ref{rational_orbit}) of (P.iv) considered in Section 
\ref{sec:singularity}. Upon fixing $w_0=w_1=w_2=w_3=1$ and $\nu=3$, $a=5$, 
$b=7$, 
we see that for $j\geq 4$  
the recursion (\ref{g2rec})  becomes 
$$
s_j = -5\,s_{j-1}-19 \, s_{j-2} 
+\sum_{i=1}^{j-1}s_is_{j-i}
+5 \, \sum_{i=1}^{j-2}s_is_{j-i-1}
-12 \,\sum_{i=1}^{j-3}s_is_{j-i-2}\;, 
$$
and the three initial values $s_1=s_2=1$, $s_3=2$ lead to the following moment %
sequence and Hankel determinants: 
$$(s_j)_{j\geqslant 1}: \, \,\,
1,1,2, -26,45,11,-116,553,1151,-26727,108897,-169157,-310959,3004412,-4722005,\ldots,$$
\begin{align*} 
\Delta_1=&1, \,
\Delta_2=1, \,
\Delta_3=\left| \begin{array}{cc} 1 & 1 \\ 1 & 2
\end{array}\right|=1,  \, 
\Delta_4=\left| \begin{array}{cc} 1 & 2 \\ 2 & -26
\end{array}\right|=-30, \, 
\Delta_5=\left| \begin{array}{ccc} 1 &1 & 2   \\ 1&  2 &  -26 \\ 2 & -26 & 45
\end{array}\right|=-743,\\  
\Delta_6= &\left| \begin{array}{ccc} 1 &  2 & -26 \\ 2 & -26 & 45 \\ -26 & 45 & 11
\end{array}\right|=10541, \qquad
\Delta_7=\left| \begin{array}{cccc} 1 &1 & 2 & -26   \\ 1&  2 &  -26 & 45 \\
 2 & -26 & 45 & 11 \\ 
-26 & 45 & 11 &-116 
\end{array}\right|=127318,\,\dots
.
\end{align*}
This reproduces the sequence of tau functions in Example \ref{experitau}, if we identify 
$\Delta_{n-3}=\tau_{n}$ in (\ref{rattauseq}).

\end{exa}

The preceding explicit form of the recursion for the entries of the Hankel determinants when $g=2$ yields a 
simple proof of the Laurent property for (\ref{taurec}). 

\begin{prfr} 
For $n\geqslant 1$ we have 
$$ 
  w_n = \frac{\tau_n\tau_{n+3}}{\tau_{n+1}\tau_{n+2}}=\frac{\Delta_{n-3}\Delta_n}{\Delta_{n-2}\Delta_{n-1}}\;, 
$$
where $\tau_n$ satisfies (\ref{taurec}). Hence the tau functions are given by Hankel determinants, up to a shift of
index and a gauge transformation of the form (\ref{gauge}), with a different scaling for even/odd $n$. Comparing
with the values $\Delta_{-2}=\Delta_{-1}=\Delta_0=1$, we see that the relation between the two sequences is
\beq\label{tauhankel}
  \tau_{2k+1}=\tau_1\, \left(\frac{\tau_3}{\tau_1}\right)^k\, \Delta_{2k-2}, \qquad
  \tau_{2k+2}=\frac{\tau_1\tau_2}{\tau_3}\, \left(\frac{\tau_3}{\tau_1}\right)^{k+1}\, \Delta_{2k-1}
\eeq
for $k\geqslant 0$.  Recall that $\cal R$ denotes the ring formed of Laurent polynomials in $\tau_2,\tau_3,\tau_4$
and polynomials in $\tau_0,\tau_1,\tau_5,\tau_6$ with coefficients in $\Z[a,b,\nu]$. Upon rewriting the formulae
(\ref{abgam}) for the coefficients in (\ref{g2rec}) in terms of the 7 initial tau functions, we see that
$\hat{\al}$, $\hat{\be}\in{\cal R}$, but (due to the presence of terms involving $w_0$ and $1/w_1$), $\hat{\gam}$
and $\hat{\delta}$ both have $\tau_1$ appearing in the denominator, so instead $\hat{\gam}$,
$\hat{\delta}\in\tau_1^{-1}{\cal R}$. However, the three initial values are
$$
  s_1=\frac{\tau_1\tau_4}{\tau_2\tau_3}, \quad 
  s_2= \frac{\tau_1\tau_5}{\tau_3^2}, \quad 
  s_3 = \tau_1\left(\frac{\tau_2\tau_5^2}{\tau_3^3\tau_4}+\frac{\tau_6}{\tau_3\tau_4}\right), 
$$ 
so $s_j\in\tau_1{\cal R}$ for $j=1,2,3$. Then by induction, since $\hat{\gam}$ and $\hat{\delta}$ appear in front
of terms of degree 2 in $s_i$ in the recursion (\ref{g2rec}), it follows that $s_j\in\tau_1{\cal R}$ for all
$j\geqslant 1$.  Then since $\Delta_{2k-1}$ and $\Delta_{2k}$ are $k\times k$ determinants, a factor of $\tau_1$
can be taken out of each row (or column), so they are each given by an overall factor of $\tau_1^k$ times an
element of ${\cal R}$. Thus the powers of $\tau_1$ exactly cancel in (\ref{tauhankel}), and hence $\tau_n\in{\cal
  R} $ for all $n\geqslant 1$.
\end{prfr}  

\begin{remark} {\bf Hankel determinant formula for negative indices:} 
As was previously noted, the Laurent property for negative indices $n$ follows automatically from the birationality and reversing symmetry of  (\ref{taurec}), but it
can also be shown directly from an appropriate extension of (\ref{tauhankel}) to $k<0$. In fact, for any $g$ there is a version of the Hankel determinant 
formula (\ref{hankelform})  which is valid for $n\leqslant 0$. Indeed, the Volterra map arises from the S-fraction expansion  (\ref {eq:F1}) of the function $F_0$, 
based on the power series $S(x)$, with $F_0=1-S$, as in (\ref{eq:F1_bis}), but more precisely this is the expansion around the  
point $(0,1)$ on the curve $\Gamma_f$ given by (\ref{weier}), with $x$ being a local parameter. 
The inverse Volterra map arises from another 
S-fraction,  associated with a power series $S^*(x)$, corresponding to the expansion of the same function $F_0$ around the  
point $(0,-1)\in\Gamma_f$, that is 
\begin{equation}\label{eq:F1_neg}
 F_0= \cfrac{w_0x}{ 1 -\cfrac {w_{-1}x}{ 1-\cfrac{w_{-2} x}{1- \cdots} } }=\sum_{j=1}^\infty s_j^* x^j =:S^*(x)\;. 
\end{equation}
Then the extension of (\ref{hankelform}) to non-positive values of the index is 
\begin{equation}\label{hankelneg} 
w_n = \frac{\Delta^*_{-n-2}\Delta^*_{-n+1}}{\Delta^*_{-n}\Delta^*_{-n-1}} \qquad \mathrm{for}\,\,n\leqslant0\;, 
\end{equation}
where $\Delta^*_{2k-1}=\det (s^*_{i+j-1})_{i,j=1,\ldots, k}$ and $\Delta^*_{2k}=\det (s^*_{i+j})_{i,j=1,\ldots, k}$
for $k\geqslant1$, with $\Delta^*_{-2}=\Delta^*_{-1}=\Delta^*_{0}=1$.  Mutatis mutandis, this is proved in the same
way as Theorem \ref{hankel}, and the moments $s_j^*$ satisfy another recursion of the form (\ref{sjrec}). The two
sets of Hankel determinants combine to produce a single sequence of tau functions $(\tau_n)_{n\in\Z}$, consistently
defined by taking $\tau_n=\Delta_{n-3}$ for $n\geqslant 1$, and $\tau_n=\Delta^*_{-n+1}$ for $n\leqslant 3$.
\end{remark}

\subsection{The birational map $M_g\to\C^{3g+1}_w$}\label{par:birat_trans}

Elaborating further on the S-fraction of $F=\frac{y+\cP(x)}{\cQ(x)}$, we construct a birational map between $M_g$
and $\C^{3g+1}_w$, where the latter stands for $\C^{3g+1}$ equipped with $w_1,\dots,w_{3g+1}$ as affine 
coordinates. 
We may call it the
\emph{unreduced} birational map, in view of the birational map $M_g^c\to\C^{2g+1}_w$ which we obtained by reduction (fixing $c=(c_1,\ldots,c_g)$). This unreduced map turns out to be less convenient for deriving
the Volterra map in the $w_i$ coordinates, but it is nevertheless useful for obtaining abstract results 
in these coordinates.

We start from the equation \eqref{eq:quad_for_S}, which we can view as a linear relation for $\cP(x),\cQ(x)$ and $\cR(x)$. 
It amounts to an infinite linear system of equations for the coefficients 
$\p1,\dots,\p g,\q1,\dots,\q g,\r1,\dots,\r{g+1}$ of these polynomials, 
given in terms of the coefficients $s_j$  of the power series $S(x)$ defined in \eqref{eq:F1_bis}.
We show that its solution can be expressed rationally in terms of
$s_1,\dots,s_{3g+1}$, and hence in terms of $w_1,\dots,w_{3g+1}$, yielding the birational map. To do this, we
investigate the first $3g+1$ equations only, namely the ones corresponding to the coefficient of $x^j$ in \eqref{eq:quad_for_S}.
For convenience, we set  
$$\rho_k:=\sum_{i+j=k}s_is_j,$$ for $k\geqslant2$, so that  
  $S^2=\sum_{i+j=k}s_is_jx^k=\sum_{k=2}^\infty\rho_kx^k$, 
and then consider the terms of \eqref{eq:quad_for_S} at each order in $x$. 
At leading order, the constant term cancels, and the  coefficients of $x,x^2,\dots,x^{g+1}$ can be used to obtain $r_1,\dots,r_{g+1}$ in terms of $\p1,\dots,\p
g,\q1,\dots,\q g,$ and $s_j$ for $1\leqslant j\leqslant g+1$, so it is sufficient to show that we can solve the coefficients of $x^j$ 
with ${g+2}\leqslant j\leqslant{3g+1}$ for $\p1,\dots,\p g,\q1,\dots,\q g$ in terms of
$s_1,\dots,s_{3g+1}$. For each such $j$, the equation at order $x^j$ in  \eqref{eq:quad_for_S} is given by (\ref{sjrec}), 
which we can rewrite as 
\begin{equation}\label{eq:linear_system}
  \sum_{i=1}^{g} s_{j-i}(\p i-\q i)+\frac12\sum_{i=1}^{g}\rho_{j-i}\q i =s_j-\rho_j\;,\qquad \mathrm{for} \quad g+2\leqslant j\leqslant 3g+1\;.
\end{equation}
We view \eqref{eq:linear_system} as a linear system in the $2g$ variables $\p j-\q j$ and $q_j/2$, with
$i=1,\dots,g$. In matrix form, it is written as
\begin{equation*}
  \begin{pmatrix}
    s_{g+1}&\cdots&s_2&\rho_{g+1}&\cdots&\rho_2\\
    \vdots&\ddots&\vdots&\vdots&\ddots&\vdots\\
    s_{2g}&\cdots&s_{g+1}&\rho_{2g}&\cdots&\rho_{g+1}\\
    s_{2g+1}&\cdots&s_{g+2}&\rho_{2g+1}&\cdots&\rho_{g+2}\\
    \vdots&\ddots&\vdots&\vdots&\ddots&\vdots\\
    s_{3g}&\cdots&s_{2g+1}&\rho_{3g}&\cdots&\rho_{2g+1}
  \end{pmatrix}
  \begin{pmatrix}
    p_1-q_1\\
    \vdots\\
    p_g-q_g\\
    q_1/2\\
    \vdots\\
    q_g/2
  \end{pmatrix}
  =
  \begin{pmatrix}
    s_{g+2}-\rho_{g+2}\\
    \vdots\\
    s_{2g+1}-\rho_{2g+1}\\
    s_{2g+2}-\rho_{2g+2}\\
    \vdots\\
    s_{3g+1}-\rho_{3g+1}
  \end{pmatrix}
\end{equation*}  
so that the matrix which governs the linear system has the following $2\times 2$ block form, each block being a
Toeplitz matrix,
\begin{equation}\label{eq:block_matrix}
  \begin{pmatrix}
    T_g(s,g)&T_g(\rho,g)\\
    T_g(s,2g)&T_g(\rho,2g)\\
  \end{pmatrix}
  \quad\hbox{where}\quad
  T_g(\sigma,k):=
  \begin{pmatrix}
    \sigma_{k+1}&\sigma_k&\dots&\sigma_{k-g+2}\\
    \sigma_{k+2}&\sigma_{k+1}&\ddots\\
    \vdots&\ddots&\ddots&\sigma_k\\
    \sigma_{k+g}&&\sigma_{k+2}&\sigma_{k+1}
  \end{pmatrix}
\end{equation}
Our claim that $\p1,\dots,\p g,\q1,\dots,\q g$ can be solved rationally in terms of $s_1,\dots,s_{3g+1}$ then
follows from the fact that its determinant, which is a polynomial in $s_1,\dots,s_{3g+1}$, is non-zero. To show
that the determinant is non-zero it is sufficient to show it for one particular set of values of the $s_i$. We pick
$s_i=1$ for $i=g$ and $i=g+1$ and $s_i=0$ for all other values of $i$. Then $\rho_i=1$ for $i=2g$ and $i=2g+2$,
$\rho_i=2$ for $i=2g+1$ and $\rho_i=0$ for all other values of $i$. With this choice, $T_g(s,g)$ is upper
triangular, with all diagonal entries equal to $1$, $T_g(s,2g)$ is the zero matrix and $T_g(\rho,2g)$ is a
tridiagonal matrix, with all diagonal entries equal to $2$ and all superdiagonal and subdiagonal entries equal to
$1$. It is easily verified that its determinant is $g+1$, hence non-zero, showing our claim.

\begin{exa}
The simplest example is the case of $g=1$. Then $3g+1=4$. The birational map between $s_1,\dots,s_4$ and
$w_1,\dots,w_4$ is given by
\begin{equation}\label{eq:s_to_w_g=1}
  s_1=w_1\;,\quad s_2=w_1w_2\;,\quad s_3=w_1w_2(w_2+w_3)\;,\quad s_4=w_1w_2\left((w_2+w_3)^2+w_3w_4\right)\;.
\end{equation}
The linear system for $\p1,\q1,\r1,\r2$ then corresponds to the coefficients of $x^j$ in 
\eqref{eq:quad_for_S} with $j=1,\dots,4$. They are explicitly given by
\begin{align}\label{eq:lin_sys_g=1}
  &\q1-2\p1-2s_1-r_1=0\;,\nonumber\\
  &2\p1s_1-2\q1s_1-2s_2+2\rho_2-\r2=0\;,\nonumber\\
  &2\rho_3-2s_3+(\rho_2-2s_2)\q1+2s_2\p1=0\;,\\
  &2\rho_4-2s_4+(\rho_3-2s_3)\q1+2s_3\p1=0\;.\nonumber
\end{align}
As in the general case, we first solve the last $2g=2$ equations for $\p1,\q1$, which gives, upon writing each $\rho_k$
in terms of the $s_j$,
\begin{align*}
  p_1&=2\frac{-3s_1^2s_2^2-2s_2^3+2(s_1s_2+s_1^3-s_3)s_3+(2s_2-s_1^2)s_4}{s_1(2s_2^2-s_1s_3)}
  =-w_1-w_2+w_3+\frac{w_3(w_1-w_4)(w_1-2w_2)}{w_1(w_2-w_3)}\;,\\
  q_1&=2\frac{s_2s_4-s_3^2-s_2^3}{s_1(2s_2^2-s_1s_3)}=-2w_2\left(1+\frac{w_3(w_1-w_4)}{w_1(w_2-w_3)}\right)\;,
\end{align*}
where we have used \eqref{eq:s_to_w_g=1} to write the formulae in terms of $w_1,\dots,w_4$. The first two equations
can now be solved for $r_1$ and $r_2$, leading to the following expressions:
\begin{equation*} 
  r_1=-2w_3\left(1+\frac{(w_1-w_2)(w_1-w_4)}{w_1(w_2-w_3)}\right)\;,\qquad
  r_2=2w_1w_3\left(1+\frac{w_1-w_4}{w_2-w_3}\right)\;.
\end{equation*}
Note also that in this case the unreduced version of the Volterra map is the birational map of $\C^{4}_w$ defined by 
\begin{align}\label{wmapg1}
&  
(w_{1},w_2,w_{3},w_{4})\mapsto (w_2,w_{3},w_{4},w_{5})\;,\nonumber\\
& \frac{w_{5}w_{4}+w_{4}^2-w_{3}^2-w_{3}w_2}{w_{4}-w_{3}} 
+  \frac{w_{1}w_{2}+w_{2}^2-w_{3}^2-w_{3}w_{4}}{w_{3}-w_2} =0\;,
\end{align} 
which (after replacing each $w_i\to w_{i-2}$) is the relation obtained by eliminating $c_1$ from (\ref{eq:rec_g1}). 
\end{exa}

For any $g$, we can describe an explicit algorithmic procedure for obtaining the exact expressions for the birational map between $s_1,\ldots,s_{3g+1}$ 
and $w_1,\ldots,w_{3g+1}$. One way round, this is given by the Hankel determinant formula (\ref{hankelform}), but to describe the inverse map 
more explicitly we must recall 
how this relates to the 
approximation problem originally considered by Stieltjes \cite{stieltjes}, 
as well as 
other classical results on convergents of continued fractions in terms of continuants (for a concise review of the latter, see \cite{duv}). 

The convergents  of the S-fraction (\ref{eq:F1}) are the sequence of rational functions of $x$ obtained by truncating the continued fraction at some finite line $n$, 
which approximate the series for $F_0=1-S$ exactly up to and including the coefficient of $x^n$, that is 
\beq\label{sapprox}
 {\cal F}_n(x):= 
  1-\cfrac{w_1x}{ 1 -\cfrac {w_2x}{ 1-\cfrac{\cdots}{1- w_nx} } }=1-S(x) +\mathcal O(x^{n+1}), 
\eeq 
Thus, for the first few convergents  we have 
$$ 
{\cal F}_0=1, \quad {\cal F}_1=1-w_1x, \quad {\cal F}_2=\frac{1-(w_1+w_2)x}{1-w_2x}, 
\quad {\cal F}_3=\frac{1-(w_1+w_2+w_3)x+w_1w_3x^2}{1-(w_2+w_3)x}, 
$$ 
and so on. By the usual correspondence between convergents and $2\times2$ matrices, we see that the 
monodromy product over the conjugation matrices  (\ref{LM}) appearing in the discrete Lax equation is given 
by 
\beq\label{monod}  {\boldsymbol{\Phi}}_{n+1}:=
\mma_1(x)  \mma_2(x)\cdots  \mma_{n+1}(x)=\left( \begin{array}{cc} {\cal K}_{n+1}(w_1,\ldots, w_n; x)  &  
-w_nx\,{\cal K}_{n}(w_1,\ldots, w_{n-1}; x)
\\
{\cal K}_{n}(w_2,\ldots, w_{n}; x) & -w_nx\,{\cal K}_{n-1}(w_2,\ldots, w_{n-1}; x) 
\end{array} \right) , 
\eeq 
with the $n$th convergent being the ratio of the entries in the first column, that is 
$$ 
{\cal F}_n(x)= \frac{{\cal K}_{n+1}(w_1,\ldots, w_n; x)}{{\cal K}_{n}(w_2,\ldots, w_{n}; x)}, 
$$ 
where the polynomial ${\cal K}_n$ is a continuant of size $n$, that is 
$$ 
{\cal K}_{n}(w_1,\ldots, w_{n-1}; x):= \left|\begin{array}{rrrrr} 
{1} & {-1} & {0} & {\dots} & {0} \\ 
{ -w_1 x} & {1} & {-1} & {\ddots} & {\vdots} \\ 
	{0}  & {-w_2 x} & {1} & {\ddots} & {0} \\ 
{\vdots} & & {\ddots} & {\ddots} & {-1} \\ 
	{0} & {\dots} & {0} & { -w_{n-1}x}  & {1} 
\end{array}  
\right| 
.
$$ 
From (\ref{monod}), the continuants are generated recursively from the linear relation ${\boldsymbol{\Phi}}_{n+1}= {\boldsymbol{\Phi}}_{n}\mma_{n+1}$, 
starting from ${\boldsymbol{\Phi}}_{0}={\boldsymbol{1}}$ (the identity matrix). 
Hence, by replacing the series for $F_0=1-S$ in (\ref{sapprox}) by its truncation at the $n$th term,  then multiplying by ${\cal K}_{n}(w_2,\ldots, w_{n}; x)$ on both sides, 
we find the relation 
$$ 
\Big(1-\sum_{j=1}^n s_j x^j \Big){\cal K}_{n}(w_2,\ldots, w_{n}; x) = {\cal K}_{n+1}(w_1,\ldots, w_n; x) +{\cal O} (x^{n+1}), 
$$ 
and comparing the coefficients of $x^j$ for $j=1,\ldots,n$  allows $s_1,s_2,\ldots,s_n$ to be calculated recursively as polynomial expressions in 
$w_1,\ldots,w_n$.   For instance, when $n=4$ the expressions \eqref{eq:s_to_w_g=1} are obtained from the 
numerator and denominator of 
$$ 
{\cal F}_4 = \frac{1-(w_1+w_2+w_3+w_4)x+(w_1w_3+w_1w_4+w_2w_4)x^2}{1-(w_2+w_3+w_4)x+w_2w_4x^2}. 
$$ 
Observe that these expressions are universal, in the sense that they depend only on the structure of the S-fraction, so that 
each $s_n$ is always given by the same polynomial in $w_j$ for $1\leqslant j\leqslant n$, independent of $g$. 

For $g\geqslant 2$, 
the corresponding formulae for the coefficients of $\cP(x),\cQ(x), \cR(x)$ in terms of 
$w_1,\ldots,w_{3g+1}$ become very complicated, and we do not have a compact way to write them. 
However, we will not need the explicit form of these formulae in what follows, even when dealing with the
example of genus 2.

\section{Discrete integrability}\label{sec:adi}
\setcounter{equation}{0}

We show in this section that the map (P.iv) is a discrete a.c.i.\ system, in the sense of Definition
\ref{def:adi}. To do this, we will first show that the affine space $M_g$ ($g\geqslant1$) of \eqref{Mg} which was constructed in
the previous section has the same integrability properties as the odd and even Mumford systems (see
\cite{tata2,vanhaecke}): the invariants $H_i$ of the Volterra map are in involution with respect to a large family
of compatible Poisson structures and their generic level sets are affine parts of Jacobi varieties. We then show
that the Volterra map is a Poisson map with respect to these Poisson structures and that it is a translation on the
latter complex tori.  In Appendix A (Section \ref{par:even_mumford}) we further show the precise relation between
the Mumford-like system, introduced here, and the even Mumford system.

\subsection{Compatible Poisson structures for the Mumford-like system}
We first introduce a $g+1$-dimensional family of compatible Poisson structures on $M_g$. The family is parametrized
by the $g+1$-dimensional vector space of polynomials $\phi\in\C[x]$ of degree at most $g+1$, vanishing at $0$, so
$\phi(0)=0$. For such a polynomial $\phi$, the corresponding Poisson structure is most naturally written by viewing
$\cP,\,\cQ$ and $\cR$ as generating functions, and is defined by
\begin{align}\label{eq:mumford_poisson}
  \pb{\cP(x),\cP(y)}^\phi&=  \pb{\cQ(x),\cQ(y)}^\phi=0\;,\ \pb{\cR(x),\cR(y)}^\phi=\phi(y)\cR(x)-\phi(x)\cR(y)\;,
  \nonumber\\
  \pb{\cP(x),\cQ(y)}^\phi&=\frac{\phi(x)y\cQ(y)-\phi(y)x\cQ(x)} {x-y}\;,\ 
  \pb{\cP(x),\cR(y)}^\phi=-y\frac{\phi(x)\cR(y)-\phi(y)\cR(x)} {x-y}\;,\\
  \pb{\cQ(x),\cR(y)}^\phi&=2y\frac{\phi(x)\cP(y)-\phi(y)\cP(x)} {x-y}-\phi(y)\cQ(x)\;. \nonumber
\end{align}
Of course,  one needs to verify that the above definition is 
coherent,  in the sense that the right-hand side of each of these
formulae is indeed a polynomial in $x$ and $y$, and also that the right-hand side does not contain any monomials
$x^iy^j$ which are absent in the left-hand side. For example, $\pb{\cP(x),\cQ(y)}^\phi$ is a polynomial in $x,y$
with only non-zero coefficients of $x^iy^j$ when $1\leqslant i,j\leqslant g$, while $\phi(x)y\cQ(y)-\phi(y)x\cQ(x)$
is clearly divisible by $x-y$ and the quotient has only non-zero coefficients of $x^iy^j$ in the same range. The
same argument applies to the other formulae. Moreover, we need to verify that $\PB^\phi$ is a Poisson bracket,
i.e., that it satisfies the Jacobi identity. This follows easily from the formulae \eqref{eq:mumford_poisson}. Let
us show for example that
\begin{equation}\label{eq:jacobi_to_check}
  \pb{{\pb{\cQ(x),\cQ(y)}^\phi},\cR(z)}^\phi+\pb{\pb{\cQ(y),\cR(z)}^\phi,\cQ(x)}^\phi+
  \pb{\pb{\cR(z),\cQ(x)}^\phi,\cQ(y)}^\phi=0\;.
\end{equation}
The first term in \eqref{eq:jacobi_to_check} is zero since $\pb{\cQ(x),\cQ(y)}^\phi=0$. Also, by direct computation
using \eqref{eq:mumford_poisson},
\begin{align*}
  \pb{\pb{\cQ(y),\cR(z)}^\phi,\cQ(x)}^\phi
  &=\pb{2z\frac{\phi(y)\cP(z)-\phi(z)\cP(y)} {y-z}-\bcancel{\phi(z)\cQ(y)},\cQ(x)}^\phi\\
  &=\frac{2z\phi(y)}{y-z} \frac{\phi(z)x\cQ(x)-{\phi(x)z\cQ(z)}} {z-x}
    -\frac{2z\phi(z)}{y-z} \frac{\phi(y)x\cQ(x)-\phi(x)y\cQ(y)} {y-x}\\
    &=-\frac{2xz\phi(y)\phi(z)}{(z-x)(x-y)}\cQ(x)-\frac{2yz\phi(x)\phi(z)}{(y-z)(x-y)}\cQ(y)-
    \frac{2z^2\phi(x)\phi(y)}{(y-z)(z-x)}\cQ(z)\;.
\end{align*}
By skew-symmetry, $\pb{\pb{\cR(z),\cQ(x)}^\phi,\cQ(y)}^\phi=-\pb{\pb{\cQ(x),\cR(z)}^\phi,\cQ(y)}^\phi$, so that
\begin{equation*}
  \pb{\pb{\cR(z),\cQ(x)}^\phi,\cQ(y)}^\phi
    =\frac{2yz\phi(x)\phi(z)}{(y-z)(x-y)}\cQ(y)+\frac{2xz\phi(y)\phi(z)}{(z-x)(x-y)}\cQ(x)+
    \frac{2z^2\phi(x)\phi(y)}{(y-z)(z-x)}\cQ(z)\;.
\end{equation*}  
Summing up, we get indeed zero, which shows \eqref{eq:jacobi_to_check}.

Notice that since $\PB^\phi$ is a Poisson bracket for any $\phi\in\C[x]$ with $\deg\phi\leqslant g+1$ and
$\phi(0)=0$, the Poisson brackets $\PB^\phi$ are compatible, simply because $\PB^\phi+\PB^{\phi'} =
\PB^{\phi+\phi'}$ for any such polynomials $\phi,\phi'$.

We show in the next proposition that the Volterra map $\cV_g$ is a Poisson map with respect to each one of these
Poisson brackets. Recall from \eqref{PQrec2} and \eqref{weq} that $\cV_g$ is given by
\begin{equation}\label{eq:map_in_PQR}
  \tilde\cP(x)=\cQ(x)-\cP(x)\;,\qquad\tilde\cQ(x)=\frac{2\cP(x)-\cQ(x)+\cR(x)}{-w x}\;,\qquad
  \tilde\cR(x)=-w x\cQ(x)\;,
\end{equation}
where $w=- \frac{2\,p_1-q_1+r_1}{2} $.
\begin{propn}\label{prp:volterra_is_poisson}
  For any $\phi=\sum_{i=1}^{g+1}\phi_ix^i$, the Volterra map $\cV_g:M_g\to M_g$ is a (birational) Poisson map with
  respect to $\PB^\phi$.
\end{propn}
\begin{proof}
We need to check that $\cV_g$ preserves the Poisson bracket $\PB^\phi$. In formulae, this means that
$\pb{\tilde\cS(x),\tilde\cT(y)}^\phi=\widetilde{\pb{\cS(x),\cT(y)}^\phi}$, where $\cS$ and $\cT$ stand for any of
the polynomials $\cP,\cQ,\cR$. To do this, it helps 
to first use \eqref{eq:mumford_poisson} to derive  the
following formulae for the Poisson brackets of $w$ with the polynomials $\cP,\cQ,\cR$:
\begin{equation}\label{eq:brackets_with_alpha}
  \pb{w,\cP(y)}^\phi=\pb{w,\cQ(y)}^\phi=\frac{\phi_1}2\cQ(y)-\frac{\phi(y)}y\;,\quad
  \pb{w,\cR(y)}^\phi=\frac{\phi(y)}y-w\phi(y)-\frac{\phi_1}2(2\cP(y)+\cR(y))\;.
\end{equation}
To derive these formulae from \eqref{eq:mumford_poisson}, it suffices to use that $-2w$ is the linear term of
$2\cP(x)-\cQ(x)+\cR(x)$. For example,
\begin{align*}
  \pb{-2w,\cR(y)}^\phi
  &=\lim_{x\to0}\frac1x\pb{2\cP(x)-\cQ(x)+\cR(x),\cR(y)}^\phi\\
  &=\lim_{x\to0}\frac1x\left[-2y\frac{\phi(x)\cR(y)-\phi(y)\cR(x)}{x-y}-
    2y\frac{\phi(x)\cP(y)-\phi(y)\cP(x)}{x-y}+\phi(y)(\cQ(x)+\cR(x))-\phi(x)\cR(y)\right]\\
  &=\phi_1\cR(y)-\phi(y)r_1+2\phi_1\cP(y)+\phi(y)\lim_{x\to0}\frac{2y\cP(x)+(x-y)\cQ(x)}{x(x-y)}\\
  &=\phi_1\cR(y)-\phi(y)r_1+2\phi_1\cP(y)-2\frac{\phi(y)}y-(2p_1-q_1)\phi(y)\\
  &={\phi_1}(2\cP(y)+\cR(y))-2\frac{\phi(y)}y+2w\phi(y)\;,
\end{align*}
which proves the last formula in \eqref{eq:brackets_with_alpha}. It is now easy to prove that the Volterra map is a
Poisson map. For example, 
%
%
it is clear from \eqref{eq:brackets_with_alpha} that $\pb{\cP(x)-\cQ(x),w}^\phi=0$, and so
\begin{align*}
  \pb{\tilde\cP(x),\tilde\cR(y)}^\phi&=\pb{\cQ(x)-\cP(x),-w y\cQ(y)}^\phi=w y\pb{\cP(x),\cQ(y)}^\phi
  =w y\frac{\phi(x)y\cQ(y)-\phi(y)x\cQ(x)} {x-y}\\
  &=-y\frac{\phi(x)\tilde\cR(y)-\phi(y)\tilde\cR(x)}{x-y}=    \widetilde{\pb{\cP(x),\cR(y)}^\phi}\;.
\end{align*}
The other verifications are done in the same way. This shows that the Volterra map is a Poisson map.
\end{proof}

We show in the next proposition that the invariants of the Volterra map, as introduced in the previous section, are in
involution with respect to any of the Poisson brackets $\PB^\phi$, where we recall that $\deg\phi\leqslant g+1$ and
$\phi(0)=0$.
\begin{propn}\label{prp:H_in_involution}
  The $2g+1$ invariants $H_1,\dots,H_{2g+1}$ of the Volterra map, defined by
  \begin{equation}\label{eq:H_def}
    \cP(x)^2+\cQ(x)\cR(x)=1+\sum_{i=1}^{2g+1}H_ix^i\;.
  \end{equation}
  are in involution with respect to $\PB^\phi$.
\end{propn}
\begin{proof}
%
%
We first compute from \eqref{eq:mumford_poisson}
\begin{align}\label{eq:vf1}
  \pb{\cP(x),\cP(y)^2+\cQ(y)\cR(y)}^\phi&= \cQ(y)\pb{\cP(x),\cR(y)}^\phi+\cR(y)\pb{\cP(x),\cQ(y)}^\phi\nonumber\\
  &=-y\cQ(y)\frac{\cancel{\phi(x)\cR(y)}-\phi(y)\cR(x)}{x-y}+\cR(y)\frac{\cancel{y\phi(x)\cQ(y)}-x\phi(y)\cQ(x)}{x-y}
  \nonumber\\
  &=\phi(y)\frac{y\cQ(y)\cR(x)-x\cQ(x)\cR(y)}{x-y}\;,
\end{align}
and similarly
\begin{align}
  \pb{\cQ(x),\cP(y)^2+\cQ(y)\cR(y)}^\phi&=2\phi(y)\frac{x\cP(y)\cQ(x)-y\cP(x)\cQ(y)}{x-y}-\phi(y)\cQ(x)\cQ(y)\;,
  \label{eq:vf2}\\
  \pb{\cR(x),\cP(y)^2+\cQ(y)\cR(y)}^\phi&=2x\phi(y)\frac{\cP(x)\cR(y)-\cP(y)\cR(x)}{x-y}+\phi(y)\cQ(y)\cR(x)\;.
  \label{eq:vf3}
\end{align}
These formulae imply
\begin{align*}
  \lefteqn{\pb{\cP^2(x)+\cQ(x)\cR(x),\cP^2(y)+\cQ(y)\cR(y)}^\phi}\\
  &= 2\cP(x) \pb{\cP(x),\cP^2(y)+\cQ(y)\cR(y)}^\phi
  +\cQ(x)\pb{\cR(x),\cP^2(y)+\cQ(y)\cR(y)}^\phi
  +\cR(x)\pb{\cQ(x),\cP^2(y)+\cQ(y)\cR(y)}^\phi\\
  &=2\cP(x)\phi(y)\frac{\cancel{y\cQ(y)\cR(x)}-\bcancel{x\cQ(x)\cR(y)}}{x-y}
    +\cQ(x)2x\phi(y)\frac{\bcancel{\cP(x)\cR(y)}-\xcancel{\cP(y)\cR(x)}}{x-y}+\cQ(x)\phi(y)\cQ(y)\cR(x)\\
  &\qquad +\cR(x)2\phi(y)\frac{\xcancel{x\cP(y)\cQ(x)}-\cancel{y\cP(x)\cQ(y)}}{x-y}-\cR(x)\phi(y)\cQ(x)\cQ(y)=0\;,
\end{align*}
so that the invariants $H_i$ are in involution, $\pb{H_i,H_j}^\phi=0$ for $1\leqslant i,j\leqslant 2g+1$.
\end{proof}
As before, we often fix the Hamiltonians $H_1,\dots,H_g$ to generic values so that we actually work on a 
(non-singular) subvariety of $M_g^c$, which is birational with $\C^{2g+1}$. We show in the following proposition
that $M_g^c$ is a bi-Hamiltonian manifold, i.e., that it is equipped with a pencil of compatible Poisson
structures.

\begin{propn}\label{prp:bi-Ham}
  The Hamiltonians $H_1,\dots,H_g$ are Casimirs of the Poisson structure $\PB^{\phi}$ if and only if $\phi$ is of
  the form $\phi=\phi_gx^g+\phi_{g+1}x^{g+1}$. In this case, 
  \begin{equation}\label{eq:Casimir_of_pencil}
    C_\phi:=\sum_{i=0}^g(-1)^i\phi_g^{g-i}\phi_{g+1}^iH_{2g+1-i}
  \end{equation}
  is also a Casimir function of $\PB^\phi$. 
\end{propn}
\begin{proof}
Consider for $i=1,\dots,g$ the Hamiltonian vector field~$\rX_{H_i}$, which is given as the coefficient of $y^i$ in
$\Pb{\cP(y)^2+\cQ(y)\cR(y)}^\phi$, which we computed in \eqref{eq:vf1} -- \eqref{eq:vf3}. These Hamiltonians are
Casimir functions for $\PB^\phi$ if and only if their Hamiltonian vector fields are zero, which is in turn
equivalent to the fact that the right hand sides in \eqref{eq:vf1} -- \eqref{eq:vf3} are divisible by
$y^{g+1}$. Notice that the latter right hand sides, without the factor $\phi(y)$, are divisible by $y$ since
$\cP(0)=1,$ $\cQ(0)=2$ and $\cR(0)=0$, without being divisible by $y^2$. Therefore, this is equivalent to $\phi(y)$
being divisible by $y^g$, i.e., that $\phi$ is of the form $\phi=\phi_gx^g+\phi_{g+1}x^{g+1}$; since
$\deg\phi\leqslant g+1$, it follows that the Poisson structures which make $M_g^c$ into a bi-Hamiltonian manifold
are the restrictions to $M_g^c$ of the Poisson pencil $\PB^\phi$ with $\phi(x)=\phi_g x^g+\phi_{g+1}x^{g+1}$. It is
easily shown by direct computation that $C_\phi$, given by \eqref{eq:Casimir_of_pencil} is a Casimir of the Poisson
pencil. For example,
\begin{align*}
  \lefteqn{\pb{\cP(x),\sum_{i=0}^g(-1)^i\phi_g^{g-i}\phi_{g+1}^iH_{2g+1-i}}^\phi}\hskip1cm\\
  &=\Res_{y=0}\frac1{y^{2g+2}}\pb{\cP(x),H_{2g+1-i}y^{2g+1}}^\phi\sum_{i=0}^g(-1)^i\phi_g^{g-i}\phi_{g+1}^i\\
  &=\Res_{y=0}\frac1{y^{2g+2}}\pb{\cP(x),\cP(y)^2+\cQ(y)\cR(y)}^\phi\sum_{i=0}^g(-1)^i\phi_g^{g-i}\phi_{g+1}^iy^i\\
  &=\Res_{y=0}\frac1{y^{2g+2}}\left[\frac{\phi(x)y\cQ(y)-\phi(y)x\cQ(x)} {x-y}\cR(y)
       -y\cQ(y)  \frac{\phi(x)\cR(y)-\phi(y)\cR(x)} {x-y}\right]\sum_{i=0}^g(-1)^i\phi_g^{g-i}\phi_{g+1}^iy^i\\
  &=\Res_{y=0}\frac1{y^{2g+2}}\; \frac{y\cQ(y)\cR(x)-x\cQ(x)\cR(y)} {x-y}
    \sum_{i=0}^g(-1)^i\phi_g^{g-i}\phi_{g+1}^iy^i\phi(y)\\
  &=\Res_{y=0}\frac1{y^{2g+2}}\; \frac{y\cQ(y)\cR(x)-x\cQ(x)\cR(y)} {x-y}
    \left(\phi_g^{g+1}y^g+(-1)^g\phi_{g+1}^{g+1}y^{2g+1}\right)\\
  &=\phi_g^{g+1}\Res_{y=0}\frac1{y^{g+2}}\; \frac{y\cQ(y)\cR(x)-x\cQ(x)\cR(y)} {x-y}=0\;,
\end{align*}
where we have used in the last two equalities respectively that $\cR(y)$ is divisible by $y$ and that the
polynomial $\frac{y\cQ(y)\cR(x)-x\cQ(x)\cR(y)}{x-y}$ has degree at most $g$ in $y$.
\end{proof}

\begin{exa}\label{g1ex}
Continuing Example \ref{exa:g=1_rec}, we specialize the above results to $g=1$ and make them more explicit. Recall
that $M_1$ is the $4$-dimensional vector space of polynomials $(\cP,\cQ,\cR)$ with
\begin{equation}
  \cP(x)=1+p_1x\;,\quad  \cQ(x)=2+q_1x\;,\quad \cR(x)=r_1x+r_2x^2\;,
\end{equation}
and the invariants $H_1,H_2,H_3$ are given in \eqref{eq:hams_g=1}. The Poisson structures $\PB^\phi$ on $M_1$ are parametrized by
$\phi(x)=\phi_1x+\phi_2x^2$ and they all have $H_1$ as a Casimir function. The Poisson matrices of the basic
Poisson structures $\PB^x$ and $\PB^{x^2}$ (with the coordinates taken in the following order:
$p_1,\,q_1,\,r_1,\,r_2$) are easily determined from \eqref{eq:mumford_poisson} and are given by
\begin{equation}\label{eq:poisson_g=1}
  \PB^{x}=\begin{pmatrix}
    0&-q_1&0&r_2\\
    q_1&0&-q_1&0\\
    0&q_1&0&-r_2\\
    -r_2&0&r_2&0
  \end{pmatrix}\;, 
  \qquad
  \PB^{x^2}=\begin{pmatrix}
    0&2&0&-r_1\\
    -2&0&2&2q_1-p_1\\
    0&-2&0&r_1\\
    r_1&p_1-2q_1&-r_1&0
  \end{pmatrix}\;.
\end{equation}
For a generic $c_1\in\C$, the subvariety $M_1^c$ is defined by $H_1=c_1$. The pencil of Poisson structures
$\PB^\phi$ can be restricted to $M_1^c$ and the Casimir function $C_\phi$ on $(M_1^c,\PB^\phi)$, given by
\eqref{eq:Casimir_of_pencil}, takes the simple form $C_\phi=\phi_2H_2-\phi_1H_3$. In particular, $H_2$ and $H_3$
are Casimir functions of $\PB^{x^2}$ and $\PB^{x}$, respectively, and $\PB^x$ can be restricted to $H_3=c_3$ while
$\PB^{x^2}$ can be restricted to $H_2=c_2$ (for generic $c_2,c_3$). Since, by Proposition
\ref{prp:volterra_is_poisson}, the Volterra map is a Poisson map 
with respect to any such $\phi$ (of degree at
most $g+1$ with $\phi(0)=0$), it follows that 
\begin{enumerate}
\item[$\bullet$] \eqref{eq:rec_g1} defines a Poisson map with respect to the full Poisson pencil $\PB^\phi$;
\item[$\bullet$] \eqref{eq:g=1_red_H3} defines a Poisson map with respect to $\PB^{x}$;
\item[$\bullet$] \eqref{eq:g=1_red_H2} defines a Poisson map with respect to 
$\PB^{x^2}$.
\end{enumerate}
It follows from the Poisson matrices \eqref{eq:poisson_g=1} and from \eqref{eq:pqr_n_g=1} that the Poisson
structures $\PB^x$ and $\PB^{x^2}$ are respectively given by $\pb{w_{n-1},w_{n}}^x=-(w_{n-1}+w_{n}+c_1/2)/2$, and
$\pb{w_{n-1},w_{n}}^{x^2}=1/2$.

In Example \ref{exa:g=1_rec} we also considered the recursion relation on the surface
$H_2'=c_2',\,H_3'=c_3'$. None of the Poisson structures $\PB^\phi$ considered above can be restricted to these
surfaces, but a Nambu-Poisson structure with $H_2'$ and $H_3'$ as Casimirs can be so restricted. It leads to the
quadratic Poisson  bracket
\begin{equation}\label{eq:quadratic_Poisson}
  \pb{w_{n-1},w_n}=w_{n-1}w_n
\end{equation}
with respect to which \eqref{eq:rec_for_somos} is a Poisson map. 
The above bracket can also be derived via reduction of a presymplectic structure 
for the tau functions, by regarding \eqref{s5} as a mutation in a cluster algebra 
\cite{fordy_hone}.

To see how the quadratic bracket arises here, recall from \cite{takhtajan} that, 
for some fixed choice of $m$-form $\Om$, a Nambu-Poisson bracket of order $m$ is defined by 
$$ 
\{f_1,f_2,\ldots,f_m\} \, \Om = \rd f_1\wedge \rd f_2\wedge \cdots \wedge \rd f_m. 
$$   
In the case at hand (taking $m=4$),  
observe that 
the unreduced  version (\ref{wmapg1}) of the Volterra map $\varphi$ on $\C^{4}_w$
with coordinates $(w_{-1},w_0,w_1,w_2)$ preserves the rational volume form 
$$ 
\Om = 
\frac{w_0w_1}{w_0-w_1}\, \rd w_{-1}\wedge \rd w_0\wedge\rd w_1 \wedge \rd w_2, \qquad \varphi^*(\Om)=\Om. 
$$ 
Then the corresponding Nambu-Poisson bracket defines a Poisson bracket 
on $\C^{4}_w$ according to 
$$ 
\{f_1,f_2\} :=\{f_1,f_2,H_2',H_3'\}, 
$$ 
which by construction has  Casimirs $H_2'$ and $H_3'$,  and restricts to 
(\ref{eq:quadratic_Poisson}) on the surface $H_2'=c_2'$, $H_3'=c_3'$.  
The same Nambu-Poisson bracket also produces any member of the 
pencil $\PB^\phi$ for $g=1$: in particular, the Poisson structures 
$\PB^{x}$ and $\PB^{x^2}$ arise in this way, by taking  (up to scaling) 
$\{f_1,f_2,H_1,H_3\}$ and $\{f_1,f_2,H_1,H_2\}$, respectively. 
However, this construction does not extend to $g>1$ in a straightforward manner. 
\end{exa}
\begin{exa} 
Continuing Example \ref{exa:g=2_rec}, taking $\phi=x^3$ when $g=2$ 
we find that the Poisson brackets on $M_2$ (with coordinates $p_1,p_2,q_1,q_2,r_1,r_2,r_3$) 
take the form 
\begin{align*} 
&\{r_1,r_3\}^{x^3} = r_1, & \{r_2,r_3\}^{x^3} = r_2, \qquad\qquad & \{p_1,q_2\}^{x^3} = 2 = \{p_2,q_1\}^{x^3},\\ 
&\{p_2,q_2\}^{x^3} = q_1, &  \{p_2,r_3\}^{x^3} = -r_2, \quad\qquad &  \{p_1,r_3\}^{x^3} = -r_1=\{p_2,r_2\}^{x^3} ,  \\
&\{q_1,r_2\}^{x^3} = 2 = \{q_2,r_1\}^{x^3},  & \{q_1,r_3\}^{x^3} = 2p_1-q_1,\,\quad & \{q_2,r_2\}^{x^3} = 2p_1, \quad \{q_2,r_3\}^{x^3} = 2p_2-q_2, 
\end{align*}
where only the 
non-zero brackets are specified here.  By  restricting  this Poisson structure to $M_g^c\cap(H_3=c_3)$, we find that 
(up to an overall factor of $1/2$) it coincides with Poisson bracket (\ref{pbiv}) 
for the (P.iv) map that was derived from a discrete Lagrangian  
in \cite{gjtv2}, 
where the parameters $\nu, a,b$ are fixed according to (\ref{g2nuab})   
in terms of the values $c_1,c_2,c_3$ of the 3 Casimirs of $\PB^{x^3}$. 

\end{exa}
%

\subsection{The generic fibers of the momentum map}

We will now give an algebro-geometric description of the generic fibers of the momentum map $\mu:M_g\to\C[x]$,
which we recall is given by $\mu(\cP,\cQ,\cR)=\cP^2+\cQ\cR$. Precisely, we will describe the fiber over any
polynomial $f\in\C[x]$ of degree $2g+1$, without multiple roots, and satisfying $f(0)=1$. Notice that such 
polynomials are exactly those $f$  in the image of $\mu$ for which 
$y^2=f(x)$ defines a non-singular affine curve 
$\Gamma_f$ 
of genus $g$; we use $\bar\Gamma_f$ to denote the completion of the latter, 
which can be thought of as a compact Riemann surface.

\begin{propn}\label{prp:generic_fibers}
  Let $f$ be polynomial of degree $2g+1$, without multiple roots, such that $f(0)=1$. Denote by $\Gamma_f$ the
  non-singular affine curve of genus $g$, defined by $y^2=f(x)$, with 
$\bar\Gamma_f$ being its  completion.
Then $\mu^{-1}(f)$ is
  isomorphic to an affine part of the Jacobian variety of $\bar\Gamma_f$ minus three translates of the theta
  divisor,
  \begin{equation}\label{eq:fiber}
    \mu^{-1}(f)\cong \Jac(\bar\Gamma_f)\setminus (\Theta\cup\Theta_+\cup\Theta_-)\;.
  \end{equation}
  The $(-1)$-involution on $\Jac(\bar\Gamma_f)$ leaves $\Theta$ invariant and permutes $\Theta_+$ and
  $\Theta_-$.
\end{propn}

\begin{proof}
Let $(\cP,\cQ,\cR)\in\mu^{-1}(f)$, so that
  \begin{equation}\label{def:fiber}
    \cP^2(x)+\cQ(x)\cR(x)=f(x)\;,\quad\deg\cP,\deg\cQ,\deg\cR-1\leqslant g\;,
    \quad(\cP(0),\cQ(0),\cR(0))=(1,2,0)\;.
  \end{equation}
Since $f$ has degree $2g+1$, $\deg\cQ=g$ and $\deg\cR=g+1$; it also implies that $\bar\Gamma_f$ is obtained from
$\Gamma_f$ by adding a single point, which we denote by $\infty$. The hyperelliptic involution on $\bar\Gamma_f$ is
denoted by $\imath$; it fixes $\infty$ and sends $(x,y)\in\Gamma_f$ to $(x,-y)$. To $(\cP,\cQ,\cR)$ we associate a
divisor $\sum_{i=1}^g(x_i,y_i)-g\infty$ on $\bar\Gamma_f$ as follows: $x_1,x_2,\dots,x_g$ are the roots of $\cQ(x)$
and $y_i:=\cP(x_i)$ for $i=1,\dots,g$. It is indeed a divisor on $\bar\Gamma_f$ since for $j=1,\dots,g$,
\begin{equation*}
  y_j^2-f(x_j)=\cP^2(x_j)-(\cP^2(x_j)+\cQ(x_j)\cR(x_j))=0\;.
\end{equation*}
Of course, $(x_i,y_i)\neq\infty$ for all $i$. Notice that when $\cQ(x)$ has multiple roots, say $x_1=\dots=x_k$,
then $y_1=y_2=\dots=y_k$. We show by contradiction that if $x_i$ is a root of $\cQ(x)$ and $y_i=0$ (so that $x_i$
is also a root of $\cP(x)$), then $x_i$ is a simple root of $\cQ(x)$. Indeed, if $x_i$ is a multiple root of
$\cQ(x)$ and is also a root of $\cP(x)$, then $x_i$ is a multiple root of $f(x)=\cP(x)^2+\cQ(x)\cR(x)$, so that
$f$ is not square-free, contrary to the assumptions. The upshot is that the obtained divisors are of the form
$\sum_{i=1}^gP_i-g\infty$, where $P_i\in\Gamma_f$ for $i=1,\dots,g$ and $P_i\neq\imath(P_j)$ when $i\neq j$. It is
well-known that two such divisors are linearly equivalent if and only if they are the same; also, that none of
these divisors are equivalent to a divisor of the form $\sum_{i=1}^{g-1}Q_i-(g-1)\infty$, with $Q_i\in\bar\Gamma_f$
for $i=1,\dots,g-1$. Since $\Jac(\bar\Gamma_f)$ is the group of degree zero divisors on $\bar\Gamma_f$, modulo
linear equivalence, this shows that the map $\mu^{-1}(f)\to\Jac(\bar\Gamma_f)$, which associates to
$(\cP,\cQ,\cR)$ the divisor class $\left[\sum_{i=1}^gP_i-g\infty\right]$, is injective.
This map is of course not surjective. In order to determine the image, let
\begin{align*}
  \Theta&:=\left\{\left[\sum_{i=1}^{g-1}P_i-(g-1)\infty\right]\mid\forall i\ P_i\in\bar\Gamma_f\right\}\;,\\
  \Theta_+&:=\left\{\left[(0,1)+\sum_{i=1}^{g-1}P_i-g\infty\right]\mid\forall i\ P_i\in\bar\Gamma_f\right\}
    =\Theta+[(0,1)-\infty]\;,\\
  \Theta_-&:=\left\{\left[(0,-1)+\sum_{i=1}^{g-1}P_i-g\infty\right]\mid\forall i\ P_i\in\bar\Gamma_f\right\}
    =\Theta+[(0,-1)-\infty]\;.
\end{align*}
The first one is the theta divisor and the other two are translates of it. Notice that $\imath(\Theta)=\Theta$ and
$\imath(\Theta_+)=\Theta_-$, since $\imath(\infty)=\infty$.  As we already said, the image contains no point of the
form $\left[\sum_{i=1}^{g-1}P_i-(g-1)\infty\right]$, i.e., is disjoint from $\Theta$. Since $P_i=(x_i,y_i)$ where
$x_i$ is a root of $\cQ(x)$ and since $\cQ(0)=2$, it is clear that $x_i\neq0$, so that every $P_i=(x_i,y_i)$ is
different from $(0,1)$ and from $(0,-1)$; 
hence the image is also disjoint from $\Theta_+$ and
$\Theta_-$. Take now any point in $\Jac(\bar\Gamma_f)\setminus (\Theta\cup\Theta_+\cup\Theta_-)$. It can as above
be written uniquely as $\left[\sum_{i=1}^gP_i-g\infty\right]$ with $P_i\notin\{\infty,(1,0),(-1,0)\}$ and
$P_i\neq\imath(P_j)$ for all $i\neq j$. When all $P_i=(x_i,y_i)$ are different, there is a unique polynomial
$\cQ(x)$ whose roots are the $x_i$ and with $\cQ(0)=2$, and there is a unique polynomial $\cP(x)$ of degree $g$,
with $\cP(x_i)=y_i$ for $i=1,\dots,g$ and $\cP(0)=1$: setting $(x_0,y_0)=(0,1)$, they are given by
\begin{equation}\label{eq:PQ_from_xy}
  \cQ(x)=2\prod_{i=1}^g\left(1-\frac x{x_i}\right)\;,\qquad \cP(x)=\sum_{i=0}^gy_i\prod_{j\neq
    i}\frac{x-x_j}{x_i-x_j}\;.
\end{equation}
This also works in the limiting case when some of the $P_i$ are the same upon adding in the definition of $\cP(x)$
a tangency condition (see \cite[page 3.18]{tata2}), which assures that $f(x)-\cP^2(x)$ is divisible by
$\cQ(x)$. The quotient is a polynomial $\cR(x)$ of degree $g$ satisfying $f(x)=\cP(x)^2+\cQ(x)\cR(x)$; by
uniqueness, $\mu(\cP,\cQ,\cR)=\left[\sum_{i=1}^gP_i-g\infty\right]$, as required.
\end{proof}

\begin{exa}
We specialize Proposition \ref{prp:generic_fibers} to the case of $g=1$. Notice that in this case we should say
\emph{elliptic} rather than \emph{hyperelliptic}, and in this case the (hyper-) elliptic involution is not
unique. Another peculiarity about $g=1$ is the well-known fact that a complete non-singular genus one curve
(i.e., any compact elliptic Riemann surface) is isomorphic to its Jacobian, a fact that we will be able to
illustrate here. Let $f(x)=1+c_1x+c_2x^2+c_3x^3$ be an arbitrary polynomial of degree 3 with no multiple roots. We
investigate $\mu^{-1}(f)$, which is the affine curve defined by the following equations, which are found by
expressing that $\cP^2(x)+\cQ(x)\cR(x)=f(x)$:
\begin{align*}
  &2(p_1+r_1)=c_1\;,\\
  &p_1^2+q_1r_1+2r_2=c_2\;,\\
  &q_1r_2=c_3\;.
\end{align*}
The curve $\mu^{-1}(f)$ can be equivalently written as a plane algebraic curve by first eliminating $r_1$ from the
first two equations and then $r_2$ from the two remaining equations; writing $p_1$ and $q_1$ simply as $p$ and $q$,
the final equation takes the simple form
\begin{equation}\label{eq:fiber_g=1}
  \mu^{-1}(f)\ :\ \frac{pq}2(p-q)+c_1\frac{q^2}4-c_2\frac q2+c_3=0\;.
\end{equation}
It is easy to see that, thanks to the conditions on $f$, this curve is non-singular, just like $\Gamma_f$. In fact,
if we denote the left-hand side of \eqref{eq:fiber_g=1} by $F$ then a singular point $(q_0,p_0)$ of $\mu^{-1}(f)$
must satisfy
\begin{equation}\label{eq:g=1_non-singular}
  \frac{\partial F}{\partial p}(q_0,p_0)= \frac{q_0}2(q_0-2p_0)=0\;,\quad
  \frac{\partial F}{\partial q}(q_0,p_0)= \frac12(p_0^2-2p_0q_0+c_1q_0-c_2)=0\;.
\end{equation}
Since $q_0\neq0$ (as $F(q_0,p_0)=0$ and $c_3\neq0$), we get $q_0=2p_0$; substituted in $F(q_0,p_0)=0$ and in the
second equation of \eqref{eq:g=1_non-singular} we get
\begin{equation*}
  p_0^3-c_1p_0^2+c_2p_0-c_3=0\;,\quad 3p_0^2-2c_1p_0+c_2=0\;,
\end{equation*}
which can be written as $f(-p_0)=f'(-p_0)=0$. Since $f$ has no multiple roots these equations have no common solution,
which shows that $\mu^{-1}(f)$ is non-singular. 

To see that $\mu^{-1}(f)$ and $\Gamma_f$ are birationally isomorphic, it suffices to consider the following rational map: 
\begin{equation}\label{eq:birat_iso_g=1}
  q=-\frac2x\;,\ p=\frac{y-1}x\;,\quad\hbox{ with inverse } \quad
  x=-\frac2q\;,\ y=\frac{q-2p}q\;.
\end{equation}
Notice that, despite the appearance of $q$ in the denominator,  the inverse map 
in \ref{eq:birat_iso_g=1} is actually regular, 
because $q\neq0$ on $\mu^{-1}(f)$. The rational map and its
inverse extend (uniquely) to an isomorphism of the completions $\bar\Gamma_f$ and $\overline{\mu^{-1}(f)}$, which
can be thought of respectively as an elliptic curve and its Jacobian. It allows us to determine the number and
nature of the points at infinity of $\overline{\mu^{-1}(f)}$, i.e., the points needed to complete ${\mu^{-1}(f)}$
into $\overline{\mu^{-1}(f)}$; they are the points corresponding to  affine points $(x,y)$ for which the map is
not defined, to wit $(x,y)=(0,\pm1)$, and to the point at infinity $\infty$ of $\bar\Gamma_f$, making a total of three
points, as asserted by Proposition \ref{prp:generic_fibers}. More specifically, by using the map we can determine a
local parametrisation around these points from local parametrisations around the points $(0,\pm1)$ and
$\infty$. For the latter, we can take $(x,y)=(t,1\pm\frac{c_1t}2+\mathcal O(t^2))$ and
$(x,y)=(t^{-2},\sqrt{c_3}t^{-3}(1+\frac{c_2}{2c_3}t^2+\mathcal O(t^4))$ to obtain, again using the map, the
following local parametrisations at the three points at infinity of $\overline{\mu^{-1}(f)}$:
\begin{equation*}
  \infty_0:\left(-2t^2,\frac{\sqrt{c_3}}t\Big(1+\frac{c_2}{2c_3}t^2+\mathcal O(t^3)\Big)\right)\;,  \quad
  \infty_1:\left(-\frac2t,\frac{c_1}2+\mathcal O(t)\right)\;,\quad
  \infty_2:\left(-\frac2t,-\frac2t-\frac{c_1}2+\mathcal O(t)\right)\;.   
\end{equation*}
Again, using the map we can derive that the (hyper-) elliptic involution on $\Gamma_f$, which is
given by $(x,y)\mapsto (x,-y)$, is given on $\mu^{-1}(f)$ by $(q,p)\mapsto (q,q-p)$. It permutes the points
$\infty_1$ and $\infty_2$ while leaving $\infty_0$ fixed (together with the points $(2p,p)$ where $-p$ is a root of
$f$). The points $\infty_0$, $\infty_1$ and $\infty_2$ correspond to $\Theta\;,\Theta_+$ and $\Theta_-$, respectively.
\end{exa}

\begin{exa}
We also specialize Proposition \ref{prp:generic_fibers} to the case of $g=2$ and provide some extra information.
The polynomial $f$ is now of degree $5$, taking the form $f(x)=1+c_1x+c_2x^2+c_3x^3+c_4x^4+c_5x^5$,  
with 
no multiple roots, which is equivalent to the curve (\ref{wg2}) associated with the map (P.iv). 
In this case, the theta divisor and its translates are genus 2 curves, isomorphic to
$\bar\Gamma_f$. This general fact can also be seen here directly from the description that
$\Theta=\left[\bar\Gamma_f-\infty\right]$ and similarly for $\Theta_+=\left[\bar\Gamma_f+(0,1)-2\infty\right]$ and
$\Theta_-=\left[\bar\Gamma_f+(0,-1)-2\infty\right]$. We show that these curves in $\overline{\mu^{-1}(f)}$ meet
according to the intersection pattern in Figure~\ref{fig:divisor_at_infinity_g=2}.

\begin{figure}[h]
  \centering
  \begin{tikzpicture}[scale=1.7]
    \draw (1,2) circle (1);
    \draw (3,2) circle (1);
    \draw (2,1) circle (1);
    \draw (1,3) node [above left] {$\Theta_+$};
    \draw (3,3) node [above right] {$\Theta_-$};
    \draw (2,0) node [below] {$\Theta$};
    \draw (2,2) node [above left] {$0$};
    \fill (2,2) circle (0.07);
    \fill (3,1) circle (0.07);
    \fill (1,1) circle (0.07);
    \draw (3,1) node [below right] {$[(0,-1)-\infty]$};
    \draw (1,1) node [below left] {$[(0,1)-\infty]$};
  \end{tikzpicture}
  \caption{When $g=2$ the divisor at infinity of $\mu^{-1}(f)$ consist of three copies of the curve $y^2=f(x)$,
  intersecting according to the indicated pattern.}\label{fig:divisor_at_infinity_g=2}
\end{figure}
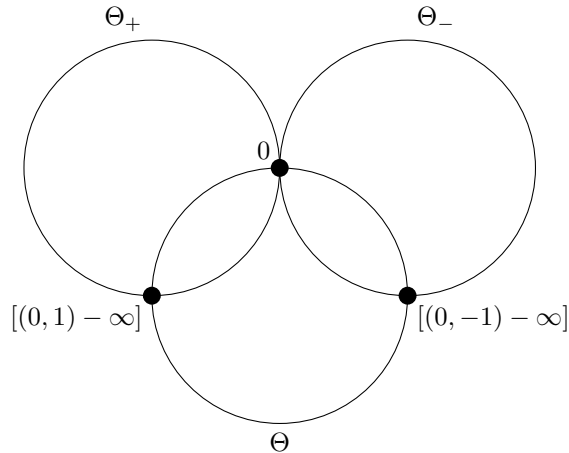

To do this, we first recall the general fact that two translates of the theta divisor (also called \emph{theta
curves}) intersect in two points which coincide if and only if the curves are tangent. Consider first a point
(divisor class) in $\Theta_+\cap\Theta_-$. It must be of the form $\left[P+(0,1)-2\infty\right]$ and of the form
$\left[Q+(0,-1)-2\infty\right]$, for some $P,Q\in\bar\Gamma_f$. In particular, the points $P$ and $Q$ must be such
that $P+(0,1)\sim Q+(0,-1)$; as we already recalled, a linear equivalence of such divisors amounts to equality, so
that $P=(0,-1)$ and $Q=(0,1)$ and there is a unique intersection point $[(0,1)+(0,-1)-2\infty]$ which is the origin
$O$ of $\Jac(\bar\Gamma_f)$, since $(0,1)+(0,-1)\sim 2\infty$. Consider next a point in $\Theta\cap\Theta_\pm$. It
must be both of the form $[P-\infty]$ and $[Q+(0,\pm1)-2\infty]$, for some $P,Q\in\bar\Gamma_f$. This leads us now
to the linear equivalence $P+\infty\sim Q+(0,\pm1)$, whose only solutions are $P=\infty,\,Q=(0,\mp1)$ and
$P=(0,\pm1),\,Q=\infty$; the first solution corresponds again to the origin $O$ while the other intersection point
is the point $[(0,\pm1)-\infty]$ (see Figure~\ref{fig:divisor_at_infinity_g=2}).
\end{exa}

\subsection{Discrete Liouville and algebraic integrability}
  
We are now ready to show that the Volterra map $\cV_g$ is Liouville integrable on $(M_g,\PB^\phi)$ when
$\phi\neq0$. Recall that a birational map $R$ on an algebraic Poisson manifold $(M,\PB)$ of dimension $n$ and
(Poisson) rank $2r$ is said to be \emph{Liouville integrable} when the following conditions are satisfied:
\begin{enumerate}
\item[(1)] $R$ is a Poisson map;
\item[(2)] $R$ has $n-r$ functionally independent invariants in involution.
\end{enumerate}
In the case at hand, $M=M_g$ so that $n=3g+1$, and $R=\cV_g$, which we already know to be a Poisson map
(Proposition \ref{prp:volterra_is_poisson}). Also, we already 
have $2g+1$ invariants in involution
(Proposition \ref{prp:H_in_involution}). So it will be sufficient to show, in the proof which follows, that the
rank of $\PB^\phi$ is $2g$ and that the invariants $H_1,\dots,H_{2g+1}$ are functionally independent.

\begin{propn}\label{prp:liouville_integrable}
  Let $\phi\in\C[x]$ be any non-zero polynomial of degree at most $g+1$, vanishing at $0$. Then the Volterra map
  $\cV_g$ is Liouville integrable on $(M_g,\PB^\phi)$.
\end{propn}
\begin{proof}
We first show that the components of $\mu$, which are the $2g+1$ polynomial functions $H_i$, defined by
$\mu(\cP,\cQ,\cR)=1+\sum_{i=1}^{2g+1}H_ix^i$, are functionally independent. According to
Proposition~\ref{prp:generic_fibers}, the generic fiber of $\mu$ (which is the generic fiber of
$H_1,\dots,H_{2g+1}$) is an open subset of the Jacobian of a curve of genus $g$, hence has dimension $g$. 
The dimension of the generic fiber of $\mu$ is given by $\dim
M_g-s=3g+1-s$, where $s$ denotes the number of functionally independent functions in
$H_1,\dots,H_{2g+1}$. Therefore, $s=2g+1$ and the components of $\mu$ are functionally independent.

It remains to be shown that the rank of $\PB^\phi$ is $2g$. Since $\dim M=3g+1$ and $\PB^\phi$ admits $2g+1$
functionally independent functions in involution, the rank of $\PB^\phi$ is at most $2g$ (see
\cite[Prop.\ 3.4]{vanhaecke}); to show equality, in a neighborhood of a generic point $(\cP,\cQ,\cR)\in M_g$ we
take the functions $x_1,\dots,x_n,y_1,\dots,y_n$, which we constructed in the proof of
Proposition~\ref{prp:generic_fibers}. Then $\pb{x_i,x_j}^\phi=\pb{y_i,y_j}^\phi=0$, since
$\pb{\cP(x),\cP(y)}^\phi=\pb{\cQ(x),\cQ(y)}^\phi=0$. We show that the brackets $\pb{x_i,y_j}^\phi$ are non-zero if
and only if $i=j$, from which it follows that the rank of $\PB^\phi$ is indeed $2g$. To do this, we compute for
$1\leqslant i\leqslant g$ the Poisson bracket $\pb{y_i,\ln\cQ(y)}^\phi$ in two different ways. First, using
\eqref{eq:mumford_poisson},
\begin{equation}
  \pb{y_i,\ln\cQ(y)}^\phi=\pb{\cP(x_i),\ln\cQ(y)}^\phi=\frac{\pb{\cP(x_i),\cQ(y)}^\phi}{\cQ(y)}
    =\frac{y\phi(x_i)-x_i\phi(y)\cQ(x_i)/\cQ(y)} {x_i-y}
    =y\frac{\phi(x_i)} {x_i-y}\;,
\end{equation}
and next, using \eqref{eq:PQ_from_xy},
\begin{equation}
  \pb{y_i,\ln\cQ(y)}^\phi=\sum_{j=1}^g\pb{y_i,\ln(1-y/x_j)}^\phi 
  =y\sum_{j=1}^g\frac{\pb{y_i,x_j}^\phi}{x_j(x_j-y)}\;,
\end{equation}
so that
\begin{equation*}
  \frac{\phi(x_i)} {x_i-y}=\sum_{j=1}^g\frac{\pb{y_i,x_j}^\phi}{x_j(x_j-y)}\;.
\end{equation*}
Since (generically) all $x_j$ are  different, $\pb{x_j,y_i}^\phi=0$ when $j\neq i$, while
$\pb{x_i,y_i}^\phi=-x_i\phi(x_i)$, so that
\begin{equation}\label{Mlike_pb_lin}
  \pb{x_i,y_j}^\phi=-x_i\phi(x_i)\delta_{ij}\;.
\end{equation}
Since $\phi\neq0$, this shows that the rank of $\PB^\phi$ is $2g$.
\end{proof}

We now conclude with the main result of this section, namely that the Volterra map is a \emph{discrete
a.c.i.\ system}. Recall from Definition \ref{def:adi} that this means that, besides being Liouville integrable,
the generic level sets defined by the invariants are affine parts of Abelian varieties (complex algebraic tori) and
the restriction of the map to any of these Abelian varieties is a translation.
\begin{thm}\label{prp:volterra_aci}
  The Volterra map $\cV_g$ is a \emph{discrete a.c.i.\ system} on $(M_g,\PB^\phi)$. 
\end{thm}
\begin{proof}
Liouville integrability was already shown in Proposition \ref{prp:liouville_integrable}. Let $f\in\C[x]$ be a
polynomial of degree $2g+1$, without repeated roots and with $f(0)=1$. Writing $f(x)=1+\sum_{i=0}^{2g+1}c_ix^i$,
the common level set defined by $H_i=c_i$, $i=1,2,\dots,2g+1$ is the fiber $\mu^{-1}(f)$ which was shown in
Proposition \ref{prp:generic_fibers} to be an affine part of the Jacobian of $\bar\Gamma_f$, and this is indeed an
Abelian variety. It remains to be shown that the restriction of the Volterra map $\cV_g$ to
$\overline{\mu^{-1}(f)}\cong \Jac(\bar\Gamma_f)$ is a translation; more precisely we will show that it is a
translation over $[(0,-1)-\infty]=[\infty-(0,1)]$. Let $(\cP,\cQ,\cR)\in \mu^{-1}(f)$ be a generic point (a regular
triplet) so that $(\tilde\cP,\tilde\cQ,\tilde\cR):=\cV_g(\cP,\cQ,\cR)$ also belongs to $\mu^{-1}(f)$. The degree
$g$ divisors on $\bar\Gamma_f$ corresponding to these two triplets are respectively denoted by
$\cD=\sum_{i=1}^g(x_i,y_i)$ and $\tilde\cD=\sum_{i=1}^g(\tilde x_i,\tilde y_i)$ (so that the corresponding divisor
classes in $\Jac(\bar\Gamma_f)$ are $[\cD-g\infty]$ and $[\tilde\cD-g\infty]$, respectively).  Consider the
rational function $\tilde F$ on $\bar\Gamma_f$ given by the action of $\cV_g$ on (\ref{func}), which in view of
\eqref{eq:map_in_PQR} and the definition \eqref{def:fiber} of $f$, can be written in a few different ways:
\begin{equation}\label{eq:rat_fun}
  \tilde F:=\frac{y+\tilde\cP(x)}{\tilde\cQ(x)}=\frac{\tilde\cR(x)}{y-\tilde\cP(x)}=
  \frac{-w x\cQ(x)}{y-\cQ(x)+\cP(x)}\;.
\end{equation}
It is clear from the last, respectively first, expression that $\tilde F$ has a simple zero at the each of the points $(x_i,y_i)$ and
$(0,-1)$, and a simple pole at each of the points $(\tilde x_i,\tilde y_i)$ and $\infty$. To verify the behaviour at 
$\infty$, one should introduce a local parameter $z$ such that $x=1/z^2$, $y=\sqrt{c_{2g+1}}z^{-(2g+1)}\big(1+\mathcal O(z)\big)$ there, 
which gives $\tilde F= \sqrt{c_{2g+1}}q_g^{-1}/z +\mathcal O(1)$. 
For the other zero or pole candidates in
$\Gamma_f$, coming from places where the numerators or denominators in \eqref{eq:rat_fun} vanish, one checks using
one of the alternative formulae 
that $\tilde F$ is finite and non-zero at these points. Thus $\tilde F$ 
has precisely $g+1$ zeros and $g+1$ poles, in accord with the fact that 
the degree of the divisor of a rational function is zero. The upshot is
that the divisor of zeros and poles of $\tilde F$ is given by
\begin{equation*}
  (\tilde F)=\sum_{i=1}^g(x_i,y_i)+(0,-1)-\sum_{i=1}^g(\tilde x_i,\tilde y_i)-\infty=\cD-\tilde\cD+(0,-1)-\infty\;,
\end{equation*}
which leads to the linear equivalence $\cD+(0,-1)\sim\tilde\cD+\infty,$ and hence to
\begin{equation}\label{shift}
  [\tilde\cD-g\infty]=[\cD-g\infty]+[(0,-1)-\infty]\;,
\end{equation}
as was to be shown.
\end{proof}
According to Proposition \ref{prp:bi-Ham}, when $\phi=\phi_gx^g+\phi_{g+1}x^{g+1}$ we can restrict the Volterra map
and its Poisson structure to $M_g^c=\cap_{i=1}^g(H_i=c_i)$, and so by the above theorem 
the Volterra map is a discrete a.c.i.\ system on $(M_g^c,\PB^\phi)$. In particular, the recursion relations obtained by
fixing the invariants $H_1,\dots,H_g$ to generic values $c_i$ (and possibly also fixing the other Casimir $C_\phi$
to some generic value) are discrete a.c.i.\ systems.

\begin{exa}
  In the genus 1 case,  it follows that \eqref{eq:rec_g1}, \eqref{eq:g=1_red_H3} and \eqref{eq:g=1_red_H2},
  equipped respectively with the Poisson structures $\PB^\phi$, $\PB^x$ and $\PB^{x^2}$, are discrete
  a.c.i.\ systems. The same holds for \eqref{eq:rec_for_somos}, which is a discrete a.c.i.\ system with respect to
  the quadratic Poisson structure \eqref{eq:quadratic_Poisson}.
\end{exa}

\begin{exa}
  In the genus 2 case, we have that for generic $a,b,\nu$ the map (P.iv) is a discrete a.c.i.\ system.
\end{exa}

\section{Continuous flows and the infinite Volterra and Toda lattices}\label{sec:continuous}
\setcounter{equation}{0}

The discrete integrable systems that we have discussed so far are naturally associated with continuous systems
which are equally integrable. More precisely, Liouville integrability 
of the Hamiltonian systems associated with the Volterra maps comes for free, 
and with some extra work we 
show that these continuous systems are also algebraically integrable. 
We further show that in the genus $g$ case any solution $w_i(t)$ of one of the
integrable Hamiltonian vector fields extends, under the action of the Volterra map,  
to a sequence $\big(w_n(t)\big)_{n\in\Z}$ that is a solution to the infinite Volterra lattice;  notice that in particular, 
as discussed in the introduction, this applies to
the map (P.iv). We also discuss the relation between the infinite Toda and Volterra lattices, which 
explains in part how some of the results in this paper
are related to the results in \cite{contfrac}, and what motivated us to introduce 
S-fractions and the corresponding Mumford-like
systems to study the map (P.iv) and its higher genus analogues.

\subsection{Liouville and algebraic integrability}\label{par:integrability}

Recall that on $M_g$ we have a family of compatible Poisson brackets $\PB^\phi$ of rank $2g$, as well as a family
of polynomial functions $H_1,\dots,H_{2g+1}$, where $H_i$ is the coefficient in $x^i$ of $\cP^2(x)+\cQ(x)\cR(x)$;
said differently, $H_1,\dots,H_{2g+1}$ are the components of the momentum map $\mu:M_g\to\C[x]$.
For the sake of clarity, and since the choice of Poisson structure is not important for what follows, we will only
consider $\phi=x^{g+1}$ here, and henceforth write $\PB$ for $\PB^\phi$. With this choice of $\phi$, \eqref{eq:vf1} --
\eqref{eq:vf3} become
\begin{align*}
  \pb{\cP(x),\cP(y)^2+\cQ(y)\cR(y)} &=y^{g+1}\frac{y\cQ(y)\cR(x)-x\cQ(x)\cR(y)}{x-y}\;,\\
  \pb{\cQ(x),\cP(y)^2+\cQ(y)\cR(y)}&=2y^{g+1}\frac{x\cP(y)\cQ(x)-y\cP(x)\cQ(y)}{x-y}-y^{g+1}\cQ(x)\cQ(y)\;,\\
  \pb{\cR(x),\cP(y)^2+\cQ(y)\cR(y)}&=2xy^{g+1}\frac{\cP(x)\cR(y)-\cP(y)\cR(x)}{x-y}+y^{g+1}\cQ(y)\cR(x)\;.
\end{align*}
As we have seen, $H_1,\dots,H_{g}$ are Casimir functions of the Poisson bracket, 
as well as  $H_{g+1}=C_\phi$ (see
\eqref{eq:Casimir_of_pencil}). The vector fields $\frac12\rX_{H_{g+i+1}}$ are denoted by~$\rX_i$; we will mainly be 
interested in $\rX_1=\frac12\rX_{H_{g+2}}$, which we can compute by dividing the above equations by $2y^{g+2}$ and taking
the limit for $y\to0$, so that 
\begin{equation}\label{eq:pdot}
  \dot\cP(x)=\lim_{y\to0}\frac{y\cQ(y)\cR(x)-x\cQ(x)\cR(y)}{2y(x-y)}=\frac{2\cR(x)-x\cQ(x)\cR'(0)}{2x}
            =\frac{\cR(x)}x-\frac{\r1}2\cQ(x)\;,
\end{equation}
where the dot denotes the derivative $\tfrac{\rd}{\rd t}$, 
and similarly
\begin{align}\label{eq:qrdot}
  \dot\cQ(x)&=\frac{2\p1-\q1}2\cQ(x)-\frac{2\cP(x)-\cQ(x)}x\;,\nonumber\\
  \dot\cR(x)&=\r1\cP(x)+\frac{\q1-2\p1}2\cR(x)-\frac{\cR(x)}x\;.
\end{align}
Notice that the vector field $\rX_1$ is (non-homogeneous) quadratic. 
From Proposition \ref{prp:H_in_involution}, the functions $H_j$ are in involution 
with one another with respect to  $\PB$, which means that the vector fields~$\rX_i$ 
all commute.  


Note that the Liouville integrability of this  continuous system is incorporated 
into the Liouville integrability of the discrete system, so the following statement 
is an automatic consequence of  Proposition
\ref{prp:liouville_integrable}. 
\begin{propn}\label{prp:cont_liouville}
  The Hamiltonian system $(M_g,\PB,\mu)$ is a Liouville integrable system.
\end{propn}
%
%

We now turn to the algebraic integrability of the Mumford-like system, which is 
slightly more involved in the continuous case. 
Recall (for example from \cite[Ch.\ 6]{amv}) that $(M_g,\PB,\mu)$ being an a.c.i.\ system means that
\begin{enumerate}
  \item[(1)] $(M_g,\PB,\mu)$ is a (complex) Liouville integrable system;
  \item[(2)] The generic fiber of $\mu$ is (isomorphic to) an affine part of an Abelian variety;
  \item[(3)] The integrable vector fields are holomorphic (hence constant) on these Abelian varieties.
\end{enumerate}
Items (1) and (2) have been shown already, in Propositions \ref{prp:cont_liouville} and
\ref{prp:generic_fibers}, respectively, so it only remains to address item (3). 

\begin{propn}\label{abelproof} 
 The Hamiltonian system $(M_g,\PB,\mu)$ is an algebraic 
completely integrable system (a.c.i.\ system). 
\end{propn}
\begin{proof}
We show (3) for one of the integrable vector fields; then it also holds for the other
integrable vector fields, since the latter are holomorphic on the fiber and commute with a constant vector field.

The vector field which we consider is the Hamiltonian vector field $\rX_1$, given by \eqref{eq:pdot} and
\eqref{eq:qrdot}.  Let $(\cP_0,\cQ_0,\cR_0)$ be a generic point of $M_g$ and consider for small $\vert t\vert$ the
integral curve $t\mapsto (\cP_t,\cQ_t,\cR_t)$ of $\rX_1$, starting at ($\cP_0,\cQ_0,\cR_0)$. Let
$\cD_t=\sum_{i=1}^g(x_i(t),y_i(t))-g\infty$ denote the associated divisor on the algebraic curve~$\Gamma_f$,
defined by it; recall that $\Gamma_f$ is given by $y^2=f(x)$ where $f=\cP_0^2+\cQ_0\cR_0=\cP_t^2+\cQ_t\cR_t$. Since
the $x_i(t)$ are the roots of $\cQ_t(x)$, upon substituting $x=x_i(t)$ in the equation \eqref{eq:qrdot} for
$\dot\cQ(x)$ we get 
\begin{equation*}
  \dot\cQ(x_i(t))=-2\frac{\cP(x_i(t))}{x_i(t)}=-2\frac{y_i(t)}{x_i(t)} \;.
\end{equation*}
However, we can also compute $\dot\cQ(x_i(t))$ directly from the explicit formula \eqref{eq:PQ_from_xy}
for $\cQ(x)$, to wit
\begin{equation*}
  \dot\cQ(x_i(t))=\frac{2\dot x_i(t)}{x_i(t)}\prod_{j\neq i}\left(1-\frac {x_i(t)}{x_j(t)}\right)\;.
\end{equation*}
Comparing these two expressions gives
\begin{equation}\label{eq:compare}
  y_i(t)=-{\dot x_i(t)}\prod_{j\neq i}\left(1-\frac {x_i(t)}{x_j(t)}\right)\;.
\end{equation}
It follows that for $k=0,\dots,g-1$,
\begin{equation}\label{eq:diff_lin}
  \sum_{i=1}^g\frac{x_i^k(t)\dd x_i(t)}{y_i(t)}=-\sum_{i=1}^gx_i^k(t)\prod_{j\neq
    i}\frac{x_j(t)}{x_j(t)-x_i(t)}\dd t=-\delta_{k,0}\,\dd t\;. 
\end{equation}
Above we have used the following identity which is well-known in the theory of symmetric functions: 
\begin{equation}
  \sum_{i=1}^gx_i^k\prod_{j\neq i}\frac{x_j}{x_j-x_i}=\delta_{k,0}\;, \qquad  
k=0,\dots,g-1;  
\end{equation}
the
proof of the latter follows easily from the fact that any lowest degree antisymmetric polynomial in $g$ variables is, up to a
factor, the Vandermonde determinant. 
Integrating \eqref{eq:diff_lin} from $0$ to $t$ gives
\begin{equation}\label{eq:abel}
 \int_{\cD_0}^{\cD_t}\frac{x^k\dd x}{y}=-t\,\delta_{k,0}\;.
\end{equation}
The left-hand side of \eqref{eq:abel} contains the differentials $x^k\dd x/y$ for $k=0,\dots,g-1$, which form a
basis for the holomorphic differentials on $\bar\Gamma_f$. Thus the left-hand side of \eqref{eq:abel} is the image of the divisor $\cD_t-\cD_0$ under 
the Abel map, which is (by Abel's Theorem) an isomorphism between the
algebraic Jacobian of $\bar\Gamma_f$, 
consisting of degree zero divisor classes on $\bar\Gamma_f$,  and the analytic
Jacobian of $\bar\Gamma_f$, which is a complex torus, that is 
$$
\Jac (\bar\Gamma_f)\cong H^0\left(\Omega^1_{\bar\Gamma_f}\right)^*/H_1(\bar\Gamma_f).
$$ 
Formula \eqref{eq:abel} then says that the integral
curve of $\rX_1$  
is a straight line in this complex torus, as was to be shown.
\end{proof}

It follows that for, generic initial conditions, the solutions to $\rX_1$ are meromorphic functions in $t$.

\subsection{Genus $g$ solutions to the infinite Volterra and Toda lattices}

The infinite Volterra lattice is given by the 
set of equations
\begin{equation}\label{eq:infinite_volterra} 
  \dot w_n = w_n (w_{n+1}-w_{n-1})\;,\qquad n\in\Z\;.
\end{equation}
It was first considered by Kac and van Moerbeke~\cite{kacvm}, who also studied the $N$-periodic case ($w_{N+n}=w_n$
for all $n$). We now show 
that the Volterra map allows us to define, for any $g$, infinite sequences
of meromorphic functions $(w_n(t))_{n\in\Z}$ which satisfy \eqref{eq:infinite_volterra}. Since these sequences are
defined from solutions of the genus $g$ Mumford-like system, and hence can be written in terms of genus $g$ theta
functions, we will refer to these solutions to the Volterra lattice as \emph{genus $g$ solutions}.

Let $g>0$ be fixed and consider the vector field on $\C^{3g+1}_w$, corresponding to the vector field $\frac12 \rX_1$ 
on the genus $g$ Mumford-like system, via the birational transformation constructed in Section
\ref{par:birat_trans}. For the sake of brevity, let us call this the w-system (in genus $g$). By algebraic integrability,
the vector field $\frac12 \rX_1$ has globally defined meromorphic solutions $w_1(t),\dots,w_{3g+1}(t)$, corresponding to 
generic initial conditions. Using the recursion relation induced by the action of the Volterra map on  $\C^{3g+1}_w$,
we get also globally defined meromorphic functions
$w_n(t)$ for all $n>3g+1$ and all $n\leqslant 0$. Algebraic integrability further implies that the
recursion and the flow of the vector field must commute, as they both correspond to translations on the fibers of the
momentum map, which are affine parts of $g$-dimensional tori. (The fact that the map and the flow 
commute is already a consequence of Liouville integrability.)   It follows that all formulae only  involving the
variables $w_n$ remain valid when all indices are shifted by the same integer. In the proof of the theorem
which follows we will make extensive use of the birational transformation between the w-system and the Mumford-like
system. The triplet corresponding to a meromorphic solution $(w_1(t),\dots,w_{3g+1}(t))$ of the w-system in genus
$g$ will be denoted $(\cP_0(x;t),\cQ_0(x;t),\cR_0(x;t))$, the index $0$ being added because we will also use the
triplets $(\cP_n(x;t),\cQ_n(x;t),\cR_n(x;t))$, obtained from it by applying  the Volterra map or its
inverse several times.  Again, all formulae involving only the polynomials $\cP_n,\cQ_n,\cR_n$, $n\in\Z$, remain valid when all
indices are shifted by the same integer, and for any $n\in\Z$, $(\cP_n(x;t),\cQ_n(x;t),\cR_n(x;t))$ corresponds to 
$(w_{n+1}(t),\dots,w_{3g+n+1}(t))$ 
under the birational map.

\begin{thm}\label{thm:volterra} 
  The sequence of meromorphic functions $\big(w_n(t)\big)_{n\in\Z}$ is a solution to the infinite Volterra lattice
  \eqref{eq:infinite_volterra}.
\end{thm}

\begin{proof}
We first recall the recursion relations \eqref{PQrec2} for the triplets $(\cP_n,\cQ_n,\cR_n)$, which we evaluate at
any meromorphic solution to $\rX_1=\frac12\rX_{H_{g+2}}$: 
\begin{gather}
  \cP_{n+1}(x;t)=\cQ_n(x;t)-\cP_n(x;t)\;,\
    \cQ_{n+1}(x;t)=\frac{2\cP_n(x;t)-\cQ_n(x;t)+\cR_n(x;t)}{-w_{n+1}(t) x}\;,\label{PQrect}\\
  \cR_{n+1}(x;t)=-w_{n+1}(t) x\cQ_n(x;t)\;.\label{Rrect}
\end{gather}
From \eqref{Rrect}, since $H_1=2(\p1+\r1)$ is a first integral, and using $\p{n+1,1}=\q{n,1}-\p{n,1}$, which
follows from the first equation in the recursion relation \eqref{PQrecn}, we have
\begin{equation}\label{eq:w_to_p}
 w_n(t)=-\frac12\r {n,1}(t)=\frac12\p{n,1}(t)-\frac{c_1}4\;,\quad\hbox{and}\quad \r{n,2}(t)=-w_n(t)\q{n-1,1}(t)\;,
\end{equation}
where $c_1$ is a constant. It follows that 
\begin{align*}
  \dot w_n(t)&\stackrel{\eqref{eq:w_to_p}}=\frac12\dot p_{n,1}(t)
    \stackrel{\eqref{eq:pdot}}=\frac12\r{n,2}(t)-\frac14
    \r{n,1}(t)\q{n,1}(t)\stackrel{\eqref{eq:w_to_p}}=\frac{w_n(t)}2(\q{n,1}(t)-\q{n-1,1}(t))\\
    &\;\stackrel{\eqref{PQrect}}=\frac{w_n(t)}2(\p{n+1,1}(t)-\p{n-1,1}(t))
    \stackrel{\eqref{eq:w_to_p}}=w_n(t)(w_{n+1}(t)-w_{n-1}(t))\;,
\end{align*}
as was to be shown.
\end{proof} 

\begin{remark} It is fairly  straightforward to modify the proof of Proposition \ref{abelproof}, and the preceding result, to all of the Hamiltonian vector fields $\rX_i$, associated 
with times $t_i$, $1\leqslant i\leqslant g$, which correspond to the first $g$ flows in the Volterra lattice hierarchy. 
This replaces $t$ by $t_i$ and modifies the Kronecker delta on the right-hand side of  \eqref{eq:abel} to $\delta_{k,i-1}$, hence producing  solutions that are  meromorphic 
in $t=t_1,t_2,\ldots,t_g$. 
\end{remark}

We now apply a standard Miura-like formula, to show how a genus $g$ solution of the Volterra lattice, given by an infinite sequence of meromorphic functions $w_n(t)$, also leads to a corresponding solution to the
infinite Toda lattice, given  by 
\begin{equation}\label{eq:infinite_Toda}
\frac{\rd a_n}{\rd t}  =  a_n(b_{n-1} - b_{n})\;,  \qquad  
  \frac{\rd b_n}{\rd t}  =  a_n - a_{n+1}\;. 
\end{equation}
(These are almost the same as the Flaschka variables for the Toda lattice, except that traditionally $\sqrt{a_n}$ is used in place of $a_n$; 
and similarly, the quantities $\sqrt{w_n}$ are used in \cite{moser}.)  
\begin{cor}\label{thm:Toda} 
  Let $w_n(t)$, $n\in\Z$ be a genus $g$ meromorphic solution to the infinite Volterra lattice
  \eqref{eq:infinite_volterra}. Upon setting, for $j\in\Z$, 
  \begin{equation}\label{eq:moser_map}
    a_{j+1}:=w_{2j-1}w_{2j}\;,\qquad b_{j+1}:=-w_{2j}-w_{2j+1}\;,
  \end{equation}
  the sequence of meromorphic functions $a_j(t),b_j(t)$ 
is a solution to the infinite Toda lattice, 
while another sequence of meromorphic solutions to  \eqref{eq:infinite_Toda} is given for $j\in\Z$ by 
\begin{equation}\label{eq:moser2}
    a^*_{j+1}:=w_{2j}w_{2j+1}\;,\qquad b^*_{j+1}:=-w_{2j+1}-w_{2j+2}\;.
  \end{equation}
\end{cor}

\begin{proof}
Differentiating 
\eqref{eq:moser_map} and using  \eqref{eq:infinite_volterra} one gets immediately
\eqref{eq:infinite_Toda}, and similarly for \eqref{eq:moser2}. 
\end{proof}

\begin{exa} \label{g1wpsoln}
Theorem \ref{thm:volterra} and  Corollary \ref{thm:Toda} imply that we can obtain elliptic (genus 1) 
solutions to the infinite Volterra
and Toda lattices, by starting from a generic solution to
(\ref{eq:rec_g1}). On a fixed orbit of the latter, any such solution can be identifed with an orbit 
of the QRT map (\ref{eq:rec_for_somos}) associated with Somos-5. 
Hence this means that the analytic results of \cite{hones5} can be applied, to write the tau function 
explicitly as  
$$
\tau_n =A_{\pm} B^n \frac{\si (z_0 +nz)}{\si(z)^{n^2}}, 
$$ 
where $A_+,A_-,B$ are non-zero constants (with $A_{\pm}$ chosen according to the parity of $n$), and  
$\si(z)=\si (z; g_2,g_3)$ denotes the Weierstrass sigma function associated with an 
elliptic curve ${y}^2=4{x}^3-g_2{x}-g_3$, 
isomorphic to (\ref{weicub}). 
The parameters $z,g_2,g_3$ all depend on $c_1,c_2,c_3$, while $z_0$ also depends on the initial point 
on the orbit. Then we can write 
the solution of the map explicitly in terms of the Weierstrass zeta function,  as 
\beq\label{wwp} 
w_n = \frac{\si \big(z_0+nz\big) \si \big(z_0+(n+3)z\big) }{\si (z)^4 \si \big(z_0+(n+1)z\big) \si \big(z_0+(n+2)z\big) } 
= \hat{c} \Big( \zeta \big(z_0+(n+2)z\big) -   \zeta \big(z_0+(n+1)z\big) + C \Big), 
\eeq 
where $\hat{c} = {\si (2z)}/{\si (z)^4}$, $ C=   \zeta (z)-  \zeta(2z)$. 
Now extending this by the flow of the vector field $\rX_1$, with parameter $t$, we find that only 
$z_0$ changes, being replaced by $z_0+\hat{c}\,t$ (giving a linear flow on the Jacobian of the elliptic curve).  
Hence we arrive at the genus 1 solution of the Volterra lattice, given for $n\in\Z$ by 
$$
w_n(t)  = \hat{c} \Big( \zeta \big(z_0+\hat{c}\,t+(n+2)z\big) -   \zeta \big(z_0+  \hat{c}\,t+(n+1)z\big) +C \Big) 
$$   
(equivalent to the travelling waves found in \cite{yan}), 
and from \eqref{eq:moser_map} we get a corresponding genus 1 solution of the Toda lattice, that is 
\begin{align}  
a_n(t) & =  \hat{c}^4 \Big( \wp\big(2z\big)  -  \wp\big(z_0 +\hat{c}\,t+(2n-1)z\big)\Big)\;, 
\label{todawp} 
\\ 
b_n(t) & =  \hat{c}\Big(\zeta \big(z_0+\hat{c}\,t+(2n-1)z\big) -\zeta \big(z_0+\hat{c}\,t+(2n+1)z\big) -2C\Big)\;, \nonumber 
\end{align} 
written in terms of  the Weierstrass $\wp$ function, with the constants $\hat{c}$, $C$ as above.  
Note that some more general elliptic solutions of the Volterra lattice, with the form of $w_n$ depending on the parity of $n$, have been 
presented elsewhere in the literature \cite{kitaev, vereshchagin, veselov_volterra}.  
\end{exa} 
\begin{remark} 
Similarly, we can 
produce genus 2 solutions to the infinite Volterra and Toda lattices, by starting from a generic solution to the
map (P.iv), extended to meromorphic functions $w_n(t)$ by the flow of the vector field $\rX_1$.
\end{remark} 

In \cite{moser},  
the transformation \eqref{eq:moser_map} was used 
to connect the finite Volterra and Toda lattices by Moser, 
who attributed it to H\'enon. The same transformation has further been applied to connect 
real-valued solutions of the infinite lattices, subject to suitable  (smoothness/boundedness) 
conditions \cite{ghsz}. 
Moser also employed finite continued fractions in \cite{moser}. 
However, it turns out that the map  \eqref{eq:moser_map} 
has a much earlier origin in the classical theory of continued fractions,  
where it arises from the method of contraction for S-fractions (see J.3 in \cite{stieltjes}, and \cite{shohat}), 
a fact that has perhaps 
been overlooked  
in the integrable systems literature.  In the next subsection, we show how the Volterra maps, 
as presented in this paper, are related to the integrable maps recently constructed by one of us \cite{contfrac}; 
the key is to apply contraction to the S-fraction \eqref{eq:F1}, which produces a J-fraction, and 
thereby yields the 
Miura-type formula \eqref{eq:moser_map}.

\subsection{Contraction of continued fractions}

The following equivalence between a pair of continued fractions, which was introduced in \cite{stieltjes},  is known as \textit{contraction}: 
$$ 
 X -\cfrac{w_1}{1-\cfrac{w_2}{X-\cfrac{w_3}{1-\cfrac{w_4}{X-\cdots}}}}
  =X+b_1-\cfrac{a_2}{X+b_2-\cfrac{a_3}{X+b_3-\cfrac{a_4}{X+b_4-\cdots}}}\;. 
$$ 
The form of the fraction on the left is the original way that an  S-fraction was written by Stieltjes, in terms of a variable 
$X=x^{-1}$, while the fraction on the right is a 
Jacobi continued fraction (J-fraction). To be  more precise, the above equality 
 is an identity of continued fractions, 
obtained by combining successive pairs of adjacent lines in the S-fraction into a single sequence of lines in the J-fraction, 
with the coefficients $a_j,b_j$ on the right being related to $w_j$ on the left by 
\begin{equation}\label{contractform}
  b_1=-w_1\;, \qquad a_{j+1} = w_{2j-1}w_{2j}\;, \qquad b_{j+1}=-w_{2j}-w_{2j+1}\;,
  \qquad \mathrm{for} \quad j\geqslant1\;.
\end{equation} 
(Within the theory of continuants, the formulae for contraction of two or more lines of a general continued fraction are presented in \cite{duv}.) 
To make contact with our previous discussion, 
observe that (\ref{contractform})  reproduces the  transformation \eqref{eq:moser_map} between   
the Volterra and Toda lattices, but for indices $j\geqslant 1$ only. 

In order to see how contraction arises in the context of Volterra maps, we start from a hyperelliptic curve $\Gamma_f$ 
of the form previously considered. 
We take a fixed set of coefficients $c_i$, which are arbitrary except that, as usual, we assume that the polynomial 
\beq\label{fpolyci} 
f(x)=1+\sum_{i=1}^{2g+1}c_ix^i
\eeq 
is square-free with $c_{2g+1}\neq0$, 
so that the hyperelliptic curve $\Gamma_f: \, y^2=f(x)$
is smooth and has genus $g$. 
In order to simplify the presentation below, 
initially  we make the further assumption that $c_1=0$. 
Then setting
\begin{equation}\label{eq:cur_birat}
  X=\frac1x\;,\qquad Y=\frac{y}{x^{g+1}} \; 
\end{equation}
establishes a birational isomorphism between $\Gamma_f$ and an algebraic curve $\cC$ 
which (by completing the square) can be written in the form
\beq\label{ccurve}
{\cal C}: \,\, Y^2=\hat{\rF}(X), \qquad \hat{\rF}(X) = A(X)^2+4R(X), 
\eeq 
where $A(X)$ is a monic polynomial in $X$ of degree $g+1$ with no term of degree $g$ (so 
that the right-hand side of (\ref{ccurve}) has no degree $2g+1$ term), and $R$ is
a polynomial of degree at most $g$ in $X$, not identically zero but otherwise arbitrary; such curves are exactly the ones which were considered in
\cite{contfrac}.

Now let $(\cP,\cQ,\cR)=(\cP_0,\cQ_0,\cR_0)\in M_g$ be a 
generic triplet satisfying $\cP^2+\cQ\cR=f$, in the sense discussed above \eqref{eq:F1}. 
As we have seen in
\eqref{eq:F1}, the associated rational function on $\Gamma_f$, denoted $F_0$, admits the following expansion as an S-fraction:
\begin{equation}\label{eq:F_1_fraction}
  F_0 = \frac{y+\cP_0(x)}{\cQ_0(x)}=1-\cfrac{w_1x}{ 1 -\cfrac {w_2x}{ 1-\cfrac{w_3x}{1- \cdots} } } \;.
\end{equation}
Then, upon multiplying  \eqref{eq:F_1_fraction} by $x^{-1}=X$, rewriting the S-fraction in terms of the new spectral 
parameter $X$, and applying contraction, we find 
\begin{equation}\label{contract}
  x^{-1} F_0 = X -\cfrac{w_1}{1-\cfrac{w_2}{X-\cfrac{w_3}{1-\cfrac{w_4}{X-\cdots}}}}
  =X+b_1-\cfrac{a_2}{X+b_2-\cfrac{a_3}{X+b_3-\cfrac{a_4}{X+b_4-\cdots}}}\;,
\end{equation}
where the J-fraction on the right above is defined to be the contraction of the S-fraction. 
Thus, via the second equality in (\ref{contract}), the coefficients $a_j,b_j$  of the J-fraction 
are specified in terms of the  $w_j$ according to (\ref{contractform}). 

We now briefly recall the construction of integrable maps associated with J-fractions, as presented in \cite{contfrac}. 
The starting point is a rational function $Y_0$ on an even hyperelliptic curve $\cC$ of the form \eqref{ccurve}, 
whose completion $\bar{\cC}$ includes two points $\infty_1,\infty_2$ at infinity. This function 
has $g+1$ simple poles and $g+1$ simple zeros, with one pole being at the point $\infty_1$ (where 
$Y\sim X^{g+1}\sim A(X)$ as $X\to\infty$), and one zero being at $\infty_2$ (where  $Y\sim -X^{g+1}\sim -A(X)$), 
taking the form 
\begin{equation}\label{Y0}
  Y_0 =\frac{Y+P_0(X)}{Q_0(X)}\;,
\end{equation}
for polynomials $P_0$, of degree $g+1$ with no term at ${\cal O}(X^{g})$, and $Q_0$, of degree $g$; and there exists another polynomial 
$Q_{-1}$, of degree $g$, satisfying $Y^2=P_0^2+Q_0Q_{-1}=\hat{\rF}$. The expansion of $Y_0$ around the point $\infty_1$, with $X^{-1}$ as a local parameter, 
can be considered as an element of $\C((X^{-1}))$, and it was shown by van der Poorten (see \cite{vdphyp, vdp3}) 
that this power series admits a J-fraction expansion of the form 
\begin{equation}\label{eq:andy_rec}
  Y_0 = \al_0 + \cfrac{1}{Y_1}
      = \al_0 + \cfrac{1}{\al_1+\cfrac{1}{Y_2}}=\cdots
      = \al_0 + \cfrac{1}{\al_1+\cfrac{1}{\al_2+\cfrac{1}\ddots}} 
\end{equation}
with $\al_n:=\lc Y_{n} \rc$, the polynomial part of $Y_n$. 
Furthermore, for a generic choice of such $P_0$, $Q_0$ in \eqref{Y0}, the polynomial parts $\al_n$ are of degree 1 in $X$ for any $n$, 
and the recursion $Y_n =  \al_n+ \frac{1}{Y_{n+1}}$
leads to a sequence of polynomials $P_n,Q_n$ satisfying the same degree constraints as above, such that  
\begin{equation}\label{YnPQ}
  Y_n=\frac{Y+P_n(X)}{Q_n(X)} = \frac{Q_{n-1}(X)}{Y-P_n(X)}\; \implies Y^2=\hat{\rF}(X) = P_n(X)^2+Q_n(X)Q_{n-1}(X), 
\end{equation}
where the above relations extend to all $n\in\Z$, not just $n\geqslant 0$, by reversing \eqref{eq:andy_rec} to find $Y_{-1}$ from $Y_0$, etc. 

The situation for the J-fraction expansion \eqref{eq:andy_rec} is very similar to that for the expansion \eqref{eq:F_1_fraction}, and 
allows the construction of a birational dynamical system that is defined by a recursion for the polynomials $P_n,Q_n$, 
analogous to the derivation of the Volterra map from the S-fraction in Section \ref{sec:volterra}. 
Here we refer to the dynamical system (in dimension $3g+1$) defined by   \eqref{eq:andy_rec} as the \textit{J-fraction map} in genus $g$, 
denoted ${\cal J}_g$. 
In order to explain the very 
close connection between  ${\cal V}_g$ and ${\cal J}_g$,   
and prove Theorem \ref{contractthm},   we further summarize some  
features of the latter, while referring the reader 
to \cite{contfrac} for a complete description. 

To describe the dynamics defined by $P_n,Q_n$, new variables $u_n$, $d_n$ and $v_n$ are introduced from 
\begin{equation}\label{PQToda}
  P_n(X) = A(X) + 2d_n X^{g-1} +{\cal O}(X^{g-2})\;, \quad Q_n(X) = u_n \big(X^g -
  v_nX^{g-1}+{\cal O}(X^{g-2})\big)\;, 
\end{equation}
so that from the terms of ${\cal O}(X^{2g})$ in the equation for $\hat\rF$ on the  right-hand side of \eqref{YnPQ}, the relation 
\begin{equation}\label{eq:a_n}
  u_n u_{n-1}=-4d_n\neq 0\;
\end{equation}
must hold, while calculating the (degree 1) 
polynomial parts in each line  of \eqref{eq:andy_rec} 
shows that, for any $n$, we have  $\al_n=(X+b_n)/u_n$. 
Thus, upon substituting for $\al_n$ and rescaling each line of the continued fraction using \eqref{eq:a_n}, we may rewrite the 
J-fraction  \eqref{eq:andy_rec} more explicitly as 
$$ 
 {Y_0}=\al_0+\frac{1}{Y_1}
={\frac{2(X+v_0)}{u_0}+
                 \cfrac1{\frac{2(X+v_1)}{u_1}+
                 \cfrac1{\frac{2(X+v_2)}{u_2}-\cdots}}},  
$$
where, by setting $\hat{s}_0=\frac{u_1}{2}$, we have 
\begin{equation}\label{jfrac}
 Y_1=\frac{1}{\hat{s}_0}\left(\,{X+v_1-\cfrac{d_2}{X+v_2-\cfrac{d_3}{X+v_3-\cdots}}}\;\right)\;.
\end{equation}
Then the J-fraction map ${\cal J}_g$ is a dynamical system 
on an affine phase space $\hat{M}_g^{\hat{c}_1=0}$ of dimension $3g+1$, which fibers over the space of curves $\cC$ of the form (\ref{ccurve}), with 
each  (generic) fiber being an affine part of the corresponding Jacobian variety $\Jac(\bar{\cC})$; and  on each such fiber, the map corresponds 
to a translation by the class of the divisor $\infty_2-\infty_1$. 
It is defined recursively by \eqref{YnPQ}, in terms of the coefficients of the 
polynomials $P_n,Q_n$, except that the prefactors $u_n$ in front of each $Q_n$, as  in \eqref{PQToda}, are completely decoupled from the dynamics. 
Indeed, the constant $\hat{s}_0$ in \eqref{jfrac} is arbitrary: it determines the first coefficient in the series expansion of the moment generating function
$1/Y_1=\sum_{j\geqslant 0}\hat{s}_j X^{-j-1}$, whose coefficients allow the solutions of ${\cal J}_g$ 
to be written in terms of tau functions given by Hankel determinants; 
but $\hat{s}_0$ can be removed by a gauge transformation on the tau functions. Moreover, 
once $\hat{s}_0$ is fixed, then $u_1$ and all the other prefactors $u_n$ are determined from $\hat{s}_0$ and $d_n$,  due to \eqref{eq:a_n}; 
and the phase space $\hat{M}_g^{\hat{c}_1=0}$ 
(which is an affine space of Lax matrices) does not include the parameter $\hat{s}_0$. After decoupling from $u_n$,  the 
map ${\cal J}_g$ can be written equivalently 
as a recursion for the remaining coefficients in $P_n,Q_n$, or as  coupled recurrences for the quantities $d_n,v_n$. 
(See (\ref{g1todamap}) and   (\ref{vdg2map}) in the examples below for the cases $g=1$ and $g=2$, respectively.) 

We would now like to identify (\ref{jfrac}) with the J-fraction appearing on the right in (\ref{contract}), but there are two problems: firstly, the relation  
\eqref{contractform} is valid only for $j\geqslant 1$, and gives a different formula for $b_1$ when $j=0$; and secondly, we initially made the assumption that $c_1=0$, which 
does not hold in general. To relax the latter assumption, we must 
shift the spectral parameter $X$, 
make a compensating shift in $b_j$, 
and allow a linear relation between $x^{-1}F_0$ and $Y_1$, yielding a  modification of  \eqref{eq:moser_map}, valid for all $j\in\Z$. 

\begin{propn}\label{todadouble} The odd, genus $g$ spectral curve $\Gamma_f: \, y^2=f(x)$ with $f(x)=1+\sum_{i=1}^{2g+1}c_ix^i$, associated with a
generic orbit of the Volterra map ${\cal V}_g$, is isomorphic to $\cC$, the even spectral curve (\ref{ccurve}) 
for a corresponding orbit of the J-fraction map ${\cal J}_g$, 
via the birational equivalence $X=x^{-1}+c_1/(2g+2)$, $Y=y/x^{g+1}$.
Under this isomorphism of 
curves, the 
function $F_0$ on $\Gamma_f$ 
and the  function $Y_1$ on $\cC$ 
are  related 
by 
\beq\label{funcFtoY} 
x^{-1}F_0-w_0=\hat{s}_0Y_1, 
\eeq  
and the quantities $d_n,v_n$  satisfying the 
map ${\cal J}_g$ are given in terms of the solution of ${\cal V}_g$ 
by 
\beq\label{contractshift} 
d_{j+1} = w_{2j-1}w_{2j}\;, \qquad v_{j+1}=-w_{2j}-w_{2j+1}-\frac{c_1}{2(g+1)}\;,
  \qquad \mathrm{for} \quad j\in\Z\;.
\eeq
Hence each iteration on the orbit of ${\cal J}_g$ 
corresponds to two iterations on the corresponding orbit of ${\cal V}_g$. 
\end{propn} 
\begin{proof}
The shift in $X$ in the birational transformation, as in  the formula $X=x^{-1}+c_1/(2g+2)$,  
ensures that the equation $Y^2=\hat{\rF}(X)=A(X)^2+4R(X)$ 
for $\cC$ has no term at ${\cal O}(X^{2g+1})$ in $\hat{\rF}(X)$, so that $A(X)=X^{g+1}+{\cal O}(X^{g-1})$, as required 
for a spectral curve of the J-fraction map. Also, from the explicit form of the function $F_0$ in \eqref{eq:F_1_fraction}, 
we may rewrite the left-hand side of (\ref{funcFtoY})  in terms of $Y$ and $\tX=X-c_1/(2g+2)$, as 
$$ 
\frac{\tX\big(\tX^{-(g+1)}Y+\cP_0(\tX^{-1})\big)-w_0\cQ_0(\tX^{-1})}{\cQ_0(\tX^{-1})} 
=\frac{Y+P_1(X)}{Q_1(X)/\hat{s}_0} 
$$ 
where we calculate 
$P_1(X)=\tX^{g+1}\cP_0(\tX^{-1})-w_0\tX^g\cQ_0(\tX^{-1})=\tX^{g+1}+(p_{0,1}-2w_0)\tX^g+{\cal O}(\tX^{g-1})$, 
and then in view of \eqref{eq:wn_to_p} we see that $P_1(X)=X^{g+1}+ {\cal O}(X^{g-1})$, while 
$Q_1(X)=\hat{s}_0 \tX^g\cQ_0(\tX^{-1})=\tfrac{u_1}{2}\big(2\tX^g+{\cal O}(\tX^{g-1})\big)=u_1\big(X^g+{\cal O}(X^{g-1})\big)$, 
so both of $P_1$ and $Q_1$ are polynomials in $X$ of the required form for the J-fraction map.  
Now from the S-fraction in (\ref{funcFtoY}), we find that combining contraction with the 
shift of spectral parameter 
modifies 
\eqref{contract}, so that, 
in terms of fractions in $X$,  $x^{-1} F_0 -w_0$ is equal to 
$$ 
 X-\tfrac{c_1}{2g+2}-w_0 -\cfrac{w_1}{1-\cfrac{w_2}{X-\tfrac{c_1}{2g+2}-\cfrac{w_3}{1-\cfrac{w_4}{X-\tfrac{c_1}{2g+2}-\cdots}}}}
  =X+v_1-\cfrac{d_2}{X+	v_2-\cfrac{d_3}{X+v_3-\cfrac{d_4}{X+v_4-\cdots}}}\;,
$$ 
where we have inserted the formula for $Y_1$ from (\ref{jfrac}), and cancelled the arbitrary constant $\hat{s}_0=\frac{u_1}{2}$. 
Comparing the first line of the above fractions on each side, we see that $v_1=-w_0-w_1-c_1/(2g+2)$, which is the correct form of the relation 
for $v_1$ in (\ref{contractshift}) when $j=0$, and contraction of the subsequent lines on the left give these expressions for $d_{j+1},v_{j+1}$ for all $j\geqslant 1$. 
One can also shift the fraction on the left down by two lines, to get a relation between $F_2$ and $Y_2$, and the fraction can be inverted to find 
a relation between $F_{-2}$ and $Y_0$; so, by continuing down or up in this way, we find that $x^{-1}F_{2j} -w_{2j}=\tfrac{u_{j+1}}{2} Y_{j+1}$ holds for 
all $j\in\Z$, extending   (\ref{contractshift}) to negative $j$ as well. In all expressions, a shift of indices $j\to j+1$ gives a single iteration 
of ${\cal J}_g$, but all indices of the Volterra variables increase by 2, giving two iterations of ${\cal V}_g$.   
\end{proof} 

\begin{cor}
Under the action of the Hamiltonian vector field $\rX_1$, each generic solution of ${\cal V}_g$ 
produces a genus $g$ solution of the Toda lattice equation \eqref{eq:TodaL} which also satisfies the 
map ${\cal J}_g$, 
via the transformation (\ref{contractshift}).
\end{cor}
\begin{proof}
This follows immediately  from the preceding result, by applying Theorem \ref{thm:volterra}, and noting that the result of 
Corollary \ref{thm:Toda} still stands if we set $a_n=d_n$, $b_n=v_n+c_1/(2g+2)$ for all $n$.  
\end{proof} 
The main results of this subsection are collected in the following statement. 
\begin{thm}\label{contractthm} 
By contraction of the S-fraction  \eqref{eq:F_1_fraction} 
for the associated rational function $F_0$,
a generic orbit 
of the Volterra map  ${\cal V}_g$ corresponding to a fixed odd spectral curve 
$\Gamma_f: \, y^2=f(x)$ of genus $g$,  
for square-free $f(x)$ as in \eqref{fpolyci}, 
is transformed to an orbit of the integrable map ${\cal J}_g$ constructed in \cite{contfrac} from the J-fraction \eqref{jfracvdp}, 
with the even spectral curve   
$\cC$ given by  (\ref{ccurve}) with $\hat{\rF}(X)=X^{2g+2}+\sum_{j=2}^{2g+2}\hat{c}_{j}X^{2g+2-j}$. 
For  coefficients $\hat{c}_j$ given suitably in terms of $c_i$, there is a birational   
equivalence between 
$\Gamma_f$ and $\cC$, given by 
 \beq\label{birat} X=x^{-1}+\frac{c_1}{2(g+1)}, \qquad Y=\frac{y}{x^{g+1}}. \eeq
Moreover, the translation 
on $\Jac (\bar{\cC})$ associated with a single iteration of the 
J-fraction map 
corresponds to twice the  translation on $\Jac(\bar{\Gamma}_f)$ 
associated with each iteration of ${\cal V}_g$. 
In fact, 
each generic orbit of ${\cal V}_g$ is related to two different orbits of 
${\cal J}_g$ in this way. 
\end{thm} 
\begin{proof} 
The main statements in the theorem were already proved in Proposition \ref{todadouble}.  
For the relation between shifts on complex tori, note that in $M_g$, we have that $\mu^{-1}(f)$, the fiber over a generic curve $\Gamma_f$, 
is an affine part of $\Jac(\bar{\Gamma}_f)$, while in $\hat{M}_g^{\hat{c}_1=0}$ the fiber over $\cC$ is an affine part of $\Jac(\bar{\cC})$; 
but then the isomorphism (\ref{birat}) between these two spectral curves means that  $\Jac(\bar{\Gamma}_f)\cong \Jac(\bar{\cC})$. 
It was shown in \cite{contfrac} that each iteration of the J-fraction map ${\cal J}_g$ 
corresponds to a translation by the class of the divisor $\infty_2-\infty_1$ on $\Jac {\bar{\cC}}$, where
$\infty_{1,2}$ are the two points at infinity on ${\bar{\cC}}$, and these are equivalent to the points $(0,\pm 1)$ on
(\ref{weier}). So in terms of 
$\Jac(\bar{\Gamma}_f)$, this is a translation by the class 
of the divisor $(0,-1)-(0,1)=2(0,-1)-(0,-1)-(0,1)\sim 2 \big(
(0,-1)-\infty\big)$, that is by $2[(0,-1)-\infty\big]$,  twice the shift corresponding to the Volterra map $\cV_g$ 
(as found in the proof of Theorem \ref{prp:volterra_aci}). 
Finally, notice that in Corollary \ref {thm:Toda} there is the second, alternative formula \eqref{eq:moser2}, with the indices 
on all $w_i$ 
shifted one step forwards. This corresponds to the freedom to start the 
contraction procedure one line lower in the S-fraction  \eqref{contract}, beginning 
with $F_1$ rather than $F_0$; so the indices on all Volterra variables must be shifted 
by the same amount, and 
the relation \eqref{funcFtoY} with the corresponding 
rational function on $\cC$ is modified to 
$x^{-1}F_1 - w_1 =\hat{s}_0 Y_1$.   
Then, in terms of the J-fraction coefficients, this produces 
\beq\label{contractshift2} 
d_{j+1} = w_{2j}w_{2j+1}\;, \qquad v_{j+1}=-w_{2j+1}-w_{2j+2}-\frac{c_1}{2(g+1)}\;,
  \qquad \mathrm{for} \quad j\in\Z\;.
\eeq
Thus each orbit of $\cV_g$ 
is transformed to two different orbits of ${\cal J}_g$, 
since the resulting orbit of the latter map remains the same when the index on the $w_i$ 
in (\ref{contractshift2}) is shifted by a multiple of 2. 
\end{proof} 
%
\begin{remark} 
In \cite{contfrac},  the phase space $\hat{M}_g^{\hat{c}_1=0}$ for the J-fraction map ${\cal J}_g$ is obtained from $\hat{M}_g$, an affine 
space of dimension $3g+2$ with a specific Poisson structure $\PB$, by restricting to a subvariety defined by setting the value of one of the Casimirs to zero. 
(This is analogous to the situation described in Appendix A, Section \ref{par:even_mumford}.) Although Theorem \ref{contractthm} has been stated in terms of a 
correspondence between specific orbits of ${\cal V}_g$ and  ${\cal J}_g$, the considerations in the proof make it clear that, since the coresponding generic 
fibers 
are birationally equivalent, the restriction of  ${\cal J}_g$ to each fiber is (conjugate to) the square $({\cal V}_g)^2={\cal V}_g\circ{\cal V}_g$. 
This gives a strong hint that $\hat{M}_g^{\hat{c}_1=0}$ and $M_g$ should also be birationally equivalent, making this into a global statement about the two  
maps. While this global statement is by no means obvious, the first example below shows that it is correct when $g=1$; but it is not true 
as a Poisson isomorphism, at least for the specific Poisson structure introduced in \cite{contfrac}. 
The problem of 
finding alternative Poisson structures on $\hat{M}_g$, and  making all these statements precise, is best left for future work. 
\end{remark} 
\begin{exa} 
When $g=1$, with the cubic $\Gamma_f: \, y^2=1+c_1x+c_2x^2+c_3x^3$, the transformation (\ref{birat}) is  
$$ 
X=x^{-1}+\frac{c_1}{4}, \qquad Y=\frac{y}{x^2} 
\quad \implies \quad  \cC: \quad Y^2 =(X^2+\hat{f})^2+4\hat{u}X+4\hat{h}, 
$$ 
where the  quartic curve $\cC$ is 
written   
in terms of the parameters  
\beq\label{fuh_coeffs}
\hat{f}=\frac{1}{2}c_2-\frac{3}{16}c_1^2, \quad 
\hat{u}=\frac{1}{4}c_3-\frac{1}{8}c_1c_2+\frac{1}{32}c_1^3, \quad  
\hat{h}=\frac{1}{16} (c_1^2c_2-c_1c_3-c_2^2) -\frac{3}{256}c_1^4. 
\eeq 
Under the transformation (\ref{contractshift}), 
solutions of the $g=1$ Volterra map ${\cal V}_1$, given by \eqref{eq:rec_g1}, 
or equivalently by \eqref{eq:g=1_red_H3} (with fixed $c_3$), or by 
\eqref{eq:g=1_red_H2} (with fixed $c_2$), are mapped to 
solutions of the corresponding J-fraction map, 
which (according to the results in Example 3.2 in \cite{contfrac}) can be written as a 2D map 
defined by 
\begin{align} 
\label{g1todamap} 
d_{n+1} &=-d_n-v_n^2-\hat{f}\;, \\
v_{n+1} & = -v_n+\frac{\hat{u}}{d_{n+1}}\;, \nonumber
\end{align} 
on a reduced phase space with  fixed parameters  $\hat{f},\hat{u}$, 
which are determined from (\ref{fuh_coeffs}) in terms of 
the values of the constants $c_1,c_2,c_3$ for the solution of the 
map ${\cal V}_1$. The map (\ref{g1todamap}) has the conserved quantity 
$$ 
\hat{H}=d_n(v_n^2+d_n+\hat{f})-\hat{u}v_n,
$$ 
which, on the orbit corresponding to 
a fixed solution of ${\cal V}_1$,  takes the value $\hat{H}=\hat{h}$ given in (\ref{fuh_coeffs}). 
The Poisson bracket presented for the J-fraction maps in \cite{contfrac} 
becomes the canonical bracket 
\beq\label{canonbr}
\{v_n,d_n\}=1
\eeq
on the 2D phase space of the map (\ref{g1todamap}), with coordinates 
$(d_n,v_n)$, and it can be verified directly that the vector field 
$\{ \cdot , \hat{H}\}$ extends to the Toda lattice flow 
(\ref{eq:TodaL}) for all $n\in\Z$ under the action of  this J-fraction map. 
Upon comparing with Example \ref{g1wpsoln}, it is clear that 
the analytic expressions for $a_n,b_n$ in \eqref{todawp} provide explicit formulae 
for the solutions of both the map \eqref{g1todamap} and the Toda lattice, by setting 
$d_n(t)=a_n(t)$ and $v_n(t)=b_n(t)  -\tfrac{1}{4}c_1$. Upon comparison 
of \eqref{wwp} with \eqref{todawp}, it can be seen that each iteration of  ${\cal V}_1$ 
gives a shift by $z$ on the Jacobian of the elliptic curve, while each iteration 
of \eqref{g1todamap} produces a shift by $2z$.

However, the bracket \eqref{canonbr} 
cannot be a reduction of any of the Poisson structures  in 
Example \ref{g1ex}, because the parameters $\hat{f},\hat{u}$ 
do not correspond to 
Casimirs of any of these brackets on the phase space $M_1$. (This is immediately 
obvious for the pencil of brackets generated by $\PB^x$ and $\PB^{x^2}$, while 
a short calculation shows this to be the case for \eqref{eq:quadratic_Poisson} as 
well.) Nevertheless, 
it is still  possible to interpret  (\ref{contractshift})  
as a Poisson map in terms of members of the pencil $\PB^\phi$, 
with $c_1$ being the fixed value of a Casimir. 
For instance, using the Nambu-Poisson structure, 
we can 
construct a birational transformation between the symplectic leaves 
of $\PB^{x^2}$, that is 
$M_1^c\cap(H_2=c_2)$, and a 2D phase space for the 
J-fraction map with coordinates $(d_n,v_n)$,  with $c_1,c_2$ viewed as fixed parameters, 
on which the bracket reduces to 
$$ 
\{v_n,d_n\}^{x^2}=v_n-\frac{1}{4}c_1.     
$$ 
%

Now we take a particular numerical example, 
with the 
elliptic curve $y^2=1-4x+4x^3$ ($c_1=-4$, $c_2=0$, $c_3=4$), and $w_1=1$, $w_2=2$. 
The sequence $(w_n)$ extends backwards to $n\leqslant 0$ to give a singular orbit of  \eqref{eq:rec_g1}, 
with the same singularity pattern appearing as was found for (P.iv) in Section \ref{sec:singularity}: 
$$ 
\ldots,\tfrac{7}{15}, \tfrac{10}{3}, -\tfrac{3}{2},-\tfrac{1}{2},2,1,0,\infty,\infty,0,1,2,-\tfrac{1}{2},-\tfrac{3}{2},\tfrac{10}{3},\tfrac{7}{15},\ldots .
$$
This orbit is symmetrical, in the sense that $w_{-n}=w_{n-3}$ for all $n\in\Z$. 
By applying the formulae in Example \ref{genus1hankel}, the solution is expressed in 
terms of Hankel determinants $\Delta_n$ constructed from the moment sequence  
determined from $s_1=1$, $s_2=2$, 
$s_j = \sum_{i=1}^{j-1}s_is_{j-i} -\sum_{i=1}^{j-2}s_is_{j-i-1}$ ($j\geqslant 3$), that is 
$(s_j)_{j\geqslant 1}:\,1,2,3,6,14,37,105,312,956,2996,9554, \ldots$, which gives 
$$ 
(\Delta_n)_{n\geqslant -2}: \quad 1,1,1,1,2,-1,3,-5,7,-4,23,29,59,129,314,\ldots. 
$$ 
With $\Delta_{-3}=0$, this extends backwards to a sequence of 
tau functions $\tau_n =\Delta_{n-3}=(-1)^{n+1}\tau_{-n}$ for $n\in\Z$;
hence from \eqref{eq:c_to_c'} and (\ref{s5})  we find $c_2'=-1$, $c_3'=1$, so 
for all $n$  
they satisfy the Somos-5 relation 
$$ 
\tau_{n+5}\tau_n = -\tau_{n+4}\tau_{n+1}+\tau_{n+3}\tau_{n+2}.
$$
Applying (\ref{birat}) in this case  produces the quartic curve 
$Y^2=(X^2-3)^2-4(X+2)$, the same one as in Example 4.2 from 
\cite{contfrac}, and  the 
contraction formulae (\ref{contractshift}) and (\ref{contractshift2}) with $j=1$ 
give initial points on two different orbits of the map \eqref{g1todamap} with 
parameters $\hat{f}=-3$, $\hat{u}=-1$, namely 
$(d_2,v_2)=(2,-\tfrac{1}{2})$ and $(-1,3)$, respectively, which 
both correspond to the value $\hat{h}=-2$ for the conserved quantity of this map. 
For these two different orbits, we find that 
$$ 
d_n = \frac{\hat{\tau}_{n-1}\hat{\tau}_{n+1}}{\hat{\tau}_n^2}, 
$$  
where $\hat{\tau}_n = \Delta_{2n-4}$ (even index Hankel determinants) for the first orbit, 
and  $\hat{\tau}_n = \Delta_{2n-3}$  (odd index Hankel determinants) for the second one. It follows from 
Proposition 5.1 in \cite{contfrac} that either of these even/odd index subsequences ($1,1,2,3,7,23,59,314,\ldots$ 
and $1,1,-1,-5,-4,29,129,\ldots$, respectively) must satisfy the same Somos-4 relation, in this case 
the original one introduced by Somos \cite{somos}:
$$ 
\hat{\tau}_{n+4}\hat{\tau}_{n}= \hat{\tau}_{n+3}\hat{\tau}_{n+1}+\hat{\tau}_{n+2}^2. 
$$ 
This connection between Somos-5 and Somos-4 has already been exploited 
elsewhere (in \cite{hones5}, and in \cite{chang}, for example); however, the specific Hankel determinants for 
Somos-4 obtained here via contraction differ from the ones found in \cite{chang}, and also 
from the ones derived directly from the J-fraction in \cite{contfrac}. 
\end{exa} 
\begin{exa}
An orbit of the map (P.iv), with the spectral curve 
\eqref{wg2},  can be 
transformed to an orbit of the $g=2$ J-fraction map, as discussed in Example 3.3 of \cite{contfrac}, 
which has an associated sextic curve, related to it via $X=x^{-1}+\nu/3$, $Y=y/x^3$, 
that can be taken in the form 
$$ 
\cC: \quad Y^2=(X^3+\hat{f}X+\hat{g})^2 +4(\hat{u}X^2+\hat{h}_1X+\hat{h}_2).  
$$ 
Defined on a 4D phase space with coordinates $(d_{n-1},d_n,v_{n-1},v_n)$ 
and depending on the 3 parameters $\hat{f},\hat{g},\hat{u}$, 
the map in \cite{contfrac} is given by 
\begin{align} 
d_{n+1}+d_n+d_{n-1} +\hat{u}/d_n +v_n^2+v_nv_{n-1}+v_{n-1}^2+\hat{f} &=0,  
\label{vdg2map} \\ 
(2v_n+v_{n-1})d_n + (2v_n+v_{n+1})d_{n+1} +v_n^3+\hat{f}v_n+\hat{g} & =0.  
\nonumber 
\end{align} 
More precisely, there are two different orbits of (\ref{vdg2map}) obtained from each 
orbit of (P.iv), depending on whether the formula 
(\ref{contractshift}) or (\ref{contractshift2}) is applied. 
Then, upon writing each term $w_j$ 
satisfying (P.iv) as a ratio of the Hankel determinants $\Delta_j$ defined in Example \ref{genus2hankel}, we see that 
the quantities $d_n$ that appear in the solution of the J-fraction map, as above, are given as a ratio of tau functions, in two different ways: 
$$ 
d_n = \frac{\hat{\tau}_{n-1}\hat{\tau}_{n+1}}{\hat{\tau}_n^2}, \qquad   
\hat{\tau}_n = \Delta_{2n-4} \quad \mathrm{or}\quad \Delta_{2n-3}, 
$$ 
for $n\in\Z$, where the orbit is determined by the choice of parity of the index on $\Delta_j$. It follows 
from the proof of Theorem 5.5 in \cite{contfrac} that (regardless of which choice is made), 
the tau functions $\hat{\tau}_n$ satisfy a Somos-8 relation, which explains why 
the relation (\ref{s8iv}) appears in Example \ref{experitau}.    
\end{exa} 

\section{Conclusions and outlook}\label{sec:conc} 
\setcounter{equation}{0}

We have seen how the map (P.iv) obtained by Gubbiotti et al.\  can 
naturally be viewed as the $g=2$ member of a family of algebraically integrable  maps, 
defined for each $g$, that are naturally related 
to the infinite Volterra  lattice equation, leading to genus $g$ solutions of the latter. 
This begs the question as to what can be said about the other Liouville integrable maps in 4D found in \cite{gjtv2}, namely (P.v) and (P.vi), which take the same form (\ref{themap}) 
but for a different rational function   $\rF$. Note that (P.vi) depends on an 
extra parameter compared with (P.v), which we denote here by $\de$ (instead of $\de^2$ in 
 \cite{gjtv2}). In fact, (P.v) arises from (P.vi) in the limit $\de\to 0$.  

So far we have made the following observations: 
\begin{itemize}
\item Each solution $\hat{w}_n$ of  (P.v) is mapped to a solution $w_n$ of (P.iv) 
via the transformation 
\beq\label{pvtrans}
w_n = \hat{w}_{n+1}\hat{w}_n,  
\eeq 
by suitably identifying the parameters $\nu,a,b$ for (P.iv) in terms of the parameters and first integrals 
for (P.v).   
\item Under the flow of  the Hamiltonian vector field $\tfrac{\rd}{\rd t}$ 
associated with one of its first integrals, the sequence  $(\hat{w}_n)_{n\in\Z}$ 
generated by iteration of the map (P.v) 
extends to a solution of the modified Volterra lattice in the form 
\beq\label{modvol} 
\frac{\rd\hat{w}_n}{\rd t} = \hat{w}_n^2 (\hat{w}_{n+1}-\hat{w}_{n-1}). 
\eeq 
\item Each solution $\hat{w}_n$ of  (P.vi) is mapped to a pair of solutions $w_n^{(+)}$, 
$w_n^{(-)}$ of (P.iv), via 
the transformations  
\beq\label{pvitrans}
w_n^{(\pm)} = (\hat{w}_{n+1}\pm \de)(\hat{w}_n\mp\de),  
\eeq 
by suitably identifying the parameters $\nu,a,b$ for (P.iv) in terms of the parameters and first integrals 
for (P.vi).   
\item Under the flow of  the Hamiltonian vector field $\tfrac{\rd}{\rd t}$ 
associated with one of its first integrals, the sequence  $(\hat{w}_n)_{n\in\Z}$ 
generated by iteration of the map (P.vi) 
extends to a solution of the  modified Volterra lattice in the form 
\beq\label{modvolde} 
\frac{\rd \hat{w}_n}{\rd t} = (\hat{w}_n^2 -\de^2) (\hat{w}_{n+1}-\hat{w}_{n-1}). 
\eeq 
\end{itemize}  
The formulae (\ref{pvtrans}) and (\ref{pvitrans}) are the well-known expressions for the Miura transformation connecting the two forms of the  modified Volterra lattice equation, given by  (\ref{modvol}) and (\ref{modvolde}), respectively,  
to the Volterra lattice \eqref{eq:infinite_volterra}. Thus the above statements about the connections between the maps can be viewed as restrictions of a Miura 
 transformation to a finite-dimensional phase space. Preliminary calculations, and 
the results of \cite{yan} on elliptic solutions, indicate that this picture should extend 
to arbitrary genus $g$.  Our initial results, including an explicit description of how both (P.v) and (P.vi) are related to (P.iv), have recently appeared in \cite{ocnmp}.   
We propose that the full description of the above observations, and their extension to genus $g$ analogues 
of the maps (P.v) and (P.vi), should be left as the subject of future work. 

It is also worth pointing out that part of the original motivation for  the work 
in  \cite{gjtv2} was to consider autonomous versions of the higher order discrete 
Painlev\'e equations from \cite{cj}, and  new applications of the latter have been found very recently.  
Non-autonomous analogues of the Volterra maps ${\cal V}_g$ have been considered in the context of Hermitian matrix models \cite{ben}, where they arise as string equations, and they also appear as recursion relations for orthogonal polynomials associated with generalised 
Freud weights of higher order \cite{cjl}. In these applications, the algebro-geometric structure of the Volterra maps should be relevant to the  asymptotic description of the oscillatory behaviour that is observed in specific parameter regimes. 

\section{Appendix A: the relation between the Mumford-like and the even Mumford system}\label{par:even_mumford}
\setcounter{equation}{0}

Here we now show that the Hamiltonian system $(M_g,\PB^\phi,\mu)$ associated with the Volterra map is birationally isomorphic 
(as a Poisson isomorphism) to  the even Mumford system, or more precisely to a subsystem thereof, obtained 
simply 
by fixing the value of one of the Casimirs. Recall from
\cite[Ch.\ 6]{vanhaecke} that the even Mumford system (of genus $g$) is the Hamiltonian system
$(M_g',\PB'{}^\psi,\mu')$, whose phase space $M_g'$ is the $(3g+2)$-dimensional affine space
\begin{equation}
  M_g':=\left\{(U(\xi),V(\xi),W(\xi))\in\C[\xi]^3  ~\Bigl|~
  \begin{array}{lll}
    \deg U(\xi)=g\;,&U\hbox{ monic}\\
    \deg V(\xi) <g\;,&\\
    \deg W(\xi)=g+2\;,&W\hbox{ monic}\\
  \end{array}
 \right\}\;.
\end{equation}
Elements $(U(\xi),V(\xi),W(\xi))$ of $M_g'$ are written as Lax matrices
\begin{equation*}
  \lax'(\xi):=\left(\begin{array}{lr}
  V(\xi)&U(\xi)\\
  W(\xi)&-V(\xi)
  \end{array}\right)\;,\qquad
\end{equation*}
whose polynomial entries have the form 
\begin{equation*}
  U(\xi)=\xi^g+\sum_{i=0}^{g-1} U_i\xi^i\;,  \qquad
  V(\xi)=\sum_{i=0}^{g-1}V_i\xi^i\;,  \qquad
  W(\xi)=\xi^{g+2}+\sum_{i=0}^{g+1}W_i\xi^i\;.  
\end{equation*}
The $3g+2$ coefficients $W_{g+1},W_g$ and $U_i,V_i,W_i$ with $0\leqslant i<g$ are used as linear coordinates on
$M_g'$. The momentum map $\mu'$ is given by
\begin{equation*}
  \begin{array}{lcrcl}
    \mu'&:&M_g'&\to&\C[\xi]\\
    & &\lax'(\xi)=\left(\begin{array}{lr}
      V(\xi)&U(\xi)\\
      W(\xi)&-V(\xi)
  \end{array}\right)&\mapsto&-\det \lax'(\xi)=V(\xi)^2+U(\xi)W(\xi)\;.
  \end{array}
\end{equation*}
It is clear that $-\det \lax'(\xi)$ is monic of degree $2g+2$, so $2g+2$ polynomial functions
$H'_0,H'_1,\dots,H'_{2g+1}$ on~$M_g'$ are defined by
\begin{equation*}
  V(\xi)^2+U(\xi)W(\xi)=\xi^{2g+2}+\sum_{i=0}^{2g+1}H'_i\xi^i\;,
\end{equation*}
For any non-zero polynomial $\psi$ of degree at most $g+1$, a Poisson structure of rank $2g$ on $M_g'$ is given by
\begin{align*}
  \pb{U(\xi),U(\eta)}'{}^\psi&=\pb{V(\xi),V(\eta)}'{}^\psi=0\;,\\
  \pb{U(\xi),V(\eta)}'{}^\psi&=\frac{U(\xi)\psi(\eta)-U(\eta)\psi(\xi)}{\xi-\eta}\;,\\
  \pb{U(\xi),W(\eta)}'{}^\psi&=-2\frac{V(\xi)\psi(\eta)-V(\eta)\psi(\xi)}{\xi-\eta}\;,\\
  \pb{V(\xi),W(\eta)}'{}^\psi&=\frac{W(\xi)\psi(\eta)-W(\eta)\psi(\xi)}{\xi-\eta}
     -(\xi+\eta+W_{g+1}-U_{g-1})U(\xi)\psi(\eta)\;,\\
  \pb{W(\xi),W(\eta)}'{}^\psi&=2(\xi+\eta+W_{g+1}-U_{g-1})\left(V(\xi)\psi(\eta)-V(\eta)\psi(\xi)\right)\;.
\end{align*}
The polynomial functions $H'_i$ are functionally independent and in involution, which accounts for the Liouville
integrability of the even Mumford system. It is algebraically integrable, with the fiber of $\mu'$ over any monic
polynomial $f'(x)$ of degree $2g+2$ and without multiple roots being an affine part of the Jacobian of the smooth
hyperelliptic curve defined by $\eta^2=f'(\xi)$.

For the isomorphism, we consider only the Poisson structures $\PB'{}^\psi$ for which $\psi(0)=0$. They admit
$H'_0$ as Casimir function, hence we can restrict $(M_g',\PB'{}^\psi,\mu')$ to the subvariety, defined by
$H'_0=0$; the resulting system is denoted by $(M_g^0,\PB^{0,\psi},\mu^0)$. We show that this system is
birationally equivalent to the Mumford-like system $(M_g,\PB^\phi,\mu)$, where the relation between $\psi$ and
$\phi$ will be spelled out below. To do this, we construct a biregular map $\tilde\Psi$ and a birational map $\Psi$
making the following diagram commutative:
\begin{equation*}
  \begin{tikzcd}[row sep=5ex, column sep=4ex]
      M_g^0\arrow{r} {\Psi}\arrow[swap]{d}{\mu^0}&M_g\arrow{d} {\mu}\\
      B_g^0\arrow[swap]{r} {\tilde\Psi}&B_g
    \end{tikzcd}
\end{equation*}
In this diagram, $B_g$ and $B_g^0$ are the images of $\mu$ and $\mu^0$, which we view as spaces of curves: $B_g$
consists of the hyperelliptic curves of the form $y^2=f(x)$, with $f(0)=1$ and $\deg f\leqslant 2g+1$, while
$B_g^0$ consists of the hyperelliptic curves of the form $\eta^2=f'(\xi)$, with $f'$ monic of degree $2g+2$,
vanishing at $0$. The curves of $B_g^0$ have two points at infinity, which we denote by $\infty_1$ and $\infty_2$.

We first establish a natural correspondence between the curves of $B_g^0$ and the curves of $B_g$. Let $y^2=f(x)$
be a curve of $B_g$ and substitute 
$x=\xi^{-1}$ and $y=\eta\xi^{-g-1}$, to get
$\eta^2=f'(\xi)=\xi^{2g+2}f(\xi^{-1})$, where $f'(0)=0$ and $f'$ is monic of degree $2g+2$, so $\eta^2=f'(\xi)$ is
a curve of $B_g^0$. From the latter, one gets back $y^2=f(x)$ by setting $\xi=x^{-1}$ and $\eta=yx^{-g-1}$. Notice
that when $f(x)=1+\sum_{i=1}^{2g+1}c_ix^i$ then $\eta^2=\xi^{2g+2}+\sum_{i=1}^{2g+1}c_{2g+2-i}\xi^i$, which yields
the biregular map $\tilde\Psi$.

For the construction of $\Psi$, the biregular map $\tilde\Psi$ 
between the spaces of curves $y^2=f(x)$ and $\eta^2=f'(\xi)$ is
extended to divisors on these curves. To do this, we compare the description of points $(\cP(x),\cQ(x),\cR(x))$
on a generic fiber of $\mu$ in terms of divisors on the corresponding curve $y^2=f(x)$ with the description of
points $(U(\xi),V(\xi),W(\xi))$ on a generic fiber of $\mu'$ in terms of divisors on the corresponding curve
$\eta^2=f'(\xi)$. The first description was given in the proof of Proposition \ref{prp:generic_fibers}, while the
second description, which we quickly recall, can be found in \cite[Ch.\ 6]{vanhaecke}. Let $\xi_1,\dots,\xi_g$
denote the roots of $U(\xi)$ and let $\eta_i:=V(\xi_i)$, for $i=1,\dots,g$. Then the points $(\xi_i,\eta_i)$ belong
to the curve $\eta^2=\phi(\xi)$ and so the divisor class $\left[\sum_{i=1}^g(\xi_i,\eta_i)-g\infty_1\right]$ is a
point of its Jacobian.  The relation between the polynomials $\cQ(x)$ and $U(\xi)$ is clearly given by
\begin{equation*}
  \frac{\xi^g}2\cQ(\xi^{-1})=U(\xi)\;.
\end{equation*}
Indeed, both sides of this equality are polynomials of degree $2g+2$ which vanish for $\xi=\xi_i=x_i^{-1}$ (which
we may assume to be different). Similarly, $\cP(x)$ is related to $U(\xi)$ and $V(\xi)$ by
\begin{equation*}
  \xi^{g}\cP(\xi^{-1})-U(\xi)=\frac{U_0V(\xi)-V_0U(\xi)}{\xi U_0}\;,
\end{equation*}
because both sides of this equation are the unique polynomial in $\xi$ of degree less than $g$ which takes for
$\xi=\xi_i$ the value $\eta_i/\xi_i$.  A formula for $\cR$ follows from the equations of the curves, namely 
\begin{equation*}
  \cP^2(\xi^{-1})+\cQ(\xi^{-1})\cR(\xi^{-1})=f(\xi^{-1})=\xi^{-2g-2}f'(\xi)=\xi^{-2g-2}(V^2(\xi)+U(\xi)W(\xi))\;.
\end{equation*}
It is clear that this defines a birational map between $M_g$ and $M_g^0$. Under this map, $H_i$ corresponds to
$H'_{2g+2-i}$, for $i=1,\dots,2g+1$. To see that $\Psi$ is a Poisson map, hence a birational Poisson isomorphism,
we recall that the Poisson bracket $\PB'{}^{\psi}$ is given in terms of the $\xi_i$ and $\eta_j$ by
$\pb{\xi_i,\eta_j}'{}^\psi=\psi(\xi_i)\delta_{ij},$ and hence
\begin{equation*}
  \pb{x_i,y_j}'{}^\psi=\left\{\xi_i^{-1},\eta_j\xi^{-g-1}_j\right\}'{}^\psi
  =-\xi_i^{-2}\xi_j^{-g-1}\pb{\xi_i,\eta_j}'{}^\psi
  =-{\xi_i^{-g-3}}\psi(\xi_i)\delta_{ij}
  =-{x_i^{g+3}}\psi(x_i^{-1})\delta_{ij}\;.  
\end{equation*}
Compared with \eqref{Mlike_pb_lin}, this shows that $\Psi:(M_g^0,\PB'{}^\psi)\to(M_g,\PB^\phi)$ is a birational
Poisson isomorphism when taking $\phi(x)=\psi(x^{-1})x^{g+2}$; notice that $\phi$, defined by this formula, is
indeed a polynomial of degree at most $g+1$, vanishing at $0$ and that we get all such polynomials $\phi$ for some
appropriate polynomial $\psi$ of degree at most $g+1$, vanishing at $0$.

\section{Appendix B: MAPLE code for Proposition \ref{bils9}} 

For completeness, below we have included MAPLE code (without output) which 
verifies the computer algebra required for the proof of Proposition  \ref{bils9}. 
For the reader interested in using the code, please see the 
repository {\tt https://github.com/anwh1729/Volterra_maps.git} 
where the original MAPLE file can be downloaded.  

\vspace{.2in}
 
\lstset{basicstyle=\ttfamily,breaklines=true,columns=flexible}
\pagestyle{empty}
\DefineParaStyle{Maple Bullet Item}
\DefineParaStyle{Maple Heading 1}
\DefineParaStyle{Maple Warning}
\DefineParaStyle{Maple Heading 4}
\DefineParaStyle{Maple Heading 2}
\DefineParaStyle{Maple Heading 3}
\DefineParaStyle{Maple Dash Item}
\DefineParaStyle{Maple Error}
\DefineParaStyle{Maple Title}
\DefineParaStyle{Maple Ordered List 1}
\DefineParaStyle{Maple Text Output}
\DefineParaStyle{Maple Ordered List 2}
\DefineParaStyle{Maple Ordered List 3}
\DefineParaStyle{Maple Normal}
\DefineParaStyle{Maple Ordered List 4}
\DefineParaStyle{Maple Ordered List 5}
\DefineCharStyle{Maple 2D Output}
\DefineCharStyle{Maple 2D Input}
\DefineCharStyle{Maple Maple Input}
\DefineCharStyle{Maple 2D Math}
\DefineCharStyle{Maple Hyperlink}
\mapleinput
{$ \displaystyle \texttt{>\,} \mathit{restart} \colon \,\mathit{with} (\mathit{LinearAlgebra})\colon \, $}

\mapleinput
{$ \displaystyle \texttt{>\,} \,\mathit{wrec} \coloneqq w [4]\cdot w [3]\cdot w [2]+w [2]\cdot w [1]\cdot w [0]+2\cdot w [2]^{2}\cdot (w [3]+w [1])+w [2]\cdot (w [3]^{2}+w [3]\cdot w [1]+w [1]^{2})+w [2]^{3}+\mathrm{nu}\cdot w [2]\cdot (w [3]+w [2]+w [1])+b \cdot w [2]+a ; $}

\mapleinput
{$ \displaystyle \texttt{>\,} \,\esnum \mathit{\,The\,(P.iv)\,equation\,(1.2),\,expressed\,as\,a\,relation\,between\,variables\,w[0],w[1],w[2],w[3],w[4].\,}  $}

\mapleinput
{$ \displaystyle \texttt{>\,} \,\esnum \mathit{\,Making\,the\,tau\,function\,substitution\,(2.3)\,to\,express\,(P.iv)\,as\,a\,homogeneous\,relation\,of\,order\,7\,for\,tau[n]:\,\,\,}  $}

\mapleinput
{$ \displaystyle \texttt{>\,} \,\,\mathit{tauseven} \coloneqq \mathit{wrec} \colon \,\mathrm{for}\,n \,\mathrm{from}\,0\,\mathrm{to}\,4\,\mathrm{do}\,\mathit{tauseven} \coloneqq \mathit{simplify} (\mathit{subs} (w [n]=\frac{\mathrm{tau}[n]\cdot \mathrm{tau}[n +3]}{\mathrm{tau}[n +1]\cdot \mathrm{tau}[n +2]},\mathit{tauseven}))\colon \,\mathrm{od}\colon \, $}

\mapleinput
{$ \displaystyle \texttt{>\,} \,\mathit{tauseven} \coloneqq \mathit{numer} (\mathit{tauseven}); $}

\mapleinput
{$ \displaystyle \texttt{>\,} \esnum \mathit{\,The\,above\,equation\,tauseven\,is\,the\,7th\,order\,degree\,8\,equation\,(2.4).\,}  $}

\mapleinput
{$ \displaystyle \texttt{>\,} \esnum \mathit{\,(2.4)\,is\,expressed\,here\,as\,a\,relation\,between\,tau[0],tau[1],tau[2],tau[3],tau[4],tau[5],tau[6],tau[7].\,}  $}

\mapleinput
{$ \displaystyle \texttt{>\,} \esnum \mathit{\,This\,directly\,verifies\,the\,proof\,of\,part\,(1)\,of\,Proposition\,2.1.\,}  $}

\mapleinput
{$ \displaystyle \texttt{>\,} \esnum \mathit{\,In\,order\,to\,derive\,and\,verify\,part\,(2),\,the\,above\,recurrence\,is\,iterated\,to\,generate\,14\,adjacent\,iterates\,\,}  $}

\mapleinput
{$ \displaystyle \texttt{>\,} \esnum \mathit{\,given\,in\,terms\,of\,the\,7\,initial\,values\,tau[0],tau[1],tau[2],tau[3],tau[4],tau[5],tau[6].\,}  $}

\mapleinput
{$ \displaystyle \texttt{>\,} \esnum \mathit{\,The\,iterates\,are\,denoted\,t[n]\,for\,brevity.\,}  $}

\mapleinput
{$ \displaystyle \texttt{>\,} \esnum  \mathit{t[0],t[1],t[2],t[3],t[4],t[5],t[6]\,are\,7\,arbitrary\,initial\,data\,for\,} ( 2.4). $}

\mapleinput
{$ \displaystyle \texttt{>\,} \mathit{topsol} \coloneqq \mathit{solve} (\mathit{tauseven} ,\mathrm{tau}[7])\colon  $}

\mapleinput
{$ \displaystyle \texttt{>\,} \mathrm{for}\,n \,\mathrm{from}\,7\,\mathrm{to}\,10\,\mathrm{do}\,t [n]\coloneqq \mathit{topsol} \colon \,\mathrm{for}\,m \,\mathrm{from}\,0\,\mathrm{to}\,6\,\mathrm{do}\,t [n]\coloneqq \mathit{factor} (\mathit{simplify} (\mathit{subs} (\mathrm{tau}[m]=t [m +n -7],t [n])))\colon \,\mathrm{od}\colon \,\mathit{print} (t [n])\colon \,\mathrm{od}\colon \, $}

\mapleinput
{$ \displaystyle \texttt{>\,} \esnum \mathit{\,Have\,generated\,t[7],t[8],t[9],t[10].\,Rather\,than\,generating\,up\,to\,t[13],\,which\,is\,slow\,\,}  $}

\mapleinput
{$ \displaystyle \texttt{>\,} \esnum \mathit{\,due\,to\,the\,increasing\,size\,of\,the\,expressions,\,apply\,the\,inverse\,map\,to\,generate\,t[-1],t[-2],t[-3].\,}  $}

\mapleinput
{$ \displaystyle \texttt{>\,} \esnum \mathit{\,This\,is\,faster\,and\,more\,efficient.\,}  $}

\mapleinput
{$ \displaystyle \texttt{>\,} \mathit{botsol} \coloneqq \mathit{solve} (\mathit{tauseven} ,\mathrm{tau}[0]); $}

\mapleinput
{$ \displaystyle \texttt{>\,} \mathrm{for}\,n \,\mathrm{from}\,1\,\mathrm{to}\,3\,\mathrm{do}\,t [-n]\coloneqq \mathit{botsol} \colon \,\mathrm{for}\,m \,\mathrm{from}\,0\,\mathrm{to}\,6\,\mathrm{do}\,t [-n]\coloneqq \mathit{factor} (\mathit{simplify} (\mathit{subs} (\mathrm{tau}[m +1]=t [m -n +1],t [-n])))\colon \,\mathrm{od}\colon \,\mathit{print} (t [-n])\colon \,\mathrm{od}\colon \, $}

\mapleinput
{$ \displaystyle \texttt{>\,} \,\esnum \mathit{\,Remark:\,Observe\,that\,these\,iterates\,are\,all\,Laurent\,polynomials,\,in\,accordance\,with\,Proposition\,2.3.\,}  $}

\mapleinput
{$ \displaystyle \texttt{>\,} \esnum \mathit{\,Now\,calculate\,determinant\,of\,matrix\,corresponding\,to\,5\,copies\,of\,Somos-9\,recurrence.\,}  $}

\mapleinput
{$ \displaystyle \texttt{>\,} f \coloneqq (i ,j)\rightarrow \,t [i +j -5]\cdot t [i -j +6]\colon \,M \coloneqq \mathit{Matrix} (5,f)\colon \,\mathit{simplify} (\mathit{Determinant} (M));\, $}

\mapleinput
{$ \displaystyle \texttt{>\,} \,\mathit{nullM} \coloneqq \mathit{NullSpace} (M);\, $}

\mapleinput
{$ \displaystyle \texttt{>\,} \,\mathit{nops} (\mathit{nullM});\,\mathit{kernelvec} \coloneqq \mathit{op} (1,\mathit{nullM}); $}

\mapleinput
{$ \displaystyle \texttt{>\,} \,\mathit{nullity} \coloneqq 5-\mathit{Rank} (M);\,\mathrm{for}\,j \,\mathrm{from}\,1\,\mathrm{to}\,5\,\mathrm{do}\,\mathit{al} [j]\coloneqq \mathit{kernelvec} [j]\,\mathrm{od}; $}

\mapleinput
{$ \displaystyle \texttt{>\,} \,\esnum \mathit{\,Entries\,of\,vector\,that\,spans\,the\,1-dimensional\,nullspace\,(kernel)\,of\,matrix\,M.\,}  $}

\mapleinput
{$ \displaystyle \texttt{>\,} \esnum \mathit{\,Checking\,that\,Somos-9\,relation\,is\,invariant\,under\,shifts:\,}  $}

\mapleinput
{$ \displaystyle \texttt{>\,} \esnum \mathit{\,Would\,like\,to\,normalize\,the\,entries\,of\,the\,vector\,in\,the\,kernel,\,and\,show\,they\,are\,invariant\,under\,shift.\,}  $}

\mapleinput
{$ \displaystyle \texttt{>\,} \esnum \mathit{\,It\,is\,convenient\,to\,first\,show\,they\,are\,functions\,of\,the\,initial\,data\,w[0],w[1],w[2],w[3],\,then\,check\,invariance.\,}  $}

\mapleinput
{$ \displaystyle \texttt{>\,} \mathrm{for}\,j \,\mathrm{from}\,1\,\mathrm{to}\,4\,\mathrm{do}\,\mathit{alw} [j]\coloneqq \mathit{al} [j]\colon \,\mathrm{for}\,m \,\mathrm{from}\,0\,\mathrm{to}\,3\,\mathrm{do}\,\mathit{alw} [j]\coloneqq \mathit{subs} (t [m]=\frac{w [m]\cdot t [m +1]\cdot t [m +2]}{t [m +3]},\mathit{alw} [j])\colon \,\mathrm{od}\colon \mathrm{od}\colon \, $}

\mapleinput
{$ \displaystyle \texttt{>\,} \mathrm{for}\,j \,\mathrm{from}\,1\,\mathrm{to}\,4\,\mathrm{do}\,\mathit{alw} [j]\coloneqq \mathit{simplify} (\mathit{alw} [j])\colon \,\mathit{print} (\mathit{alw} [j])\colon \,\mathrm{od}\colon \,k [1]\coloneqq \mathit{denom} (\mathit{alw} [4]);\, $}

\mapleinput
{$ \displaystyle \texttt{>\,} \esnum \mathit{\,Can\,recognize\,this\,polynomial\,as\,the\,first\,integral\,K[1]\,for\,the\,map\,(P.iv),\,given\,by\,(1.3).\,}  $}

\mapleinput
{$ \displaystyle \texttt{>\,} \mathit{alw} [5]\coloneqq 1\colon \,\mathrm{for}\,j \,\mathrm{from}\,1\,\mathrm{to}\,5\,\mathrm{do}\,\mathrm{alpha}[j]\coloneqq \mathit{simplify} (-k [1]\cdot \,\mathit{alw} [j]\cdot \mathit{denom} (\mathit{alw} [1]))\,\mathrm{od}\colon \, $}

\mapleinput
{$ \displaystyle \texttt{>\,} \esnum \mathit{\,Write\,all\,normalized\,coefficients\,as\,polynomials\,in\,the\,two\,invariants\,(first\,integrals)\,for\,the\,map\,(P.iv).\,}  $}

\mapleinput
{$ \displaystyle \texttt{>\,} k [2]\coloneqq \mathit{collect} (\mathit{collect} (\mathit{collect} (\mathit{simplify} ((\mathrm{alpha}[2]+k [1]^{2})\cdot a^{-1}),a),b),\mathrm{nu});\, $}

\mapleinput
{$ \displaystyle \texttt{>\,} \esnum \mathit{\,Recognize\,this\,as\,the\,first\,integral\,K[2]\,for\,the\,map\,(P.iv),\,given\,by\,(1.4).\,}  $}

\mapleinput
{$ \displaystyle \texttt{>\,} \mathit{simplify} (\mathrm{alpha}[1]-k [1]);\,\mathit{simplify} (\mathrm{alpha}[2]-a \cdot k [2]+k [1]^{2});\,\mathit{simplify} (\mathrm{alpha}[3]-a \cdot (a \cdot k [2]-2\cdot \,k [1]^{2}));\,\mathit{simplify} (\mathrm{alpha}[4]-a \cdot (a^{2}\cdot k [1]+b \cdot k [1]^{2}+\mathrm{nu}\cdot k [1]\cdot k [2]+k [2]^{2}));\,\mathit{simplify} (\mathrm{alpha}[5]+k [1]\cdot (a^{2}\cdot k [1]+b \cdot k [1]^{2}+\mathrm{nu}\cdot k [1]\cdot k [2]+k [2]^{2}));\, $}

\mapleinput
{$ \displaystyle \texttt{>\,} \,\esnum \mathit{\,This\,completes\,the\,proof\,of\,part\,(2)\,of\,Proposition\,2.1.\,}  $}

\bigskip

\noindent \textbf{Acknowledgments:} The research of ANWH was supported by Fellowship EP/M004333/1 from the
Engineering \& Physical Sciences Research Council, UK and grant IEC\textbackslash R3\textbackslash 193024 from the
Royal Society; he is also grateful to the School of Mathematics and Statistics, University of New South Wales, for
hosting him during 2017-2019 as a Visiting Professorial Fellow with funding from the Distinguished Researcher
Visitor Scheme, and to Wolfgang Schief, who provided additional support during his time in Sydney. 
We would like to thank the two anonymous referees for numerous detailed comments, which have helped to clarify the 
paper in several places.

\end{document}